\author{\Large Jesse Hemerik
\\
 Erasmus University Rotterdam, The Netherlands\\
e-mail: hemerik@ese.eur.nl}

\date{\Large May 2026}

\documentclass[11pt]{article}

\usepackage{amsmath}
\usepackage{amsfonts}
\usepackage{amssymb,mathrsfs}
\usepackage{bm}
\usepackage{bbm}
\usepackage{amsthm}
\usepackage{comment}
\usepackage[nottoc]{tocbibind}
\usepackage{graphicx}

\usepackage{subcaption}

\usepackage{changepage}
\usepackage{natbib}
\usepackage{abstract}
\usepackage[colorlinks=true,citecolor=blue,runcolor=blue,linkcolor=blue, pdfborder={0 0 0}]{hyperref}

\usepackage{tikz}
\usetikzlibrary{positioning, arrows.meta, calc} 
\usepackage{pgfplots}
\pgfplotsset{compat=1.17}

\usepackage{algorithm}
\usepackage[noend]{algorithmic}

\usepackage[margin=4cm]{geometry}

\theoremstyle{plain}

\newtheorem{theorem}{Theorem}[section]
\newtheorem{lemma}[theorem]{Lemma}
\newtheorem{proposition}[theorem]{Proposition}

\newtheorem{fact}[theorem]{Fact}

\theoremstyle{definition}
\newtheorem{definition}{Definition}[section]
\newtheorem{assumption}{Assumption}[section]

\newtheorem{remark}{Remark}[section]
\newtheorem{example}{Example}[section]

\newcommand{\reals}{\mathbb{R}}

\newcommand{\N}{\mathcal{N}}
\newcommand{\F}{\mathcal{F}}
\newcommand{\R}{\mathcal{R}}
\newcommand{\K}{\mathcal{K}}

\newcommand{\I}{\mathcal{I}}
\newcommand{\J}{\mathcal{J}}
\newcommand{\G}{\mathcal{G}}
\newcommand{\X}{\mathcal{X}}
\newcommand{\A}{\mathcal{A}}
\newcommand{\C}{\mathcal{C}}
\newcommand{\B}{\mathcal{B}}
\newcommand{\E}{\mathcal{E}}
\newcommand{\Hy}{\mathcal{H}}

\newcommand{\pr}{\mathbb{P}}

\newcommand{\Sigm}{\mathbf{\Sigma}}
\newcommand{\subs}{\mathcal{C}}
\newcommand{\ssp}{\mathcal{S}}

\newcommand{\citt}{\citet}
\newcommand{\citp}{\citep}

\excludecomment{sols}

\excludecomment{commentj}

\newcommand{\indep}{\mathrel{\perp\!\!\!\perp}}

\DeclareMathOperator{\Var}{Var}
\DeclareMathOperator{\Cov}{Cov}

\begin{document}

\begin{titlepage}

\centering

\vspace*{2cm}

{\huge\bfseries Multiple Testing\par}

\vspace{2cm}

{\Large Jesse Hemerik\par}

\vspace{0.3cm}

{\large Econometric Institute, Erasmus University Rotterdam, the Netherlands\par
\texttt{hemerik@ese.eur.nl}\par}

\vspace{1cm}

{\large \today\par}

\vspace{2cm}

\begin{minipage}{0.8\textwidth}
\small
\begin{center}
\textbf{Abstract}
\end{center}
This text provides an introduction to multiple hypothesis testing. 
It covers various error criteria and testing procedures, and includes references to relevant R packages.
An earlier version of this text served as the lecture notes for a PhD-level course on multiple testing.
\end{minipage}

\vfill

\end{titlepage}

\pagenumbering{roman}
\clearpage
\tableofcontents

\newpage

\section*{Mathematical notation}
\addcontentsline{toc}{section}{Mathematical notation}


\begin{tabular}{ll}
\textbf{Notation} & \textbf{Meaning} \\

$P_i$ & p-value corresponding to hypothesis $\Hy_i$ \\
$P_{(i)}$ & $i$-th sorted p-value\\
$Q_i$ & p-value corresponding to the $i$-th true hypothesis \\
$\N$ & the set of indices of true hypotheses\\
$\R$ & the set of indices of rejected hypotheses\\
$V$ &  the number of false positives, $|\N\cap \R|$\\
$\subs$ & the set of all nonempty subsets of $\{1,...,m\}$\\
$\Hy_{\I}$ &  intersection hypothesis corresponding to the indices in $\I$.\\
$\X$ &  set of all $\I\in \C$ for which $\Hy_{\I}$ is rejected by the CTP\\
$\reals$ &  set of real numbers, $(-\infty,\infty)$ \\
$\log(a)$ &  natural logarithm of  $a$\\
$\mathbbm{1}$ & indicator function\\
$a\wedge b$ & minimum of numbers  $a$ and $b$\\
$a\vee b$ & maximum of numbers  $a$ and $b$\\ 
$\lceil a \rceil$ & smallest integer that is at least $a$\\ 
$\lfloor a \rfloor$ & largest integer that is at most $a$\\ 
$|a|$ & absolute value of $a$ (if $a$ is a number)\\
$|\mathcal{A}|$ & cardinality of $\mathcal{A}$  (if $\mathcal{A}$ is a set)\\
$\mathcal{A}\cup \mathcal{B}$ & union of sets $\mathcal{A}$ and $\mathcal{B}$\\
$\mathcal{A}\cap \mathcal{B}$ & intersection of sets $\mathcal{A}$ and $\mathcal{B}$ \\
$\mathcal{A}\setminus \mathcal{B}$ & set $\mathcal{A}$ minus set $\mathcal{B}$ \\
$\mathcal{A}^c$ &$\{1,....,m\}\setminus\mathcal{A}$ (if $\mathcal{A}\subseteq \{1,...,m\}$) \\
$\mathcal{A}\subseteq \mathcal{B}$ & $\mathcal{A}$ is a subset of $\mathcal{B}$\\
$\emptyset$ &  empty set \\
$\bigcup_{...}...$ & union over ... of ... \\
$\bigcap_{...}...$ & intersection over ... of ... \\
$\exists$ & there exists\\
$\forall$ & for all\\
\end{tabular}

\newpage
\section*{Abbreviations}
\addcontentsline{toc}{section}{Abbreviations}
\begin{tabular}{ll}
\textbf{Notation} & \textbf{Meaning} \\
i.i.d. & independent and identically distributed \\
CDF  & cumulative distribution function\\
FWER & familywise error rate\\
k-FWER & k-familywise error rate\\
FDX & false discovery exceedance\\
FDP & false discovery proportion\\
FDR & false discovery rate\\
MTP & multiple testing procedure\\
CTP & closed testing procedure\\

\end{tabular}
\newpage
\pagenumbering{arabic}

\section{Introduction}
When we test a single hypothesis,  we are interested in the probability of committing a type I error, as well as the power. When we test multiple hypotheses,  it likewise makes sense to try  to know the risks of incurring false positives. At the same time, we want  to have a good chance of rejecting hypotheses that are false. These are the aims of  multiple hypothesis testing procedures (MTPs).

There are many MTPs. One reason is that there a various kinds of demands that we could make regarding type I error control. For example, some methods ensure that the probability of incurring any false positives is small. This is called \emph{familywise error rate} (FWER) control. Other methods are more lenient. Another reason is that there are different ways to measure power, and given some definition of power, which method is most powerful may depend on the situation. Further, different methods require different assumptions. Finally, knowledge on multiple testing has improved throughout the years, and sometimes an older method is ``uniformly improved'' by a new method. 

This text covers some of the most important and widely used MTPs. We start off with a section (\S\ref{secsingle}) about testing a single hypothesis, focusing on ideas that will return in later sections. \S\ref{secmtbasics} covers some basic building blocks of multiple testing theory, which will return in later sections as well. For example, in \S\ref{secmtbasics} we cover various \emph{global tests}, which will be important in e.g. \S\ref{secCT}. In that section we discuss \emph{closed testing} for FWER control. This is a fundamental theory that can be used for designing and understanding FWER controlling methods. \S\ref{secresamplingFWER} covers powerful resampling-based methods for FWER control. This section extends  resampling-based tests from \S\ref{secsingle} to multiple testing settings. \S\ref{secFDX} discusses \emph{false discovery exceedance control}, which is  less conservative than FWER control. The Romano-Wolf method discussed in  \S\ref{secFDX} relies on so-called k-FWER methods, which are covered in earlier sections.

\S\ref{secCTPfdp} introduces an approach to multiple testing that is quite different from those in other sections. Traditional MTPs tell the user which hypotheses to reject, and the user has limited freedom to investigate other hypotheses afterwards using the same data. However, the methods in \S\ref{secCTPfdp}, which connect to \S\ref{secCT}, provide information about all possible subsets of hypotheses simultaneously: for each subset of hypotheses, we obtain a lower bound for the number of false hypotheses in the set. With large probability, all these bounds are simultaneously correct, so that the user can freely explore various subsets of hypotheses in a post hoc manner.

\S\ref{secFDR} and \S\ref{secknockoffs}  are about \emph{false discovery rate} (FDR) control, which means ensuring that the expected proportion of false positives is small. The theory in these sections is quite different from that of \S\S\ref{secCT}-\ref{secCTPfdp}, so that \S\ref{secFDR} and \S\ref{secknockoffs} can mostly be read independently from those sections. Two proofs in \S\ref{secFDR} and \S\ref{secknockoffs} rely on some limited martingale theory, which is covered in the appendices.

The end of each section offers various exercises that help improve and expand knowledge of the topics. Very nonstandard exercises are indicated with a star. 

There are many relevant methods which this text does not cover. For example, there is a huge literature on FDR controlling procedures. This text only covers a few such methods, which are relatively prominent. Other relevant topics that are not covered are e.g. graph-based procedures \citp{bogomolov2021hypotheses}, Bayesian approaches \citp{storey2002direct}, e-value-based approaches \citp{ramdas2025hypothesis} and selective inference. Regarding the topics that \emph{are} covered, only a selection of the available theory is provided, and some theory is only developed in exercises.

\newpage

\section{Testing a single hypothesis} \label{secsingle}
This section  recalls core ideas on hypothesis testing. Nearly all ideas introduced here will return in later sections.
Moreover, this section  defines the meanings of various terms, whose interpretations may differ slightly  between textbooks.

\subsection{Hypotheses and p-values} 
A statistical model is a set of distributions in which we assume that the true distribution of our data lies. More formally, a model is a set of probability measures (see \S\ref{apprv}). For example, the model may be that the data $X$ is an $n$-dimensional vector with i.i.d. normally distributed variables,  with unknown mean $\mu\in \reals$ and variance $\sigma^2>0$. Often we are interested in testing whether the true distribution of our data lies in some particular set of distributions, which is a subset of our model. More precisely, we want to see whether we can statistically prove that the true distribution does not lie in this subset. This subset of our model is then called the \emph{null hypothesis} $\Hy_0$.
For example, we may be interested in testing whether $\mu=0$, i.e., we want to see whether we can prove that $\mu\neq0$. The null hypothesis is then the set of all distributions in our model for which $\mu=0$. The null hypothesis could also be $\Hy_0: \mu\in [-0.1,0.1]$. This is an example of a \emph{composite null hypothesis}, because $\Hy_0$ corresponds to multiple parameter values. Even if the null hypothesis is $\Hy_0: \mu=0$, we will call it composite, since $\Hy_0$ contains multiple distributions, namely all $n$-vectors of i.i.d. normal variables with mean 0 and positive variance. A hypothesis that contains a single distribution is called \emph{simple}.

It can be useful to formulate an alternative hypothesis $\Hy_{\text{a}}$ as well, which is a subset of our model that does not intersect with the null hypothesis. $\Hy_{\text{a}}$ is the set of distributions under which  we wish to have a high probability of rejecting $\Hy_0$. $\Hy_{\text{a}}$ may influence which test statistic we use (see below). Specifying an alternative hypothesis can be relevant  when considering power properties of a test, i.e., when  studying the probability of rejecting $\Hy_0$ given that the data follow $\Hy_{\text{a}}$.

Hypothesis testing is related to the concept of a \emph{p-value}.
We will call a random variable $P$ a \emph{valid} p-value if it satisfies the following property  under every distribution in $\Hy_0$:
\begin{equation} \label{eqlunif}
\pr(P\leq c)\leq c \qquad \text{for all }c\in[0,1].
\end{equation}
This property means that $P$  is stochastically larger than standard uniform, meaning its CDF is smaller than (or equal to) the CDF of the $U[0,1]$ distribution. 

In order to compute a p-value, we typically compute a test statistic $T=T(X)$, which is a function of the data $X$, say. 
We try  to define $T$ in such a way that it is likely to reject (see below) $\Hy_0$ when $\Hy_0$ is false, or similarly, when a particular alternative hypothesis is true.
 A test statistic is  a random variable. When we consider a specific realized value of the test statistic, we will write $T_{\text{obs}}$ here.  In case $\Hy_0$ is a \emph{simple hypothesis} and large values of $T_{\text{obs}}$ are evidence against $\Hy_0$, we often compute a p-value as follows:
$$p = \pr_{\Hy_0}(T\geq T_{\text{obs}}).$$
Here $\pr_{\Hy_0}(...)$ means the probability of ``...'' under $\Hy_0$, i.e., if $\Hy_0$ is true.

If we are testing a hypothesis $\Hy_0:\theta=\theta_0$ about a one-dimensional parameter $\theta\in \reals$ and we want to have power against left-sided alternatives, then usually, very low values of $T_{\text{obs}}$ are evidence against $\Hy_0$, so that we compute the p-value as
$$p = \pr_{\Hy_0}(T\leq T_{\text{obs}}).$$ If $\theta\in \reals$ and we want power in both directions, we can compute a left-sided and a right-sided p-value and take $p$ to be 2 times their minimum.
We wrote ``$p$'' as a  lowercase letter, because we considered its realized value. Of course a p-value is a random variable and we will often write ``$P$'' to emphasize that.

When $\Hy_0$ is a \emph{composite hypothesis} (see above), we often compute the p-value as follows:
$$p = \sup_{\pr\in\Hy_0}\pr(T\geq T_{\text{obs}}).$$
To understand this notation, realize that $\Hy_0$ is a set of distributions, i.e., it is a set of probability measures $\pr$ on the sample space (see \S\ref{apprv}). By taking the supremum, we ensure that $p$ is larger than 
$\pr(T\geq T_{\text{obs}})$ for every $\pr\in\Hy_0$.

If $\Hy_0$ is true and $P$ is computed in one of the ways indicated, then \eqref{eqlunif} is satisfied.
When $\Hy_0$ is simple and the null distribution of $T$ is continuous and $P$ is computed in one of the above-mentioned ways, then $P$ is exactly standard uniform on $[0,1]$ under $\Hy_0$. It is equally valid to say that $P$ is exactly standard uniform on $(0,1]$ or $(0,1)$, since a standard uniform variable is strictly larger than 0 and smaller than 1 \emph{almost surely} (see Appendix \ref{apprv}).
Note that if $P$ is valid, then $\pr_{\Hy_0}(P=0)=0$. This means that under $\Hy_0$, a valid p-value cannot be 0 in practice.


\subsection{Hypothesis testing}

One typically \emph{rejects} $\Hy_0$ when $P\leq \alpha$, where $\alpha\in(0,1)$ is some prespecified significance threshold. This ensures that the rejection probability is at most $\alpha$ if $\Hy_0$ is true. Note that if $p$ is computed as $\pr_{\Hy_0}(T\geq T_{\text{obs}})$, then it is equivalent to  rejecting when $T$ exceeds its $(1-\alpha)$-quantile. When $\Hy_0$ is rejected, it is common to say that the test or p-value is ``significant''. Otherwise we do not reject (so we ``retain'' or  ``accept'' $\Hy_0$). That does not mean that we conclude that $\Hy_0$ is true; when we do not reject $\Hy_0$, we simply do not make a claim at all. Since we can reject or not, and $\Hy_0$ can be true or false, there are four possible outcomes of a hypothesis test, see Table \ref{tabtestoutc}.

\begin{table}[h]
\centering
\begin{tabular}{ c c c } 
\hline
 & \textbf{$\Hy_0$ true} & \textbf{$\Hy_0$ false} \\
\hline
\textbf{Reject } 
& False Positive (Type I error) 
& True Positive (Correct rejection) \\
\hline
\textbf{Don't reject} 
& True Negative  
& False Negative (Type II error) \\
\hline
\end{tabular}
\caption{Possible outcomes of a hypothesis test}
\label{tabtestoutc}
\end{table}

A true positive is often called a ``true discovery'' and a false positive is sometimes called a ``false discovery''.
Of course, this is because scientific discoveries are often based on rejecting a null hypothesis. For example, if the parameter $\theta\in \reals$ indicates the effect of some genetic mutation on a disease or the effect of some governmental policy, then rejecting $\Hy_0:\theta=0$ means ``discovering'' that there is some effect.

If we consider some distribution in $\Hy_0$, then we may consider the \emph{type I error rate}, which is also called the  \emph{level} and means the probability of rejecting $\Hy_0$, when the data follow that null distribution.
If we consider some distribution in the alternative hypothesis $\Hy_{\text{a}}$, then the \emph{power} is the probability of rejecting $\Hy_0$ if the data follow that distribution.
The \emph{type II error rate} is then $1$ minus the power.

If we are testing a hypothesis $\Hy_0: \theta=\theta_0$ about a one-dimensional parameter $\theta\in \reals$, then often we desire to have power in both directions, as mentioned above.  Different approaches are then possible, after we have computed a  statistic $T$ that is 0-centered under $\Hy_0$:
\begin{enumerate}
\item compute a valid left-sided and  right-sided p-value and take $p$ to be 2 times their minimum. This gives a valid p-value for $\Hy_0$ and we can reject if $p\leq \alpha$;
\item reject if $T$ lies below its $\alpha/2$ or above its $(1-\alpha/2)$-quantile;
\item compute the absolute value $|T|$ and compare it  with its $(1-\alpha)$-quantile.
\end{enumerate}
If the test statistic is symmetric about 0 under $\Hy_0$, then these approaches are equivalent. If the test statistic is not symmetric, then the third approach is not equivalent to the first two. 

The value $\alpha$ usually indicates the desired or maximally allowed type I error rate, specified by the user. It is sometimes called the \emph{nominal} type I error rate. When the actual type I error rate is (potentially) larger than $\alpha$, then we call the test \emph{liberal} or \emph{anti-conservative}. 
We say that a test controls the type I error rate (or is ``valid'') if $\pr(\text{reject})\leq \alpha$ (for every distribution $\pr$ in the null hypothesis, or for some specified set of null distributions). When the test  has a type I error rate strictly below $\alpha$ (for every $\pr\in \Hy_0$ or for some specified set of distributions in $\Hy_0$), we call the test \emph{conservative}. 
The terms \emph{size} and \emph{level} are related and sometimes used interchangeably.
We define the \emph{size} of a test to be $\sup_{\pr\in\Hy_0}\pr(\text{reject})$. 
A test is called \emph{exact}  when $\pr(\text{reject})=\alpha$ (under some specified set of null distributions). The definitions of these terms can differ slightly between textbooks.

A concept closely related to hypothesis tests, are confidence intervals. If we consider a one-dimensional parameter $\theta$ and a test for hypotheses of the form $\Hy_0:\theta=\theta_0$, then a confidence interval for  $\theta$ can be constructed as the set of all $\theta_0$ for which the corresponding hypothesis test does not reject. This is called \emph{inverting} a test.

Note that when a null hypothesis $\Hy_0$ is rejected (a ``discovery''), this often leads researchers to publish the result (if the estimated effect is considered substantial). On the other hand, if $\Hy_0$ is not rejected, then this usually not a reason to publish the result. The result may still be published, but even then it will usually receive less attention than  rejected hypotheses. This is one of the reasons why type I errors are often considered to be more problematic than type II errors.
 Also note that a type II error is in a sense not an error at all, because we do not make any claim when we fail to reject $\Hy_0$.

In recent years, many articles have been published that warn against an overly dichotomous approach to hypothesis testing, where a finding is treated as important when $p\leq \alpha$ and unimportant otherwise. In line with this, some researchers argue that ``significance testing'' should be abandoned in their field.
They often  propose to instead look at \emph{how} small the p-value is and to look at the size of the estimated effect  --- or at the confidence interval. For example, when the sample size is large, then the p-value may be very small, but this does not mean that the estimated effect is large; the estimated effect may be very small, so that it is not very relevant. Note that some people have not only suggested to ban significance testing, but to even stop reporting p-values. However,  to many statisticians that goes too far, since a p-value provides valuable information in addition to an effect estimate. Indeed, a large estimated effect is highly dubious when the sample size is too small. Note that confidence intervals are also highly useful in this context, although they can be less practical when the parameter of interest is multi-dimensional.

\subsection{Examples of parametric tests}

A \emph{parametric statistical model} is a model in which finitely many parameters fully describe the possible distributions. An example is the model mentioned earlier, where we assume that the data $(X_1,...,X_n)$ consists of i.i.d. $N(\mu,\sigma^2)$ variables. That model consists of infinitely many distributions, but each distribution is fully specified by the two parameters $\mu\in \reals$ and $\sigma>0$.

If we only assume that $(X_1,...,X_n)$ consists of i.i.d. variables, and make no assumption about the distributional shape of the variables, then our model is nonparametric. Such a model is useful if we are not comfortable with assuming that the observations are e.g. normally or exponentially distributed. Another advantage of nonparametric approaches is that they can be very powerful in a multiple testing context.  Nonparametric models will be considered later. We first discuss some simple examples of parametric models.


\subsubsection{Example 1: Testing for Correlation} \label{secexcortest}

Suppose we observe i.i.d. pairs $(X_1,Y_1),\ldots,(X_n,Y_n)$, where the variables $(X_i,Y_i)$ are assumed to be jointly normally distributed. 
Note that this model is fully  parametric, because if we know the mean vector and covariance matrix of $(X_i,Y_i)$, we know the distribution of the data.

A common question is whether  the two variables are correlated.
Let $\rho$ denote the true Pearson  correlation between $X$ and $Y$. We wish to test
\[
H_0: \rho = 0
\qquad \text{versus} \qquad
H_{\text{a}}: \rho \neq 0.
\]


The test statistic is based on the sample Pearson correlation coefficient
\[
r = \frac{\sum_{i=1}^n (X_i - \bar X)(Y_i - \bar Y)}
{\sqrt{\sum_{i=1}^n (X_i - \bar X)^2}\sqrt{\sum_{i=1}^n (Y_i - \bar Y)^2}}.
\]
If $\Hy_0$ is true, the statistic
\[
T  = r \sqrt{\frac{n-2}{1-r^2}}
\]
follows a $t$-distribution with $n-2$ degrees of freedom \citp[][p.218]{lehmann2022testing}.


The p-value is obtained by comparing the observed value of $T$ to this reference distribution. 
We can for example obtain the two-sided p-value by computing 2 times the minimum of the left- and the right-sided p-value.
If the model is correct, then $P$ is exactly standard uniformly distributed under $\Hy_0$.  This will usually not be the case if the  multivariate normality assumption fails. However, when $n$ is large, the test is still roughly exact, even when that assumption fails.

In Table \ref{tab:cars} two columns of some  rows of the dataset \emph{mtcars} from the \emph{dataset} package in R are shown. This dataset contains characteristics of various car models; here we look at the continuous variables \emph{mpg} (miles per US gallon) and \emph{weight} (in 1000 pounds). Applying the correlation test described above (using the R function \emph{cor.test()}) leads to the p-value  $p= 0.2412$.\footnote{If we compute the two-sided p-value for the hypothesis that $\beta_1=0$ in a simple Gaussian regression model (using \emph{lm()} in R), then we will obtain the same p-value as above. This is related to the fact that if the $(X_i,Y_i)$ are multivariate normal and $\rho=0$, then conditional on $X_1,...,X_n$, the slope $\beta_1$ in the regression model is 0. 
Note however that the traditional regression model does not assume that the  $(X_i,Y_i)$ are multivariate normal --- in fact, it does not require the  $X_i$ to be stochastic. A consequence is that for the traditional regression model, $\beta_1\neq0$ does not necessarily mean that $\rho\neq0$. } The empirical Pearson correlation is $-0.6437$.



\begin{table}[h]
\centering
\caption{Dataset with two continuous variables}
\label{tab:cars}
\begin{tabular}{ccc}
\hline
car model & mpg & weight \\
\hline
Mazda RX4 &21.0  & 2.62 \\
Datsun 710 &22.8 & 2.32 \\
Valiant      & 18.1   & 3.46 \\
Duster 360 &14.3 & 3.57 \\
Merc 240D  &24.4 & 3.19 \\
\hline
\end{tabular}
\end{table}

\subsubsection{Example 2: Non-inferiority Testing} \label{secnoninf}

Suppose we compare a new treatment to a standard treatment. For example, we may want to compare the effect of a new policy with the effect of the standard policy, or the effect of a new drug with the effect of the standard drug. Let $\mu_N$ and $\mu_S$ denote the mean outcomes under the new and standard treatments, respectively. Suppose higher values are  better.

We are willing to accept the new treatment as adequate if it is not worse than the standard treatment by more than a pre-specified margin $\Delta > 0$. This leads to the hypotheses
\[
\Hy_0: \mu_N - \mu_S \le -\Delta
\qquad \text{versus} \qquad
\Hy_1: \mu_N - \mu_S > -\Delta.
\]

Here, the null hypothesis represents the case where the new treatment is inferior by more than the acceptable margin. Rejecting $\Hy_0$ provides evidence that the new treatment is \emph{non-inferior} to the standard treatment.


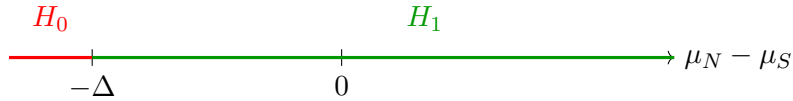
\begin{figure}[h]
\centering
\begin{tikzpicture}[scale=1.1]

\draw[->] (-4,0) -- (4,0) node[right] {$\mu_N - \mu_S$};

\draw (-3,0.1) -- (-3,-0.1) node[below] {$-\Delta$};
\draw (0,0.1) -- (0,-0.1) node[below] {$0$};

\draw[very thick, red] (-4,0) -- (-3,0);
\draw[very thick, green!60!black] (-3,0) -- (4,0);

\node[red, above] at (-3.5,0.2) {$H_0$};
\node[green!60!black, above] at (1,0.2) {$H_1$};

\end{tikzpicture}
\caption{Non-inferiority testing. The null hypothesis corresponds to effects worse than $-\Delta$.}
\label{fig:noninferiority}
\end{figure}

Note that $\Hy_0$ is a composite null hypothesis, since it corresponds to all values of $\mu_N-\mu_S$ below $-\Delta$. Assume that the differences 
$$D_1=X^N_1-X^S_1,\ldots,D_n=X^N_n-X^S_n$$ 
are i.i.d. and normally distributed with mean $\mu_N-\mu_S$.

To test $\Hy_0$, we compute the t-statistic

$$T=\frac{\bar{D}-(-\Delta)}{\hat{\sigma}/\sqrt{n}},$$
where $$\hat{\sigma} = \sqrt{\frac{1}{n-1} \sum_{i=1}^n( D_i - \bar{D})^2}.$$

The p-value is 
$$\sup_{\pr\in\Hy_0}\{\pr(T\geq T_{\text{obs}})\}=$$
$$\pr_{\mu_N-\mu_S=-\Delta}(T\geq T_{\text{obs}}).$$
This probability can be computed by using that if $\mu_N-\mu_S=-\Delta$, then $T$ has a t-distribution with $n-1$ degrees of freedom.

\subsubsection{Example 3: Equivalence Testing} \label{secequiv}

In equivalence testing, the goal is to show that two variables have similar means, up to a small and practically irrelevant difference. For example, in the field of food safety testing, the aim may be to show that the concentration of a certain molecule in a new variety of potato is similar to the concentration in an existing variety.

Let $\mu_A$ and $\mu_B$ denote the means of the two variables, and let $\Delta > 0$ be a pre-specified equivalence margin. The hypotheses are
\[
\Hy_0: |\mu_A - \mu_B| \ge \Delta
\qquad \text{versus} \qquad
\Hy_1: |\mu_A - \mu_B| < \Delta.
\]

Equivalently, the null hypothesis can be written as 
$$\Hy_0 = \Hy_0^-\cup \Hy_0^+,$$
 where 
\[
\Hy_0^-: \mu_A - \mu_B \le -\Delta \;\;, \qquad \;\; \Hy_0^+: \mu_A - \mu_B \ge \Delta.
\]

Rejecting $\Hy_0$ implies that the difference between the means is small enough to be considered practically negligible, and we conclude that the treatments are \emph{equivalent}. In practice, equivalence testing is often implemented using Two One-Sided Tests (TOST). Thus, we test $\Hy_0^-$ and $\Hy_0^+$ separately, e.g. using the test from\S\ref{secnoninf}. We reject $\Hy_0$ if and only if both $\Hy_0^-$ and $\Hy_0^+$ are rejected. Note that this is a simple example of \emph{multiple hypothesis testing}, since we test two hypotheses, $\Hy_0^-$ and $\Hy_0^+$. However, we do not need a \emph{multiple testing correction} here. The reason is that at most one  of the hypotheses $\Hy_0^-$ and $\Hy_0^+$ can be true, so that the probability of rejecting a true hypothesis stays below $\alpha$ with the TOST procedure (see also Example \ref{exampTOST}).

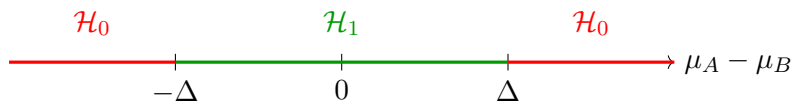
\begin{figure}[h]
\centering
\begin{tikzpicture}[scale=1.1]

\draw[->] (-4,0) -- (4,0) node[right] {$\mu_A - \mu_B$};

\draw (-2,0.1) -- (-2,-0.1) node[below] {$-\Delta$};
\draw (2,0.1) -- (2,-0.1) node[below] {$\Delta$};
\draw (0,0.1) -- (0,-0.1) node[below] {$0$};

\draw[very thick, red] (-4,0) -- (-2,0);
\draw[very thick, green!60!black] (-2,0) -- (2,0);
\draw[very thick, red] (2,0) -- (4,0);

\node[green!60!black, above] at (0,0.2) {$\Hy_1$};
\node[red, above] at (-3,0.2) {$\Hy_0$};
\node[red, above] at (3,0.2) {$\Hy_0$};

\end{tikzpicture}
\caption{Equivalence testing. ``Equivalence'' means that that $|\mu_A-\mu_B|\leq\Delta$.}
\label{fig:equivalence}
\end{figure}

\subsection{Permutation tests and related nonparametric tests}

In \S\ref{secintronptests} we indicate what nonparametric tests are and why they are useful. In the subsequent sections we give some examples of permutation tests and related tests, to give an impression of how these tests work. Then we provide somewhat general theory on permutation tests. Finally, we discuss bootstrap tests.

\subsubsection{Nonparametric models and tests} \label{secintronptests}
Earlier we mentioned the distinction between parametric and nonparametric models: parametric models involve a finite number of parameters, that completely define the possible distributions under this model. An example is a negative binomial model, which is completely define by the coefficients and the overdispersion parameter.

Nonparametric models are models that are not defined in terms of any parameters. For example, the model may be that we have i.i.d. observations $X_1,...,X_n$, with no known distributional shape.

Semiparametric models are partly defined by parameters, but not completely. An example is a model that is the same as a classical Gaussian linear model, except that we do not assume that the residuals are normally distributed or have any other known shape. That model is sometimes called the \emph{general} linear model --- not be confused with a \emph{generalized} linear model.

A nonparametric test is a test that can be applied when consider a nonparametric model. An example of a nonparametric test is a permutation test. A permutation test can also be semiparametric, in the sense that a semiparametric model is assumed. For example, there are permutation tests for testing hypotheses about coefficients in the non-Gaussian, general linear model mentioned above.

One advantage of nonparametric tests such as permutation tests is well-known: they require fewer distributional assumptions than parametric tests in order to provide reliable inference. In the context of multiple testing however, there is another major potential advantage: permutation-based (and related) multiple testing methods can be much more powerful than other multiple testing methods. The reason is that they can adapt to the dependence structure in the data, as will be discussed later.

\subsubsection{Example 1: a permutation test} \label{secexptest}
We again consider the dataset from Table  \ref{tab:cars}, which is a subset of the dataset  \emph{mtcars}. Suppose we wish to know whether the variables \emph{mpg} ($X$, say) and \emph{weight} ($Y$, say) are independent of each other, so $\Hy_{\text{indep}}: X \indep Y$. If we proceed as in \S\ref{secexcortest} and assume that the pairs $(X_1,Y_1),\ldots(X_5,Y_5)$ are i.i.d. and  multivariate Gaussian, then we can obtain a valid p-value for $\Hy_0$ by using the parametric correlation test. Indeed, if those assumptions are valid and $\Hy_{\text{indep}}$ is true, then the zero-correlation hypothesis $\Hy_0$ from 
\S\ref{secexcortest} is also true. Thus the probability of rejecting $\Hy_0$ will be at most $\alpha$ and likewise
$\pr(P\leq c)\leq c$ for every $c\in[0,1]$. 

However, suppose that the model assumed in \S\ref{secexcortest}  is not valid. More precisely, suppose that $(X_1,Y_1),\ldots(X_5,Y_5)$ are i.i.d. but not multivariate Gaussian. Then the properties of the parametric correlation test break down. We can then instead perform a \emph{permutation test}, which does not require parametric assumptions such as normality. Besides requiring fewer or no paramatric assumptions, permutation or bootstrap methods also have an additional advantage when we are testing \emph{multiple} hypotheses; namely, permutation or bootstrap methods then often have superior power (see e.g. \S\ref{secresamplingFWER}).

All we need for the permutation test to control the type I error rate, is that under $\Hy_{\text{indep}}$, the observations $Y_1,...,Y_n$ are i.i.d. conditional on $X=(X_1,...,X_n)$ (or the other way around).
The permutation test proceeds as follows. First it computes a test statistic based on the original data. A typical choice is to use the empirical correlation $r=r(X,Y)$, which is $-0.6437$ for these data.
Note that there are $n!=5!$ ways to shuffle (permute) the values $Y=(Y_1,...,Y_n)$, if we include the original ordering. The test permutes $Y$ in all possible ways, thus creating $5!-1$ ``fake'' datasets. For each of these permuted datasets, the test recomputes the test statistic $r(X,\pi(Y))$, where $\pi(Y)$ is a permuted version of $Y$ and $r(X,\pi(Y))$ is the corresponding empirical Pearson correlation. A histogram of these  test statistics is shown in Figure \ref{fig:histptest}. This is a reference distribution --- often termed a \emph{permutation distribution} --- with which we compare our original test statistic $-0.6437$. This original test statistic is somewhat in the left tail of the reference distribution, but not very extreme.

\begin{figure}[ht!]
    \centering
    \includegraphics[width=0.85\textwidth]{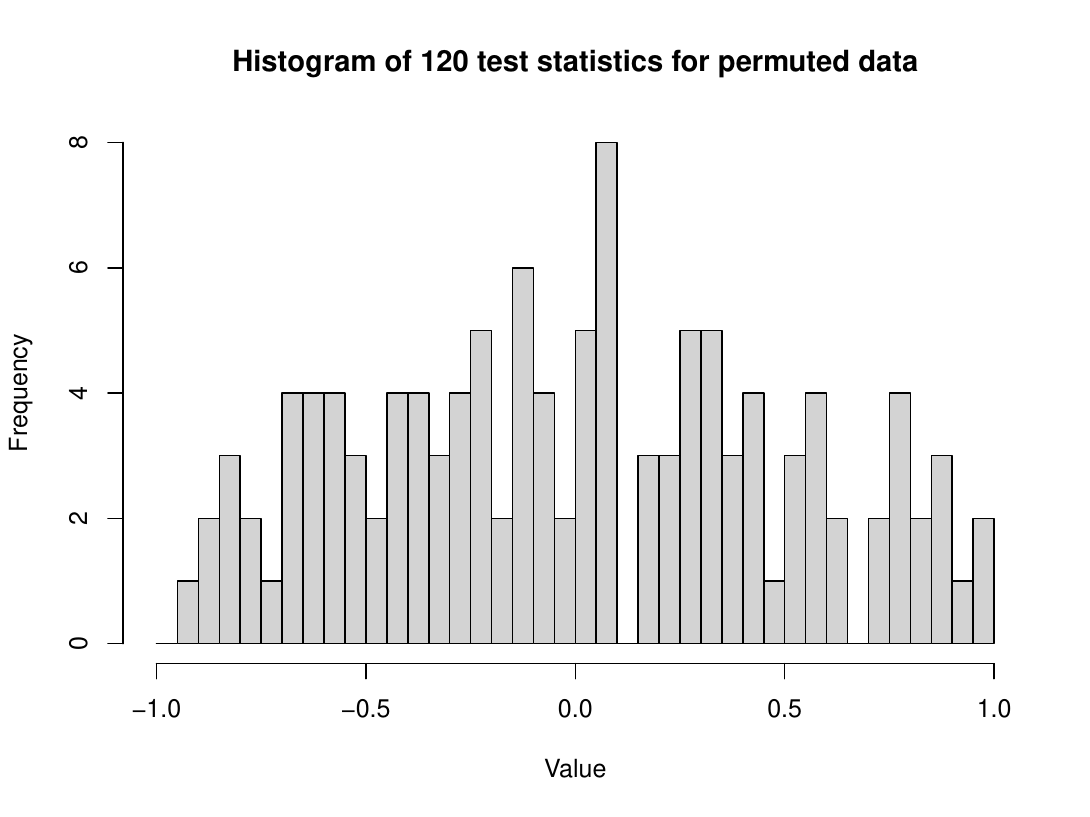}
    \caption{Histogram of $5!=120$ test statistics based on permuted versions of the dataset on car models}
    \label{fig:histptest}
\end{figure}

Note that each permutation map $\pi()$ is a function from $\reals^5\rightarrow \reals^5$. Let $\G$ be the set containing all these 120 permutation maps.
It turns out that we can compute a valid right-sided permutation p-value as 
$$ P_r := \frac{|\{\pi\in \G: r(X,\pi(Y))\geq r(X,Y)\}|}{|\G|} =\frac{107}{120}.$$
and valid left-sided permutation p-value as 
$$ P_l := \frac{|\{\pi\in \G: r(X,\pi(Y))\leq r(X,Y)\}|}{|\G|} =\frac{14}{120}.$$
A valid two-sided p-value is then 
$$P_{lr}:=2\cdot\min\{P_r,P_l\}=28/120\approx  0.2333.$$

An alternative  two-sided p-value --- which is also valid but may be somewhat different from the two-sided p-value above --- is
$$ P_{lr}':= \frac{|\{\pi\in \G: |r(X,\pi(Y))|\geq |r(X,Y)|\}|}{|\G|} =\frac{29}{120}.$$
In this example, $P_{lr}'=29/120\approx 0.2417,$ which is slightly different from $P_{lr}$. In general, these p-values are close when the permutation distribution (shown in Figure \ref{fig:histptest}) is approximately symmetric.
As usual, we reject $\Hy_{\text{indep}}$ when the p-value is at most $\alpha$. This ensures the type I error rate is at most $\alpha$.

It turns out that if the data are continuous (so that with probability 1 there are no ties among the test statistics), then both these p-values are exact. More precisely,  if $\Hy_{\text{indep}}$ is true, then $P_{lr}$ is uniformly distributed on $$\{2/120, 4/120, 6/120,...,120/120\}$$ and $P_{lr}'$ is uniform on $$\{1/120, 2/120, 3/120,...,120/120\}.$$ Note that an advantage of using $P_{lr}'$ is that its  smallest possible value is $1/120$, whereas the smallest possible value of $P_{lr}$ is 1/60.

Recall that the p-value computed with the parametric correlation test (\S\ref{secexcortest}) was $p= 0.2412$. This is close to the permutation-based p-values. This does not have to be case in general, in particular when $n$ is small and the data are not Gaussian.


\subsubsection{Example 2: a nonparametric test based on sign-flipping} \label{secmaizeflip}

In \citt[][ch. IV]{fisher1935} Fisher first proposed a certain nonparametric (or perhaps semiparametric) test based on sign-flipping. This test is somewhat related to the permutation test above. Fisher applied the test to a dataset collected by Charles Darwin. The dataset contains paired observations of heights of maize plants (crossed vs. self-fertilized plants). The first few columns of the dataset are shown in Table \ref{tab:maizedata}. For simplicity we will only use these observations as our dataset.

\begin{table}[h]
\centering
\begin{tabular}{ccc}
\hline
Crossed & Self-fert.  \\
\hline
$23\frac{1}{8}$ & $17\frac{3}{8}$ \\
12 & $20\frac{3}{8}$  \\
21 & 20 \\
22 & 20\\
$19\frac{1}{8}$ & $18\frac{3}{8}$  \\
$21\frac{4}{8}$ & $18\frac{5}{8}$  \\
\hline
\end{tabular}
\caption{Heights (inches) of crossed and self-fertilized plants.}
\label{tab:maizedata}
\end{table}

 The observations are paired, which means that the rows are independent of each other, but the two observations in each row may be dependent of each other.
 Denote the differences between the first and second observation in each row by $d=(d_1,...,d_6)$.
A classical parametric way to test the null hypothesis  $\Hy_0$ that the $d_i$ have mean 0, is to assume the $d_i$ are i.i.d. and normally distributed and to perform a t-test. However, we then require the following (if $\Hy_0$ is true):
\begin{itemize}
\item The $d_i$ are independent of each other;
\item The $d_i$ are normally distributed;
\item The $d_i$ all have the same variance.
\end{itemize}
Fisher's sign-flipping test requires fewer assumptions. It only requires the following:
\begin{itemize}
\item The $d_i$ are independent of each other;
\item The $d_i$ are symmetric about their means.
\end{itemize}
The test works as follows: we consider all $2^n=2^6$ sign-flipped versions of the data, i.e., all vectors of the form $(s_1 d_1,...,s_6 d_6)$, where $(s_1,...,s_6)\in\{-1,1\}^6$. 
Let $\G_{\pm}$ be the set of all $2^6$ transformations of the form $(x_1,...,x_6)\mapsto (s_1x_1,...,s_6x_6)$, where $(s_1,...,s_6)\in\{-1,1\}^6$. Thus, the sign-flipped versions of the data are all vectors $g(d)$, $g\in \G_{\pm}$.

For each of these  $2^n$ transformed vectors $g(d)$, we compute the sum of the entries:
$$T(g(d))=\sum_{i=1}^n (g(d))_i.$$
A histogram of the resulting sums (test statistics) is shown in Figure \ref{fig:histflipmaize}. Note that it is symmetric (why?).
To compute a p-value, we can proceed as in  \S\ref{secexptest}, with the difference that we use different transformations (sign-flips instead of permutations) and different test statistics (sum instead of empirical correlation). For example, 
a valid two-sided p-value is  
$$ P_{lr}:= \frac{|\{g\in \G_{\pm}: |T(g(d))|\geq |T(d)|\}|}{|\G_{\pm}|} =\frac{44}{64}.$$
If we would compute a p-value using a t-test, we would likely get a similar result. Indeed, there is no sign that these particular  data are very non-Gaussian or heteroscedastic.

\begin{figure}[ht!]
    \centering
    \includegraphics[width=0.85\textwidth]{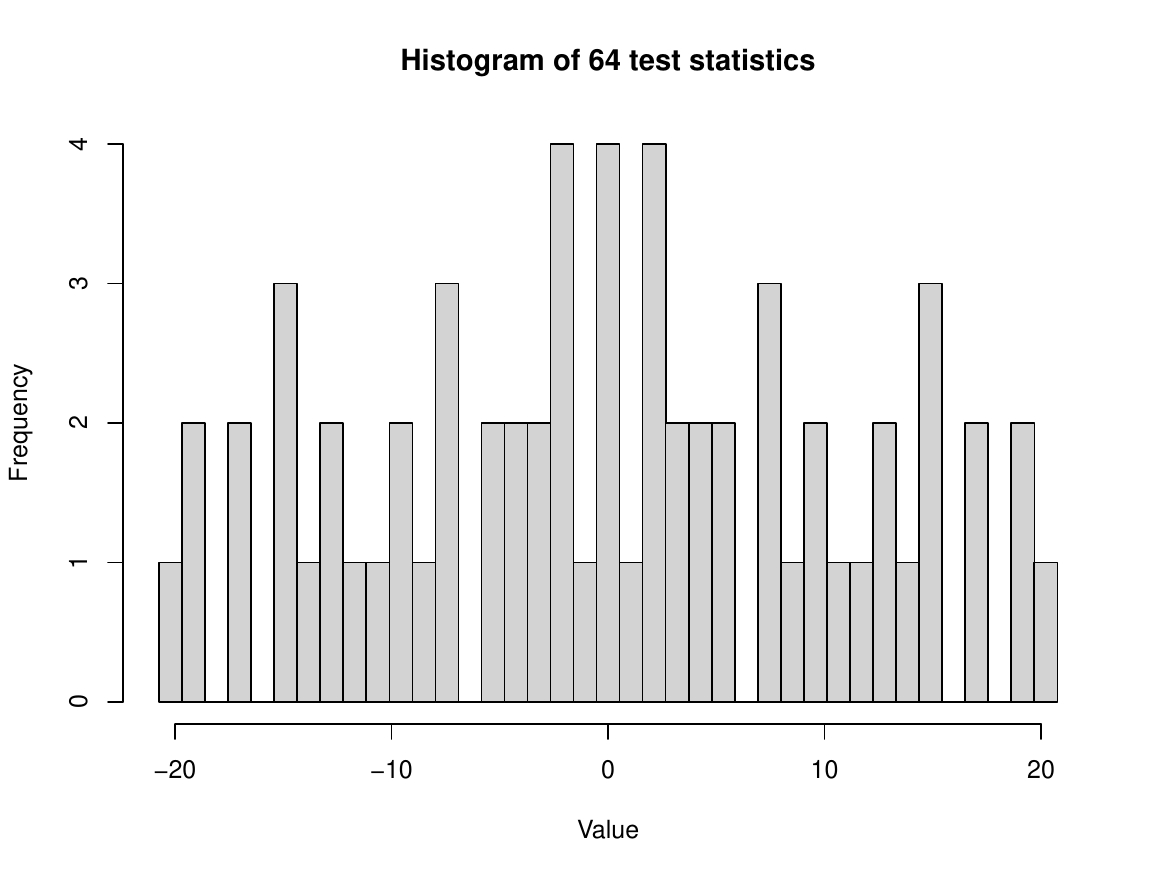}
    \caption{Histogram of $2^6=64$ test statistics based on sign-flipped versions of the dataset on maize plants.}
    \label{fig:histflipmaize}
\end{figure}

Darwin's full dataset has 15 rows. If we would apply the test to the full dataset,  this would mean applying $2^{15}=32768$ different sign-flipping transformations. In Fisher's time, that was was very time-consuming, although he managed to perform the test for this particular dataset, see \citt[][\S21]{fisher1935}. Nowadays it is easily doable with a computer. 


\subsubsection{General theory on permutation and group invariance tests} \label{secgroupinv}
In the examples above, we considered permutation maps and sign-flipping maps. It turns out that these sets of transformations are \emph{groups}, in the algebraic sense. A set $\G$ of transformations $g:\ssp\rightarrow\ssp$ from a sample space $\ssp$ to itself is called a group when it has the following properties:
\begin{enumerate}
\item every element $g\in \G$ has an inverse $g^{-1}$ in $\G$;
\item if $g$, $h\in \G$, then $g\circ h\in \G$;
\item the identity map $id$ is contained in $\ssp$;
\end{enumerate}
Here $g\circ h$, or $gh$ for short, indicates the composition of the maps $g$ and $h$, so that $g\circ h(x)=g(h(x))$. Note that the third property is in fact implied by the first two. 
An important property of groups that we will use, is that 
\begin{equation} \label{GgisG}
\forall g\in \G: \quad \G g=\G,
\end{equation}
 where $\G g$ means $\{h\circ g: h\in \G\}$.
Because tests such as those in \S\ref{secexptest} and \S\ref{secmaizeflip} involve groups of transformations, they are sometimes called \emph{group invariance tests}. Thus, a permutation test for example, is a special case of a group invariance test.

\begin{example}[Permutation maps form a group.]
In the example from \S\ref{secexptest} the sample space can be considered to be $ \reals^5\times\reals^5$. The transformations $g\in \G$ act on the data by permuting the second column. Note that the group $\G$ of $5!$ permutation maps has a group structure: every permutation map has an inverse which is also a permutation map, and if $g$, $h\in \G$, then clearly $g\circ h\in \G$.
\end{example}

\begin{example}[Sign-flipping maps form a group.]
In the example from \S\ref{secmaizeflip} the sample space was $\ssp=\reals^6$, since there are 6 observations, $d_1,...,d_6$. 
Every $g\in \G$ is of the form $(x_1,...,x_6)\mapsto (s_1x_1,...,s_6x_6)$, with $(s_1,...,s_6)\in\{-1,1\}^6$.
Note that every $g\in \G$ has an inverse $g^{-1}$ in $\G$. Indeed, note that every $g\in \G$ is its own inverse: $g=g^{-1}$. Moreover, if $g$, $h\in \G$, then clearly $g\circ h\in \G$.
\end{example}

We will call the data $X$. This could e.g. be a random vector or matrix. 
A sufficient  property that makes the tests valid (and often exact) is that \emph{under $\Hy_0$, the distribution of the data  is invariant under the transformations}, i.e., 
\begin{equation} \label{eqperminv}
\forall g\in \G: \quad X \,{\buildrel d \over =}\, gX,
\end{equation}
where $gX$ is short for $g(X)$.
In fact, it is sufficient that the following assumption (which is implied by property \eqref{eqperminv}) is satisfied. We will assume that $\G$ is finite, although this can be generalized.
\begin{assumption} \label{assgroupinv}
Under $\Hy_0$, the joint distribution of the test statistics $T(gX)$, $g\in \G$, is invariant under all transformations in $\G$ of $X$. That is, writing $\G=\{a_1,...,a_{|\G|}\}$, under $\Hy_0$
\begin{equation} \label{pinv}
\big(T(a_1X),...,T(a_{|\G|}X)\big)\,{\buildrel d \over =}\, \big(T(a_1gX),...,T(a_{|\G|}gX)\big)
\end{equation}
for all $g\in \G$.
\end{assumption}
The test statistic $T(\cdot)$ can be anything, for example a left-sided test statistic or a two-sided test statistic (e.g., the absolute value of a t-statistic). In order to have an asymptotically exact permutation test, it is not necessary for  Assumption \ref{assgroupinv} to be exactly satisfied \citp{romano1990behavior, canay2017randomization, kashlak2022asymptotic}.

In the examples in  \S\ref{secexptest} and \S\ref{secmaizeflip}, we computed a test statistic for each transformed version $gX$ of the data. In general a group invariance test considers $|\G|$ test statistics $T(gX)$, $g\in \G$ and bases the test on these test statistics. The null hypothesis $\Hy_0$ is rejected if the original statistic
$T(X)$ is extreme compared to the reference statistics $T(gX)$, $g\in \G$.
The p-value can be computed as 
\begin{equation} \label{eq:groupinvpv}
P= \frac{|\{g\in \G: T(gX)\geq T(X)\}|}{|\G|}.
\end{equation}

A basic group invariance test test rejects $\Hy_0$ when $T(X)>T^{(k)}(X)$, where 
$$ T^{(1)}(X)\leq...\leq T^{(|\G|)}(X)$$ are the sorted test statistics $T(gX)$, $g\in \G$, and $k=\lceil (1-\alpha)|\G|\rceil$, where $\lceil\cdot\rceil$ means rounding upwards to an integer. Note that the p-value \eqref{eq:groupinvpv} is at most $\alpha$ if and only if the group invariance test rejects.
 As is  stated in the following theorem, the test has level at most $\alpha$. 

\begin{theorem} \label{basic}
Suppose Assumption \ref{assgroupinv} is satisfied.
Then, under $\Hy_0$, $\mathbb{P}\big\{T(X)>T^{(k)}(X)\big\}\leq \alpha,$ i.e., the group invariance test controls the type I error rate.
\end{theorem}

\begin{proof}
By the group structure, $\G g=\G$ for all $g\in \G$. Hence $T^{(k)}(gX)=T^{(k)}(X)$ for all $g\in \G$.
Let $h$ have the uniform distribution on $\G$. Then under $\Hy_0$, the rejection probability is
$$
\begin{aligned}
&\mathbb{P}\big\{T(X)>T^{(k)}(X)\big\}=\\
&\mathbb{P}\big\{T(hX)>T^{(k)}(hX)\big\}=\\
&\mathbb{P}\big\{T(hX)>T^{(k)}(X)\big\}.
\end{aligned}
$$
The first equality follows from the null hypothesis and the second equality holds since $T^{(k)}(X)=T^{(k)}(hX)$. 
Since $h$ is uniform on $\G$, the above probability equals
$$\mathbb{E} \Big [(|\G|)^{-1} \cdot |\big \{g\in \G: T(gX)> T^{(k)}(X) \big \}| \Big ] \leq \alpha,$$ 
as was to be shown.

\end{proof}
The above proof and a different proof can be found in e.g. \citt{hemerik2018exact}. The different proof considers the \emph{orbit} of the data $X$, which is defined as $O_X=\{gX:g\in G\}$ and shows that under $\Hy_0$, conditional on the orbit $O_X$, the probability of rejecting $\Hy_0$ is at most $\alpha$.

Any proof of the validity of the group invariance test, strongly relies on the group structure of the set $\G$ of transformations. Indeed, the group invariance test is based on the property that under $\Hy_0$, for each $g\in \G$, the probability  $\mathbb{P}\big\{T(gX)>T^{(k)}(X)\big\}$ is the same. That follows from the fact that for every $g\in \G$, the joint distribution of $(gX,\G X)$ is the same. That holds because if $g, g'\in \G$, then under $\Hy_0$ we have
$$ (gX,\G X) =  (gX,\G g  X) \,{\buildrel d \over =}\,  (X,\G   X)   \,{\buildrel d \over =}\,  (g'X,\G g'  X) =(g'X,\G X),$$
where the property $\G=\G g$ holds due to the group structure.

If we use a set of transformations that is not a group, then the group invariance test can be very conservative or anti-conservative. An example is when  we use a subset (that is not a subgroup) of the group of permutation maps \citp{southworth2009properties}. \citt[][p.822]{hemerik2018exact} show however that we can validly use set of transformations that is not a group, if we modify the test by applying a certain additional random transformation. This idea is extended in \citt{ramdas2023permutation}.

In \S\ref{secexptest} and \S\ref{secmaizeflip} we saw two examples of groups of transformations, namely permutation maps and sign-flips. These are perhaps the most common examples. Another example are rotations \citp{solari2014rotation,koning2024more}. Tests based on rotations require more assumptions than tests based on sign flips, so they are used less often in practice. It  can be useful to e.g. permute multiple columns of a matrix (i.e. applying the same permutation map to every column). The transformations then also form a group.

In the example from \S\ref{secexptest}, the group $\G$ contained $n!=5!$ elements. Note that the size of the group quickly grows with the sample size. For example, if $n=10$, then $n!=3628800$ and if $n=20$, then $n!\approx 2.4329\cdot 10^{18}$. As a solution, researchers then do not use the full group, but use a random collection of transformations from $\G$ that are uniformly sampled. The sampling can be done with or without replacement. For example, if we sample with replacement, we can proceed as follows: let $M$ be large number (e.g. 1000) and sample $M-1$ independent, random (uniformly distributed) transformations from $\G$, then add the identity map to this collection of transformations.
Then apply the group invariance test as usual, but using these $M$ transformations rather than the full group $\G$.
If we sample the random transformations in this way, then (on average) the group invariance test still controls the type I error rate \citp{hemerik2018exact}. Tests based on random permutations or other transformations, are usually considered to be  part of the class of \emph{resampling methods}, just like bootstrap tests, which are discussed in \S\ref{secboot}.

\subsubsection{Permutation tests in more complex models} \label{secpermcompl}

Consider a linear model with responses $Y\in \reals^n$ and covariates $X^1,Z^1,Z^2,...,Z^p\in\reals^n$. The responses satisfy $Y= X \beta +Z\gamma +\epsilon$, where $Z=(Z^1,...,Z^p)$ and the residuals $\epsilon_i$ are independent with mean 0. Suppose we are interested in testing  $\Hy_0:\beta=0$.  If the classical Gaussian model holds, we can test $\Hy_0$ in an exact way. Suppose however that the residuals are not Gaussian or are not homoscedastic. Then a permutation-type test can be a good alternative, in the sense that it may have type I error rate closer to $\alpha$ and may have more robust power. A well-known permutation method in this situation is the Freedman-Lane method (\citealp{freedman1983nonstochastic}, and see \citealp{winkler2014permutation} for a review of related methods). If there are multiple response variable that we are interested in, then besides robustness, a potential advantage of Freedman-Lane is that it can be combined with powerful resampling-based multiple testing methods (see \S\S\ref{secresamplingFWER}, \ref{secRW}, \ref{secribo}). This means that if the residuals are perfectly Gaussian and homoscedastic, it can nevertheless be useful to apply  Freedman-Lane, since it can be combined with those powerful methods.

The Freedman-Lane method is defined as follows.
\begin{enumerate}
\item Compute some (pivotal) test statistic $T^1$ for testing $\Hy_0$, e.g. the (absolute value of the) Wald statisic $\hat{\beta}/se(\hat{\beta})$.
\item Regress $Y$ on $Z$, obtaining $\hat{Y}= Z(Z'Z)^{-1}Z'Y$ and residuals $e= Y-\hat{Y}$.
\item (Randomly) permute (or sign-flip) the residuals, obtaining new responses $Y^*=\hat{Y}+g(e)$, where $g(e)$ are the permuted or sign-flipped residuals. Compute a test statistic $T^2$ in the same way as before, but using the changed responses $Y^*$.
\item Repeat the previous step many times, obtaining test statistics $T^3,...,T^w$.
\item The (right-sided) p-value is $P= w^{-1}|\{1\leq j \leq w: T^j\geq T^1\}|$.
\end{enumerate}
The intuition behind Freedman-Lane is as follows. Note that if $\Hy_0$ is true  and the residuals $\epsilon_i$ are i.i.d., then $Y^*$ has roughly the same distribution as $Y$, so $T^2, T^3,...$ have roughly the same distribution as $T^1$, so that $P$ is roughly uniformly distributed on $\{1/w,2/w,...,w/w\}$. If $\Hy_0$ is false, then $T^1$ tends to be large, but after permuting the residuals, the conditional dependence between $Y^*$ and $X$ will be small, so $T^2, T^3,...$ do not tend to be large. (For more theory, see \citealp{anderson2001permutation}.)

Permutation-type tests for \emph{generalized} linear models are available as well, although the literature is small. See \citt{potter2005permutation} for a permutation test for logistic regression.  See \citt{hemerik2020robust} and \citt{de2025inference} for sign-flipping-based tests for generalized linear models. These are implemented in the R package \emph{flipscores} on CRAN. \citt{de2025permutation} illustrate how the tests from the latter two papers can be combined with the resampling-based multiple testing method \emph{maxT} (defined in \S\ref{secresamplingFWER}).

\subsubsection{Bootstrap tests} \label{secboot}
The term ``bootstrap method'' is often used in contexts where confidence intervals are constructed based on resampling with replacement from the data. By \emph{bootstrap testing}, however, we mean a type of \emph{test}. 
Sometimes such a bootstrap test is simply equivalent to making a bootstrap confidence interval and checking whether $\theta_0$ (e.g. 0) is within the confidence interval.
However, when we combine bootstrap tests with multiple testing method such as maxT and Romano-Wolf (discussed) then we really need a bootstrap test in the strict sense, which does not involve making a confidence interval.
Like bootstrapping for making confidence intervals, bootstrap tests are also based on sampling with replacement (except so-called ``wild bootstrap tests'').

Like permutation testing, 
bootstrap testing is manner of nonparametric or semiparametric testing. Bootstrap tests are related to permutation tests in the sense that very roughly speaking, both approaches repeatedly ``simulate samples from the null distribution''. This description is indeed very rough; for example, usually there is no such thing as \emph{the} null distribution, since the null hypothesis typically contains many distributions.

The class of \emph{bootstrap tests} is quite varied, and there is no simple overarching theory. Some bootstrap tests are so-called \emph{wild bootstrap tests} \citp{davidson2008wild,dikta2021bootstrap}. These are not bootstrap tests in the strictest possible sense, but are  related to the permutation test for linear models which was discussed in \S\ref{secpermcompl}. Below, we focus on ``non-wild'' bootstrap tests.

The differences between permutation tests (or, more generally, group invariance tests) and (non-wild) bootstrap tests are that, unlike permutation tests,
\begin{itemize}
\item bootstrap tests start with adjusting the data by removing the signal;
\item  bootstrap tests sample with replacement from this adjusted dataset; 
\item  bootstrap tests do not e.g. permute the data.
\end{itemize}
Those are differences between the two types of tests if we see them purely as \emph{procedures}. Another important difference lies in the types of hypotheses that can be tested: (non-wild) bootstrap tests are strictly speaking not used for testing a hypothesis of \emph{exchangeability}, but for testing a hypothesis about a \emph{parameter}. For example, if our data are $(X_1,...,X_n)$ and we want to test that $X_1,...,X_n$ are i.i.d. (i.e. exchangeable) then we would use a permutation test. On the other hand, if we assume that $X_1,...,X_{n/2}$ have mean $\mu_1$ and $X_{n/2+1},...,X_{n}$ have mean $\mu_2$ and want to test $\Hy_0:\mu_1=\mu_2$ then in some cases we might prefer a bootstrap test over a permutation test. Another important difference is that bootstrap tests are only asymptotically exact (i.e., the level converges to $\alpha$ if $n\rightarrow\infty$), while permutation tests are sometimes  exact for finite samples (as shown in \S\ref{secgroupinv}).

Denote the data by $X$ and suppose the null hypothesis is of the form $\Hy_0:\theta=\theta_0$. Generally speaking, a (non-wild) bootstrap test, based on $b$ bootstrap samples, proceeds as follows:
\begin{enumerate}
\item Compute a test statistic $T(X)$, extreme  values of which are evidence against $\Hy_0$.
\item Modify the data such that the distribution of the modified data $X^*$ ``falls under $\Hy_0$''.
\item For all $1\leq j \leq b$, do the following. Create a dataset $X^j$ by sampling with replacement from $X^*$, and compute $T(X^j)$.
\item Reject $\Hy_0$ if the original test statistic $T(X)$ is extreme compared to the reference values $T(X^1),...,T(X^b)$.
\end{enumerate}
In the above, ``extreme'' typically means one of the following things:
\begin{itemize} 
\item $T(X)$ exceeds the $(1-\alpha)100\%$-quantile of the reference distribution;
\item $T(X)$ lies below the $(\alpha/2)100\%$-quantile or above the $(1-\alpha/2)100\%$-quantile of the reference distribution.
\end{itemize}
As mentioned, the level of a bootstrap test is not guaranteed to be exactly $\alpha$, but in various cases the level is $\alpha$ asymptotically. Finally, note that we can also compute a p-value, by computing the number of reference statistics that are more extreme than the original statistic $T(X)$ and dividing by $b$. This p-value can be 0. We can avoid that by including the original statistic among the reference statistics. We now give two examples of bootstrap tests.
\\
\\
\textbf{One-sample bootstrap test.}
Consider an i.i.d. sample $X_1,...,X_n$, denote the mean of the observations by $\mu$, and consider $\Hy_0:\mu=\mu_0$. Permuting the data is not going to help us in this case. A test based on sign-flipping (see \S\ref{secmaizeflip}) can be a good choice, although is only exact if the variables have  symmetric distributions.
Another approach is to use the following bootstrap test.
\begin{enumerate}
\item Let $T(X)$ be the usual one-sample t-statistic corresponding to  $\Hy_0:\mu=\mu_0$.
\item Let $\hat{\mu}$ be the sample mean $n^{-1}\sum_{i=1}^n X_i$.
\item Construct the  modified dataset $X^*:=(X_1,...,X_n) - \hat{\mu} + \mu_0$.
\item For all $1\leq j \leq b$, do the following. Consider a bootstrap sample $X^j$ by sampling with replacement from $X^*$. Compute the one-sample t-statistic $T(X^j)$ corresponding to the hypothesis that the mean is $\mu_0$.
\item Reject $\Hy_0$ if the original test statistic $T(X)$ is extreme compared to the reference values $T(X^1),...,T(X^b)$.
\end{enumerate} \phantom{.}


A different but equivalent algorithm is the following.
\begin{enumerate}
\item Let $T(X)$ be the usual one-sample t-statistic corresponding to  $\Hy_0:\mu=\mu_0$.
\item Let $\hat{\mu}$ be the sample mean $n^{-1}\sum_{i=1}^n X_i$.
\item For all $1\leq j \leq b$, do the following. Consider a bootstrap sample $X^j$ by sampling with replacement from $X$. Compute the one-sample t-statistic $T(X^j)$ corresponding to the hypothesis that the mean is $\hat{\mu}$.
\item Reject $\Hy_0$ if the original test statistic $T(X)$ is extreme compared to the reference values $T(X^1),...,T(X^b)$.
\end{enumerate} \phantom{.}

To see that the two algorithms are equivalent, note that resampling from $X^*$ and computing a test statistic  corresponding to the hypothesis that the mean is $\mu_0$, is equivalent to resampling directly from $X$ and computing a test statistic  corresponding to the hypothesis that the mean is $\hat{\mu}$.

Even if $X_1,...,X_n$ are nicely i.i.d., this bootstrap test is only asymptotically exact. If $X_1,...,X_n$ are not i.d.d. but have different variances or shapes, then this bootstrap test can in some cases still be asymptotically exact.
\\
\\
\textbf{Two-sample bootstrap test.}
Consider two i.i.d. samples $Y_1,...,Y_{n_1}$ and $Z_1,...,Z_{n_2}$, with means  $\mu_1$ and $\mu_2$ respectively. Write $X=(Y_1,...,Y_{n_1}, Z_1,...,Z_{n_2})$. Consider the null hypothesis $\Hy_0: \mu_1=\mu_2$.
It is well known that under certain assumptions, the two-sample t-test will be exact. If e.g. the normality assumption does not hold, we could consider a permutation test. For such a test to be exact, we would need that all $n_1+n_2$ variables are identically distributed under $\Hy_0$. If the two groups have different variances, then sometimes a permutation test based on  Welch statistics \citp{janssen1997studentized} is more reliable.
In general, if the two groups do not have the same distribution under $\Hy_0$, the following bootstrap test may also be a good choice. 
\begin{enumerate}
\item Let $T(X)$ be the two-sample Welch  t-statistic.
\item Compute the two sample means $\hat{\mu}_1$ and $\hat{\mu}_2$.
\item Let $Y^*=Y -\hat{\mu}_1$ and $Z^* = Z- \hat{\mu}_2$.
\item For all $1\leq j \leq b$, do the following. Take a bootstrap sample $Y^j$ (of length $n_1$) by sampling with replacement from $Y^*$, and  take a bootstrap sample $Z^j$ (of length $n_2$) by sampling with replacement from $Z^*$. Write $X^j=(Y^j, Z^j)$. Compute the the two-sample Welch  t-statistic $T(X^j)$.
\item Reject $\Hy_0$ if the original test statistic $T(X)$ is extreme compared to the reference values $T(X^1),...,T(X^b)$.
\end{enumerate} \phantom{.}
The above version of the two-sample bootstrap test samples the bootstrap samples $Y^j$ and $Z^j$ independently, and uses the Welch t-statistic. Such an approach
is recommended in particular when the two samples cannot be assumed to have identical distributions under $\Hy_0$. For example,  sometimes the  variables in the second sample  have different variances than the variables in the first sample.

\subsection{Exercises}
\begin{enumerate}
\item 
(a) Consider a valid p-value $P$ for a hypothesis $\Hy_0$. As usual, we reject $\Hy_0$ when $P\leq \alpha$. This generally ensures that the type I error probability stays below $\alpha$. Does this still generally hold is we choose $\alpha$ based on the data?

(b) Related to that, can you mention an advantage of the fact that in several scientific fields, researchers all use  $\alpha=0.05$ by default?

\item Using the fact that $\G$ is a group, prove that property \eqref{GgisG} indeed holds.
\item At the end of \S\ref{secgroupinv}, it is explained that if we use (i.i.d. uniform) random transformations in a group invariance test, we should always add the identity map. Explain that if we do not do that, then the p-value of the group invariance test may potentially be 0 under $\Hy_0$. Explain why that would mean that the p-value is not valid.
\item Consider the one-sample bootstrap test from \S\ref{secboot}. Suppose that the observed values $x_1,...,x_n$ are all distinct. Suppose that the we use $b$ bootstrap samples $X^1,...,X^b$.  We compute the p-value  by taking the number of reference statistics that are larger than or equal to the original statistic $T(X)$ and dividing by $b$. What is the smallest possible p-value that we may get? And if we take the first bootstrap sample $X^1$ to be the original data $X$?

\item $\bigstar$ Consider the one-sample bootstrap test from \S\ref{secboot}. Suppose that the observed values $x_1,...,x_n$ are all distinct. Recall that a \emph{bootstrap sample} is a collection $\{x_1',...,x_n'\}$ such that   $x_i'\in\{x_1,...,x_n\}$ for every $1\leq x_i' \leq n$. Compute (in terms of $n$) the number of different possible bootstrap samples.
To get the answer, make use of  the following fact from combinatorics: the number of sequences $(b_1,...,b_n)$ such that the $b_i$ are nonnegative integers and $b_1+...+b_n=n$, equals  $\binom{2n-1}{n}.$ 

\item In recent years, \emph{e-values} have been a topic of interest in the literature on hypothesis testing. An e-value for a null hypothesis $\Hy_0$ is defined as a nonnegative random variable $e$ that satisfies $\mathbb{E}(e)\leq 1$ under $\Hy_0$. 
In the context of e-values, one rejects $\Hy_0$ if $e\geq \alpha^{-1}$. This guarantees that $\pr(\text{reject})\leq\alpha$ under $\Hy_0$. Prove that.

Further, prove that for every e-value $e$ that takes values in $(0,\infty)$,  the variable  $P:=\min\{e^{-1},1\}$ is a valid p-value.

(e-values have certain attractive properties: they allow for online testing with optional stopping, and allow  $\alpha$ to be chosen after seeing the data while still providing error control in a particular sense.)



\end{enumerate}

\begin{sols}
\subsection{Solutions}
\begin{enumerate}

\item 

(a) No.  Suppose for example that we choose $\alpha$ based on the p-value, namely $\alpha:=P$. Then we always reject $\Hy_0$. Suppose $\Hy_0$ is true. Then clearly, conditional on the $\alpha$ we have selected, the type I error probability is 1, which is larger than $\alpha$ (if $\alpha<1$). (\citealp{hemerik2026choosing} further discusses consequences of choosing $\alpha$ based on the data.)

(b) If everyone in a field uses $\alpha=0.05$, an advantage is that it is clear that researchers in this field do not choose $\alpha$ based on the data.

\item Let $g\in\G$. We will show that $\G\subseteq \G g$ and $\G g \subseteq \G$.   

To show that $\G\subseteq \G g$, we show that if $h\in \G$ then $h\in \G g$.
Thus, we must show that $h$ is of the form $ag$, with $a\in \G$. This is indeed true, since we can take $a=hg^{-1}$ (which is contained in $\G$ since $h$ and $g^{-1}$ are in $\G$ and $\G$ is a group) and then $h=ag$.

Next, we must show that $\G g \subseteq \G$. This follows from the fact if $a\in \G$, then $ag\in \G$, by definition of a group.

\item
The p-value for the group invariance test is \eqref{eq:groupinvpv}. 
If we use random transformations from $\G$, the formula is the same, except that we replace $\G$ by a random, smaller collection of transformations from $\G$.
Note that if we use random transformations but add the identity map, then it is guaranteed that the numerator in the formula of the p-value is at least 1, since $T(id(X))\geq T(X)$. Otherwise the numerator may be 0, since all sampled random transformations $g$ might satisfy $T(gX)< T(X)$.

If a p-value has a strictly positive probability of being 0 under $\Hy_0$, then it does not satisfy $\pr_{\Hy_0}(P\leq c)\leq c$ for all $c\in[0,1]$, so it is not valid.
(A test based on such a p-value does not control the type I error rate for all $\alpha\in(0,1)$.)

\item The smallest possible p-value is 0, since it may happen that $T(X)$ is more extreme than the reference statistics. The answer to the second question is $1/b$, since we always have $T(X^1)\geq T(X)$ in the second case.

\item
For each $x_i'$ there are $n$ possibilities, so if we would pay attention to the order of $x_1',...,x_n'$, then there would be $n^n$ possibilities. However,  we only care about the set $\{x_1',...,x_n'\}$  and not about the order of $x_1',...,x_n'$. This might suggest that the answer is $n^n/n!$. However, then we would ignore the fact that there are usually ties among $x_1',...,x_n'.$ For example, in the extreme case that $x_1'=...=x_n'$, there are not $n!$ ordered samples corresponding to that sample, but only one.

What we really need to compute, is the number of \emph{multisets} of size $n$ from a collection $\A$ of $n$ elements. A multiset is a set that can contain elements more than once.
Such a multiset can be represented by a vector $(b_1,...,b_n)$, where $b_i$ indicates the number of times that the multiset contains the $n$-th element of $\A$. Note that such a vector satisfies $b_1+...+b_n=n$. Hence, we can use the given fact, which means that the answer is $\binom{2n-1}{n}.$ 

(The given fact follows from a so-called stars-and-bars reasoning:
imagine we draw $2n-1$ bars in a row (e.g. $|||||||$ in case $2n-1=7$.). Now turn $n$ of the bars into stars. There are $\binom{2n-1}{n}$ possible ways to do this.  We then have $n$ stars among $n-1$ bars, and each such pattern corresponds 1-to-1 to a way of putting $n$ (identical) things into $n$ bins; indeed, we can see the $n-1$ bars as the borders between the bins. Clearly, the number of possible ways of putting $n$ marbles into $n$ bins equals the number of sequences $(b_1,...,b_n)$ such that the $b_i$ are nonnegative integers and $b_1+...+b_n=n$.)

\item 
We first prove the first claim.
Suppose $\Hy_0$ holds.
By Markov's inequality, 
$\mathbb{E}(e)\geq \alpha^{-1}\pr(e\geq \alpha^{-1})$, which means that $\pr(\text{reject})\leq \alpha \mathbb{E}(e)\leq \alpha$.

We now prove the second claim. Consider  a valid e-value $e$ and   $P:=1\wedge e^{-1}$.
Let $c\in(0,1]$ and suppose $\Hy_0$ is true. We then have $\pr(P\leq c) = \pr(e\geq c^{-1}) $. By Markov's inequality, $\mathbb{E}(e)\geq c^{-1}\pr(e\geq c^{-1})$, and it follows that $\pr(P\leq c)\leq \mathbb{E}(e) c\leq c$. (Clearly, $\pr(P\leq c)\leq c$ also holds for $c=0$.) Thus, $P$ is a valid p-value.

\end{enumerate}

\end{sols}

\newpage

\section{Multiple testing: basic concepts and methods} \label{secmtbasics}
When we test a single hypothesis, we want to keep the probability of a false positive (type I error) small.
When we test many hypotheses, we likewise need to be careful not to incur a lot of false positives. This can be done by using a multiple testing method. This section discusses basics from multiple testing theory. Many ideas in this section will be built upon in later sections.

\subsection{Motivation for multiple testing methods} \label{motiv}
Suppose we wish to test $m>1$ hypotheses $\Hy_1,...,\Hy_m$ and we  compute corresponding p-values $P_1,...,P_m$. As a simple example, assume for now that all hypotheses are true and all p-values are independent of each other and uniformly distributed on $[0,1]$. Naively we would simply compare every p-value with $\alpha$ and reject all hypotheses with p-value below $\alpha$. However, the probability of one or more false findings is then not below $\alpha$ anymore. Indeed, we then have
$$\pr(\text{one or more false positives})= 1-\pr(\text{no false positives})$$
$$1-\prod_{1\leq i \leq m}\pr(\Hy_i \text{ not rejected}) = 1-(1-\alpha)^m,$$
which is larger than $\alpha$. Indeed, it is close to 1 if $m$ is large. 
Here we assumed that the p-values are independent of each other, but if they are dependent, then the probability of false positives will usually also be too high. If we would ignore this issue, then testing many hypotheses would mean that we have a high chance of rejecting one or more hypotheses, even if all hypotheses are true. This would flood the scientific literature, companies and institutions with false discoveries. To avoid this, we need multiple testing methods. Some such methods are also known as ``multiple testing corrections''.

\subsection{Global tests} \label{secglobaltests}
We start with \emph{global testing}, which can be seen as a simple (often too simple) approach to multiple testing, which is nevertheless insightful and serves as an important fundament for later theory. Indeed, global tests will be an essential concept in e.g. \S\ref{secCT} and \S\ref{secCTPfdp}.

A global test is a test of the so-called \emph{global null hypothesis}, which is the hypothesis that all our $m$ hypotheses $\Hy_1,...,\Hy_m$ are true. Thus, the global hypothesis, which we will denote by $\Hy_{\{1,...,m\}}$,  is the \emph{intersection} of our $m$ hypotheses:
$$\Hy_{\{1,...,m\}}= \bigcap_{1\leq i \leq m}\Hy_i.$$
To understand the above notation, remember that a hypothesis is a \emph{set} (of distributions), so the global null hypothesis is the intersection of sets. The global null hypothesis is the set of all distributions that are contained in each of the sets $\Hy_1,...,\Hy_m$. Thus, the global null hypothesis is smaller (or at least not larger) than each of the individual hypotheses, i.e., for every $1\leq i \leq m$ we have $\Hy_{\{1,...,m\}}\subseteq \Hy_i$. A global test ensures that the probability of rejecting $\Hy_{\{1,...,m\}}$ is at most $\alpha$ if $\Hy_{\{1,...,m\}}$ is true.
One could debate whether a global test deserves the title ``multiple testing method'' or not. Indeed, a global test allows saying something about a family of hypotheses, but strictly speaking it tests a single hypothesis.

In a sense, the global null hypothesis is always easier to reject than any of the individual hypotheses $\Hy_1,...,\Hy_m$. Indeed, if any of the hypotheses  $\Hy_1,...,\Hy_m$ is false, then that logically implies that $\Hy_{\{1,...,m\}}$ is also false. This means that in some applications, we will be able to reject $\Hy_{\{1,...,m\}}$, even if we are not able to reject any individual hypotheses. 

\subsubsection{Global test by Bonferroni} \label{secbonglobal}
The name \emph{Bonferroni} is well known because of the \emph{Bonferroni multiple testing correction} discussed later. However, the idea underlying Bonferroni's method can also be used to define a global test. Bonferroni's global test is defined as follows. Compute p-values $P_1,...,P_m$, corresponding to the hypotheses $\Hy_1,...,\Hy_m$. We will need that each p-value $P_i$ is valid, in the sense that  if $\Hy_i$ is true, then $P_i$ is uniformly distributed on $[0,1]$ --- or stochastically larger than uniform. We define $\N$ to be the set of the indices of the true hypotheses:

$$\N = \{1\leq i \leq m: \Hy_i \text{ is true}\}.$$
The p-values corresponding to $\N$ are sometimes called the \emph{null p-values}.

\begin{assumption} \label{assvalidp}
The $m$ p-values are valid, i.e., for all $i\in \N$ we have for all $c\in [0,1]$
$$\pr(P_i\leq c)\leq c.$$
\end{assumption}

We make this assumption \emph{throughout this text}, unless stated otherwise. Thus, by ``p-value'' we will mean ``valid p-value''.
Note that Assumption \ref{assvalidp} does not contain any assumption about the dependence structure of the p-values. Some methods will require such assumptions, but not Bonferroni.

The Bonferroni global test rejects the global hypothesis $\Hy_{\{1,...,m\}}$ if and only if at least one of the p-values is smaller than or equal to $\alpha/m$, i.e., if and only if
$$\min\{P_i: 1\leq i \leq m\}\leq \alpha/ m.$$

\begin{proposition} \label{propBonferroniglobal}
The Bonferroni global test has type I error rate at most $\alpha$. In other words, if $\Hy_{\{1,...,m\}}$ is true, then  
$$\pr\Big(\min\{P_i: 1\leq i \leq m\}\leq \alpha m\Big)\leq \alpha.$$
\end{proposition}

\begin{proof}
Suppose $\Hy_{\{1,...,m\}}$ is true. This means that all $m$ individual hypotheses are true. Hence 
$$\pr\Big(\min\{P_i: 1\leq i \leq m\}\leq \alpha /m\Big) =$$
$$  \pr\Big(P_i\leq \alpha/ m \text{ for some } 1\leq i \leq m\Big)=$$
$$  \pr\Big(\bigcup_{1\leq i \leq m}    \{P_i\leq \alpha /m \}\Big).$$
By Boole's inequality, the probability of a union of events is at most the sum of the probabilities of the events. Hence, the above is at most
$$\sum_{i=1}^m  \pr(P_i\leq \alpha /m) \leq m\cdot \alpha /m =\alpha.$$

\end{proof}

Later we will discuss Bonferroni's multiple testing procedure. There is a simple relationship between that procedure and Bonferroni's global test: Bonferroni's global test rejects $\Hy_{\{1,...,m\}}$ if and only if Bonferroni's multiple testing method rejects at least one of the hypotheses $\Hy_1,...,\Hy_m$.

\subsubsection{\v{S}id\'{a}k's global test}

\v{S}id\'{a}k's global test assumes that  the null p-values are independent of each other. More precisely, it is sufficient that under $\Hy_{\{1,...,m\}}$, all p-values $P_{1},..., P_{m}$ are independent of each other.

Let $c\in[0,1]$. 
In Section \ref{motiv} we discussed that \emph{if} all p-values are independent and standard uniform, 
then 
$$\pr(\exists 1\leq i \leq m: P_i\leq c)=$$
$$ 1-\pr(\forall 1\leq i \leq m: P_i> c)=$$
$$1-(1-c)^m.$$
It follows that  that if we choose $c$ such that $1-(1-c)^m=\alpha$, then we have an exact global test if we reject $\Hy_{\{1,...,m\}}$ if and only if $\min\{P_i: 1\leq i \leq m\} \leq c.$ It turns out that this value $c$ is $1-(1-\alpha)^{1/m}$.

\begin{theorem} \label{thmsidakglobal}
Suppose the null p-values 
are independent of each other, or at least that under $\Hy_{\{1,...,m\}}$, all p-values $P_{1},..., P_{m}$ are independent of each other. Consider \v{S}id\'{a}k's global test, i.e., reject  $\Hy_{\{1,...,m\}}$ if and only if 
$$\min\{P_i: 1\leq i \leq m\} \leq 1-(1-\alpha)^{1/m}.$$ Then the rejection probability is at most $\alpha$ under $\Hy_{\{1,...,m\}}$. The rejection probability is exactly $\alpha$ if we additionally assume that the null p-values are not only valid, but exactly  standard uniform.
\end{theorem}
\begin{proof}
Note that  if we take $c=1-(1-\alpha)^{1/m}$, then $$1-(1-c)^m    =$$
$$ 1-(1-   (1-(1-\alpha)^{1/m})   )^m  = $$
 $$ 1-((1-\alpha)^{1/m})^m =\alpha.$$

We saw that for this  $c$, if all p-values are independent and standard uniform, we have 
$$\pr(\text{reject } \Hy_{\{1,...,m\}})=$$
$$\pr(\exists 1\leq i \leq m: P_i\leq c)= \alpha,$$
so the test is exact. Clearly, when some of the p-values are stochastically larger than standard uniform, the rejection probability is at most $\alpha$.
\end{proof}

The advantage of \v{S}id\'{a}k's global test compared to Bonferroni is that it is more powerful, since $1-(1-\alpha)^{1/m} > \alpha/m$. However, a major caveat is that it is not valid for all possible dependence structure of the p-values. In some cases it is reasonable to assume that the p-values are independent, e.g. when the $m$ p-values are computed based on $m$ independent datasets. In many situations however, there is some dependence between the p-values and then \v{S}id\'{a}k's global test may be too liberal. Roughly speaking, this can happen when there are negative correlations between the p-values.

When there are dependencies between the p-values, \v{S}id\'{a}k's global test \emph{may} still be valid. However, this is often difficult to check, for example because the dependence structure of the p-values if often unknown.

\subsubsection{Global test using Fisher combinations}
Another well-known global test which can be more powerful than Bonferroni, is Fisher's combination test. The test is exact when the p-values are i.i.d. and standard uniform.
 An important caveat is that this test is often anti-conservative when there are positive dependencies between the p-values.

 The test rejects for large values of the test statistic
$$T = -2\sum_{i=1}^m \ln(P_i).$$ Suppose the p-values are i.i.d. and standard uniform.  Then $T$ has a  $\chi^2_{2m}$ distribution (chi-squared distribution with 2 degrees of freedom). 

\begin{theorem}
Suppose $P_1,...,P_m$ are independent and exactly uniform on $[0,1]$.
Then $$ -2\sum_{i=1}^m \ln(P_i) \sim \chi^2_{2m}.$$
\end{theorem}
\begin{proof}
The negative logarithm of a standard uniform variable follows an $Exp(1)$ distribution. Indeed, for $c\geq 0$,
\begin{itemize}
\item the CDF of the $Exp(1)$ distribution evaluated at $c$ is $1-\exp(-c)$ and 
\item the CDF of $-\log(P_i)$ evaluated at $c$ is $\pr(-\log(P_i)\leq c)  =\pr(P_i\geq \exp(-c))=1-\exp(-c).$
\end{itemize}
Further, it turns out that multiplying a $Exp(1)$ variable  by 2, gives a $\chi^2_{2}$ distribution. Summing $m$ independent $\chi^2_{2}$ variables gives an $\chi^2_{2m}$ distribution.
\end{proof}

Roughly speaking, when there are positive dependencies between the p-values, the global test based on Fisher combinations tends to be anti-conservative.

\subsubsection{Global test using Simes' inequality} \label{secglobalsimes}

We have seen that Bonferroni's global test is always valid as long as the p-values are valid (Assumption \ref{assvalidp}), while the other global tests that we have discussed may be invalid when the p-values are dependent. The latter is also true for \emph{Simes' global test}, although it is in fact quite hard to find a dependence structure that invalidates Simes' global test. Simes'  global test rejects $\Hy_{\{1,...,m\}}$ when
there is an $1\leq i \leq m$ for which
$$P_{(i)}\leq \frac{i\alpha}{m},$$
where $P_{(1)}\leq ...\leq P_{(m)}$ indicate the sorted p-values.  In other words, the global test rejects when the event
\begin{equation}  \label{eventSimes}
 \bigcup_{1 \leq i \leq m} \{P_{(i)}\leq \frac{i\alpha}{m} \} 
\end{equation} 
happens.
For example, if the second smallest p-value is at most $(2\alpha)/m$, then this global test always rejects. See Figure \ref{FigSimes} for an illustration.

\begin{figure}[ht!] 
\centering
\begin{tikzpicture} 
\begin{axis}[
    width=14cm,
    height=8cm,
    xlabel={p-value index},
    ylabel={p-value},
    xtick={1,2,...,15},
    ytick={0,0.1,...,0.4},
    ymin=0, ymax=0.4,
    grid=none,
    legend pos=north west,
    title={Illustration of Simes' global test for $m=15$ hypotheses},
    tick label style={font=\small},
    label style={font=\small},
    title style={font=\small},
]

\addplot[
    color=blue,
    mark=*
] coordinates {
    (1, 0.001)
    (2, 0.0023)
    (3, 0.0038)
    (4, 0.04)
    (5, 0.08)
    (6, 0.10)
    (7, 0.18)
    (8, 0.25)
    (9, 0.38)
    (10, 0.42)
    (11, 0.66)
    (12, 0.75)
    (13, 0.76)
    (14, 0.81)
    (15, 0.93)
};

\addplot[
    color=red,
    mark=square*,
    dashed,
    thick
] coordinates {
    (1, 0.0033)
    (2, 0.0067)
    (3, 0.0100)
    (4, 0.0133)
    (5, 0.0167)
    (6, 0.0200)
    (7, 0.0233)
    (8, 0.0267)
    (9, 0.0300)
    (10, 0.0333)
    (11, 0.0367)
    (12, 0.0400)
    (13, 0.0433)
    (14, 0.0467)
    (15, 0.0500)
};

\legend{Sorted p-values, Simes critical values $(i/m)\cdot\alpha$}
\end{axis}
\end{tikzpicture}
\caption{Illustration of the global test using Simes' inequality ($m=15$, $\alpha=0.05$). Simes' critical values are shown in red and the sorted p-values are shown in blue. Note that the Simes' global test rejects the global null hypothesis, since at least one sorted p-value lies below the corresponding critical value.}  \label{FigSimes}
\end{figure}
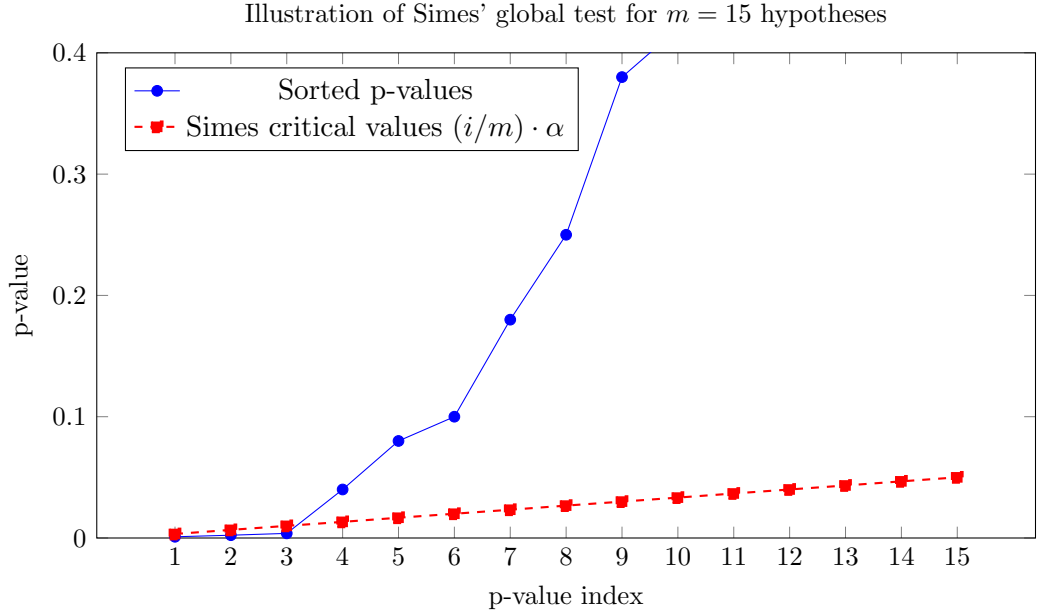

Let  $Q_{1},..., Q_{N}$ be the null p-values, i.e., the  p-values corresponding to the true hypotheses (in the original order) and let $Q_{(1)}\leq...\leq Q_{(N)}$ be the corresponding sorted values.
Note that Simes's global test  has type I error rate at most $\alpha$ when the null p-values satisfy \emph{Simes' inequality}, i.e., when 
\begin{equation}  \label{SimesProbEq}
\pr \big( \bigcup_{1 \leq i \leq N} \{Q_{(i)}\leq \frac{i\alpha}{N} \}\big)\leq \alpha.
\end{equation}

It turns out that when $Q_{(1)},...,Q_{(N)}$ are standard uniform and \emph{independent} of each other, then the probability in \eqref{SimesProbEq} is \emph{exactly} $\alpha$ \citp{sarkar1998some}. This means that Simes' global test then has a type I error rate of exactly $\alpha$. Note that the same is true for \v{S}id\'{a}k's global test and the global test based on Fisher combinations.

\subsubsection{Comparison of the global tests}
Here we provide a brief comparison of the global tests by \v{S}id\'{a}k, Fisher and Simes.
We first compare \v{S}id\'{a}k's global test and Simes' global test. Both tests reject with probability exactly $\alpha$ if the p-values are independent and standard uniform. It is clear that these tests do not always reject at the same time. Hence, it is possible that \v{S}id\'{a}k rejects and Simes not, and vice versa.
Each of these two global tests rejects if one of the sorted p-values is smaller than a corresponding critical value.  Simes compares each sorted p-value $P_{(i)}$ with the critical value $i\alpha/m$, which depends on $i$.  \v{S}id\'{a}k simply compares each sorted p-value $P_{(i)}$ with the critical value $1-(1-\alpha)^{1/m}$, which does not depend on $i$.
Note that the first critical value of Simes is $\alpha/m$, which is smaller than \v{S}id\'{a}k's critical value. However, for $i\geq 2$, Simes' critical value is larger. Roughly speaking, Simes' global test has better power than \v{S}id\'{a}k's when the signals are spread out across multiple hypotheses.

As a final rule of thumb, when many of the hypotheses are false and these signals are roughly equally strong, then the global test based on Fisher combinations can be more powerful than the global test based on Simes. However, Fisher's global test tends to be anti-conservative when there are positive dependencies between the p-values. Simes' global test is usually valid in practice.

\subsubsection{Example of a global test for ANOVA} \label{secglobalanova}
Consider the balanced one-way ANOVA model with $k$ groups and $n$ observations per group. (The approach below only works when each group has the same number of observations.)
Thus, for each $1\leq i \leq k$ and $1\leq j \leq n$, the $j$-th observation in the $i$-th group satisfies 
$$Y_{ij} \sim \theta_i +\epsilon_{ij},$$
where the $\epsilon_{ij}$ have variance $\sigma^2$. Then we can consider the global hypothesis $\Hy_0:\theta_1=...=\theta_k$. Note that this hypothesis is the intersection of all the pairwise hypotheses of the form $\Hy_{ii'}: \theta_i=\theta_{i'}$.
The usual estimator of  $\theta_i-\theta_{i'}$ is $\bar{Y}_{i} - \bar{Y}_{i'}$, with variance $2\sigma^2/n$.

To test the global hypothesis $\Hy_0$, we could use the F-test.
An alternative test statistic is the  maximum of the pairwise $|t|$-statistics
$$|T_{ii'}|= |\bar{Y}_{i} - \bar{Y}_{i'}|/(\hat{\sigma}\sqrt{2/n}),$$
where $\hat{\sigma}$ is the sample standard deviation, i.e., 
$  \hat{\sigma}^2= \sum_{i=1}^k \sum_{j=1}^{n} (Y_{ij}-\bar{Y}_{i})^2/(nk-k).$
Thus, we consider
$$  \sqrt{2}\max_{1\leq i < i'\leq k} |T_{ii'}| = \max_{1\leq i < i'\leq k}\frac{|\bar{Y}_i-\bar{Y}_{i'}|}{\hat{\sigma}/\sqrt{n}},$$
which  equals
\begin{equation} \label{eqttTW}
\frac{\max_{1\leq i \leq k} \bar{Y}_i- \min_{1\leq i \leq k}\bar{Y}_{i}}{\hat{\sigma}/\sqrt{n}}.
\end{equation}

Under $\Hy_0$, this statistic is distributed as the \emph{range} of $k$ independent standard normal random variables divided by an independent $\sqrt{\chi^2_{\nu}/{\nu}}$ random variable, where $\nu=k(n-1)$ \citp[][p.31]{hochberg1987multiple}. This variable is denoted by $Q_{k,\nu}$ and is called the Studentized range variable with parameter $k$ and $\nu$ degrees of freedom \citp[for the CDF, see][p.376]{hochberg1987multiple}.

We reject the global null hypothesis when 
$$   \sqrt{2}\max_{1\leq i < i'\leq k} |T_{ii'}| \geq Q_{k,\nu}^{(1-\alpha)},$$
where $Q_{k,\nu}^{(1-\alpha)}$ is the $(1-\alpha)$-quantile of $Q_{k,\nu}$. This test has better power than the F-test in some situations with sparse signal, and can be computationally convenient in a multiple testing context, as will be discussed in \S\ref{secANOVA-CT}.


\subsection{Familywise error rate (FWER) control} \label{secdefFWER}

Earlier we discussed the probability of incurring one or more false positives, when testing multiple hypotheses. This probability is called the \emph{familywise error rate} (FWER). Thus, if $V$ is the number of false positives (incorrect rejections), then
$$FWER=\pr(V>0).$$
Note that $V$ is the number of true hypotheses that are rejected, so
$$V=|\N\cap \R|,$$
where $\N\subseteq\{1,...,m\}$ is the set of indices of true hypotheses and $\R \subseteq\{1,...,m\}$ is the set of indices of rejected hypotheses.
If $FWER\leq \alpha$, we say that the FWER is controlled at level $\alpha$.
Note that if we use a multiple testing method that controls the FWER, then we know that with probability at least $1-\alpha$, all rejected hypotheses are true discoveries. (This does not mean that conditional on rejecting something, this probability is at least $1-\alpha$. For example, suppose that all hypotheses are true. Then conditional on $|\R|>0$, the probability that $V>0$ is 1.) In \S\S\ref{secbonf}, \ref{secHolm}, \ref{secFWERanova} and several later sections, we discuss examples of FWER controlling methods.

Sometimes a distinction is made between \emph{weak FWER control} and \emph{strong FWER control}. The former means that the FWER is controlled when all hypotheses are true, but not necessarily in other cases. Thus, such methods provide no guarantees, except that likely nothing is rejected when all hypotheses are true. \emph{Strong FWER control} means that the FWER is controlled, regardless of how many hypotheses are false. If we simply say ``FWER control'', then this is understood to mean strong FWER control.

When we test  a single hypothesis, the \emph{power} of a test is simply defined as the probability of rejecting the hypothesis. When we test multiple hypotheses, we need a new notion of power. Let $\N^c=\{1,...,m\}\setminus\N$. By the \emph{power} of a multiple testing method, we will usually mean
$$\mathbb{E}\frac{|\N^c\cap \R|}{|\N^c|},$$
which is the expected fraction of all false hypotheses that is rejected. The power might also be defined as the probability of rejecting at least one false hypothesis, but that is less common in the literature.
We say that procedure A is \emph{uniformly} more powerful than procedure B if A always rejects all the hypotheses that B rejects, and possibly more.

\subsection{Bonferroni} \label{secbonf}
\subsubsection{FWER control with Bonferroni}
We saw that if all hypotheses are true and we only reject the hypotheses $\Hy_i$ with $P_i\leq \alpha/m$, then the probability of any false positives is at most $\alpha$. It can be easily seen that this is still true if some of the hypotheses are false. This leads us to the \emph{Bonferroni method}, which is defined as the procedure that rejects all hypotheses with $\Hy_i$ with $$P_i\leq \alpha/m.$$

\begin{theorem} \label{propBonf}
For Bonferroni's method, $FWER\leq \alpha$. In fact,  $FWER\leq \alpha\pi_0$, where $\pi_0=N/m$.
\end{theorem}

\begin{proof}
If all hypothese are false, the FWER is simply 0.  Suppose there are more than 0 true hypotheses. 
Note that $$FWER = \pr\big(\bigcup_{i\in \N}\{P_i\leq \alpha/m\}\big).$$
As in the proof of Proposition \ref{propBonferroniglobal}, it follows from Boole's inequality that this is at most
$$N\cdot \alpha/m = \alpha\pi_0.$$
\end{proof}

 If $\pi_0$ is near 1, and the null p-values are standard uniform and independent or negatively dependent, then the power of Bonferroni's method is not so far from optimal. However, it follows from Proposition \ref{propBonf} that the FWER of Bonferroni is much smaller than $\alpha$ if $\pi_0$ is much smaller than 1. This means that Bonferroni is then rather suboptima. Thus, if we would know $\pi_0$, we could use that knowledge to make the method more powerful (how?); however, in practice we often do not know $\pi_0$. 
 
 When there are many strong positive dependencies between the p-values, then the FWER of Bonferroni is also much smaller than $\alpha$ (even if all hypotheses are true). The reason is that under strong positive dependence, the distribution of $\min\{P_i:\in \N\}$ is stochastically larger than when the p-values are independent or negatively dependent. In that case, we would still get valid FWER control if we used a rejection threshold that is much larger than $\alpha/m$. Thus, in a sense Bonferroni then also has rather suboptimal power. However, in practice we often do not know the dependence structure of the p-values. 
 
 We now consider two simple examples and discuss them through the lens of Bonferroni's method.
 
  \begin{example}
Consider $\theta\in \reals$ and suppose we wish to test $\Hy_0:\theta=\theta_0$ against the complementary alternative $\Hy_{\text{a}}:\theta\neq\theta_0$. The usual approach is then to consider the one-sided null hypotheses $\Hy_{\text{L}}: \theta\leq \theta_0$ and $\Hy_{\text{R}}: \theta\geq \theta_0$, and corresponding p-values $P_{\text{L}}$ and $P_{\text{R}}$. Then, we  reject $\Hy_0$ if and only if
\begin{equation} \label{eq:2sidedt}
2\cdot\min\{P_{\text{L}}, P_{\text{R}}\} \leq \alpha.
\end{equation}
 A way to see that this test is valid, is to consider Bonferroni's method. We can apply Bonferroni to the p-values $P_{\text{L}}$ and $P_{\text{R}}$. Bonferroni rejects each hypothesis with p-value below $\alpha/2$. Thus Bonferroni rejects a  hypothesis if and only if \eqref{eq:2sidedt} holds. If $\Hy_0$ is true, then $\Hy_{\text{L}}$ and $\Hy_{\text{R}}$ are true and with probability at least $1-\alpha$, Bonferroni rejects nothing. In that case, the test we just defined does not reject either.
 \end{example}

 \begin{example} \label{exampTOST}
 In \S\ref{secequiv}, we considered equivalence testing using the Two One-Sided Tests (TOST) principle. There, to test the null hypothesis $\Hy_0: |\mu^N-\mu^S|\geq \Delta$, we used that $\Hy_0 = \Hy_0^+\cup \Hy_0^-$, which means that $\Hy_0$ is false if and only if $\Hy_0^+$ and  $\Hy_0^-$ are both false.
We can apply Bonferroni to test $\Hy_0^+$ and  $\Hy_0^-$. Note that they cannot be both false, so that we know that $|\N|\leq 1$, i.e., $\pi_0\leq 0.5$. Hence, by Proposition \ref{propBonf}, if we simply reject the hypotheses with p-value below $\alpha$, then the FWER is at most $2\alpha\pi_0\leq \alpha$. Thus, we do not really need a multiple testing correction here. The TOST approach makes use of that fact.
 \end{example}
 
  In case we know that the null p-values  are independent of each other, in Bonferroni's method we can replace $\alpha/m$ by $c$ from \v{S}id\'{a}k's global test. We then get a procedure that controls the FWER and is slightly more powerful than Bonferroni.
 
 Bonferroni's method is quite popular. This is because it is simple and valid regardless of the dependence structure of the p-values.  
 However, Bonferroni's method is uniformly improved by the slightly more complicated Holm  procedure (also known as Bonferroni-Holm), see \S\ref{secHolm}.
 
  When there are positive dependencies between the p-values, it may be more powerful to use a resampling-based multiple testing method such as the maxT method, which will be discussed in \S\ref{secresamplingFWER}.

\subsubsection{Simultaneous confidence intervals}
Suppose we consider $m$ parameters $\theta_1,...,\theta_1\in \mathbb{R}$ and we have a method that provides a $(1-\alpha)100\%$-confidence interval (CI) for each of these parameters, for any chosen $\alpha\in(0,1)$. Note that each CI contains its true parameter with probability (at least) $1-\alpha$, but the probability that they all contain their true parameters is usually smaller than $1-\alpha$, which is often not satisfactory. Instead, we can consider \emph{simultaneous CIs}, i.e., intervals which are such that with probability (at least) $1-\alpha$, all intervals contain their parameters. A highly popular way of constructing such CIs is related to Bonferroni's method: we simply change $\alpha$ to $\alpha/m$, i.e., we make $m$ $(1-\alpha/m)$-CIs. The probability that all CIs contain their parameters is then
$$\pr(\text{all CIs contain their parameter})=$$
$$1- \pr(\text{some CIs don't contain their parameter}) \geq 1- m\cdot \alpha/m=1-\alpha,$$ where we used Boole's inequality again.  If the CIs are independent of each other, we even known that 
$$\pr(\text{all CIs contain their parameter}) \geq (1-\alpha/m)^m>1-\alpha.$$ This allows us to make slightly smaller CIs (use $c$ from \v{S}id\'{a}k's global test instead of $\alpha/m$) and still guarantee that they are are simultaneously valid. However, in many cases we cannot assume that the CIs are independent.

Suppose the hypotheses are of the form $\Hy_i: \theta_i=\theta_i^0$, $1\leq i \leq m$, and we compute $m$ $(1-\alpha/m)$-CIs by inverting $m$ t-tests. Then we can control the FWER by rejecting the hypotheses $\Hy_i$ corresponding to the CI that do not contain $\theta_i^0$. This is the same as applying Bonferroni to the p-values corresponding to the t-tests. In case we use the slightly narrower CIs using \v{S}id\'{a}k's $c$, we get a similar relationship with {S}id\'{a}k's global test. In these examples, we see that there is a simple correspondence between a multiple testing method and a method for constructing simultaneous confidence intervals. However, for most other multiple testing methods, there is no such correspondence.

\subsection{Step-up and step-down methods} \label{secupdown}
In \S\ref{secHolm}, we will consider 
Holm's procedure \citp{holm1979simple}, also known as Bonferroni-Holm.
This is a FWER controlling procedure that is closely related to Bonferroni. Like Bonferroni, Holm controls the FWER when the p-values are valid (Assumption \ref{assvalidp}). Holm rejects all hypotheses that Bonferroni rejects, but may reject additional hypotheses on top of that. This means that Holm is a \emph{uniform improvement} of Bonferroni. We say that procedure A uniformly improves procedure B if A (is still valid and) rejects all hypotheses that B rejects and potentially more. Holm is a \emph{sequential rejection procedure}, because it is a stepwise procedure, that rejects more and more hypotheses until it stops.  

More precisely, Holm belongs to the class of \emph{step-down} procedures. Step-down procedures are often based on p-values, but some step-down procedures are based on test statistics and do not use p-values. For example, the sequential maxT method from \S\ref{secresamplingFWER} is a step-down method which is usually based on test statistics. For now, we discuss procedures based on p-values.
 Step-down methods based on p-values first sort the  p-values from the smallest to the largest. Throughout, we denote the sorted p-values by $p_{(1)}\leq ...\leq p_{(m)}$. 
Starting from the smallest p-value, step-down methods check for each p-value $p_{(i)}$ whether it is  smaller than some critical value $c_i$.
They  stop as soon as a p-value $p_{(i)}$ is reached that satisfies $p_{(i)}>c_i$; Then, all, hypotheses with p-values strictly smaller than $p_{(i)}$ are rejected, see Algorithm \ref{alg:stepdown}.

There are also \emph{step-up} procedures, for example the Benjamini-Hochberg method, which is discussed later. Step-up methods find the largest $i$ for which $p_{(i)}\leq c_i$ and reject all hypotheses with p-values that are at most $p_{(i)}$. An efficient way to find this index $i$ is to start with the largest p-value and work backwards, see Algorithm \ref{alg:stepup}.

\begin{example}
Suppose $m=4$, the sorted p-values $p_{(1)},...,p_{(4)}$ are
$$0.007, \quad 0.024,\quad 0.030,\quad 0.185$$
and the critical values $c_1,..,c_4$ are
$$0.010, \quad 0.020,\quad 0.030,\quad 0.040$$
A step-down method would then only reject the hypothesis corresponding to $p_{(1)}$, since $p_{(2)}$ exceeds it critical value.
A step-up method would reject the hypotheses corresponding to the three smallest p-values, since $p_{(3)}=0.30\leq c_3$.
\end{example}

The terminology \emph{step-down} and \emph{step-up} is potentially confusing. Indeed, a step-down method for example, actually  steps up: it starts with the smallest p-value and then steps to the next one. The terminology often makes sense however when we reformulate the methods in terms of test statistics: a step-down method often starts with the largest test statistic and then steps down to the second-largest, etcetera.

\begin{algorithm}[H]
\caption{Step-down procedure based on p-values and fixed critical values.} \label{alg:stepdown}
\begin{algorithmic}[1]
\REQUIRE P-values $p_1, \dots, p_m$; critical values $c_1 \le \dots \le c_m$

\STATE Order p-values such that 
$p_{(1)} \le \dots \le p_{(m)}$

\STATE $r = 0$

\FOR{$i = 1$ to $m$}
    \IF{$p_{(i)} \le c_i$}
        \STATE $r = i$
    \ELSE
        \STATE \textbf{break}
    \ENDIF
\ENDFOR
\IF{$r=0$}
	\STATE reject nothing and stop
\ELSE 
	 \STATE reject the hypotheses with p-values at most $p_{(r)}$, and stop
\ENDIF
\end{algorithmic}
\end{algorithm}

\begin{algorithm}[H] 
\caption{Step-up  procedure based on p-values and fixed critical values.} \label{alg:stepup}
\begin{algorithmic}[1]
\REQUIRE P-values $p_1, \dots, p_m$; critical values $c_1, \dots, c_m$
\STATE Order p-values such that 
$p_{(1)} \le  \dots \le p_{(m)}$

\STATE $r=0$

\FOR{$i = m$ to $1\quad $} 
    \IF{$p_{(i)} \le c_j$}
        \STATE $r=i$
        \STATE \textbf{break}
    \ENDIF
\ENDFOR
\IF{$r=0$}
	\STATE reject nothing and stop
\ELSE 
	 \STATE reject the hypotheses with p-values at most $p_{(r)}$, and stop
\ENDIF
\end{algorithmic}
\end{algorithm}

\subsection{FWER control with Holm's method} \label{secHolm}
The first step of the step-down method of Holm is identical to Bonferroni, i.e., all p-values are compared with $\alpha/m$. If no hypotheses are rejected in the first step, the method stops. If $r\geq 1$ hypotheses are rejected in the first step, then the method moves on to the second step, where all remaining p-values are compared with the threshold $\alpha/(m-r)$. We will refer to this as \emph{step 2} in the proof below. If this leads to no additional rejections then stop. Otherwise, move to step 3, where all remaining p-values are compared with $\alpha/(m-r)$, where $r$ is the total number of hypotheses that have been rejected in the previous steps. The method continues like this until a step is reached where no additional hypotheses are rejected, or when all hypotheses have been rejected.

A different way to formulate Holm's method is the following. 
\begin{enumerate}
\item Consider the sorted p-values $p_{(1)}\leq ...\leq p_{(m)}$ and find the smallest $1\leq k \leq m$ for which
$$p_{(k)}>\frac{\alpha}{m-k+1},$$
if there is such a $k$, otherwise take $k=m+1$. 
\item Then reject the hypotheses corresponding to the $p_{(i)}$ with $i<k$.
\end{enumerate}
Note that this  method is simply the step-down method as defined in \S\ref{secupdown}, with critical values $$c_j=\frac{\alpha}{m-j+1}.$$

\begin{theorem} \label{thmHolm}
Holm's method controls the FWER, i.e., $\pr(V>0)\leq\alpha.$
\end{theorem}

\begin{proof}
If all hypotheses are false, the FWER is simply 0. Hence, suppose that there are one or more true hypotheses. Write $N=|\N|$. Consider the event
$$\E= \{\forall j\in \N: P_j> \alpha/N\}.$$
Note that $$1-\pr(\E) =\pr(\E^c) = \pr\{\exists j\in \N: P_j\leq \alpha/N\}. $$
By Boole's inequality, this is at most
$$ \sum_{j\in\N}  \pr\{ P_j\leq \alpha/N \}\leq N \cdot \alpha/N=\alpha.$$
Thus, $\pr(\E)\geq 1-\alpha$.

Finally, we show that if $\E$ happens, there are no false positives.
Indeed, note that if $\E$ happens, then no true hypotheses are rejected in Holm's first step.
Hence, in Step 2 (defined above), the p-value threshold is at least $\alpha/N$. Hence, since  $\E$  holds, there are no false positives in Step 2 either. 
In general, if $\E$ holds and no true hypotheses were rejected in Step $i$, then no true hypotheses are rejected in Step $i+1$ either. Thus, under $\E$, there are no false positives. We conclude that with probability at least $1-\alpha$, there are no false positives. In other words, the method controls the FWER. 

In \S\ref{secCT} we will prove in a different way that Holm controls the FWER, namely by showing that Holm coincides with a closed testing method.
\end{proof}

Other examples of step-down methods are the maxT method and the Romano-Wolf method for FDX control, which will be discussed later. The method  in \S\ref{secFWERanova} is also a step-down method (if we use that term in a somewhat general sense).    Holm is an example of a  \emph{closed testing method}, as will also be explained later. Clearly, Holm is uniformly more powerful than Bonferroni. The gain in power is the largest when $\pi_0$ is far below 1 and the signals are strong. Indeed, in that case many hypotheses tend to be rejected in the first step of Holm, which leads to a substantial increase in the critical value, compared to the initial critical value $\alpha/m.$

We mentioned that if there are many strong positive dependencies between the p-values, then Bonferroni is rather conservative. The same is true about Holm. When there are such dependencies, the FWER of Holm is (much) smaller than $\alpha$. When there are no dependencies, then Holm has fairly good power --- among the class of methods that take a vector of $m$ p-values as input and use no further information. Of course, in many cases, we do not know a priori what the dependence structure of the p-values is. If we can make no assumptions on the dependence structure, then  Holm is  often a good choice --- among the class of methods that take a vector of $m$ p-values as input and are valid under general dependence and use no further information. One may wonder whether Holm is \emph{admissible}. By ``admissible'' we mean that the procedure cannot be uniformly improved without invalidating the method.\footnote{Admissibility is only defined \emph{given the assumptions}: if we make additional assumptions, then we can potentially uniformly improve a method, even if we could not under the original assumptions.}
There is no known proof that Holm is admissible  but it has been conjectured that it is. (It can be proved that Holm admissible in a somewhat different sense \citp{goeman2021only}, but we will not consider that concept here.)

\subsection{FWER control for ANOVA}\label{secFWERanova}
Here we consider an important but non-straightforward FWER controlling method for one-way ANOVA models, for example the 
balanced one-way ANOVA model from \S\ref{secglobalanova}. The procedure is often referred to as the Tukey-Welsch procedure (see \citealp[][p.69]{hochberg1987multiple}, \citealp{fay2022statistical}). (We do not prove here that this procedure controls the FWER.) We  could also use Holm (applied to the p-values from the pairwise comparisons) for example, but it will have low power in this situation, due to dependencies between the pairwise comparisons.

For each subset $\mathcal{A}$ of $\{1,...,k\}$, define $\Hy_{\A}$ as the hypothesis that all $\theta_i$ with $i\in \A$ are equal. (In later sections that are not about ANOVA, the notation $\Hy_{\A}$ will have a different meaning.) The goal here is to test $\Hy_{\A}$ for every $\mathcal{A}\subseteq \{1,...,m\}$ with $\A\geq 2$.

The procedure proceeds in a particular type of ``step-down'' manner. We define $\alpha_p=\alpha$ for $p\in \{k,k-1\}$, and $\alpha_p = 1-(1-\alpha)^{p/k}$ for $2\leq p\leq k-2$. Note that as $p$ decreases, $(1-\alpha)^{p/k}$ increases, so $\alpha_p$ decreases. Below we give a specific version of the procedure for the one-way balanced ANOVA model. For certain other ANOVA models and test statistics (e.g. F-statistics) the method below also provides FWER control (with the same definition of $\alpha_p$, but the critical values should be adjusted to the test statistics used). We now give the steps of the Tukey-Welsch procedure.

\begin{enumerate}
\item First test $\Hy_{\{1,...,k\}}$ at level $\alpha$ using the test from \S\ref{secglobalanova}. If it is not rejected, stop. Otherwise, continue.
\item 
Test all hypotheses $\Hy_{\A}$ with $\A\subset\{1,...,k\}$ and $|\A|=m-1$ using the test statistic \eqref{eqttTW} from \S\ref{secglobalanova} applied to the data corresponding to $|\A|$, 
$$\frac{\max_{i\in \A} \bar{Y}_i- \min_{i\in \A}\bar{Y}_{i'}}{\hat{\sigma}/\sqrt{n}}.$$
where $\hat{\sigma}$ still denotes the sample SD based on the \emph{full} dataset, i.e., 
$$\hat{\sigma}^2=\sum_{i=1}^k \sum_{j=1}^{n} (Y_{ij}-\bar{Y}_{i})^2/(nk-k),$$
 and using critical value $Q_{|\A|,\nu}^{(1-\alpha_p)}$, with $\nu=k(n-1)$ and $p=|\A|$. If they are all not rejected, stop. Otherwise continue.
\item For all $\A$ that were rejected in the previous step,  test all hypotheses $\Hy_{\B}$ with $\B\subset\A$ and $|\B|=m-2$ using the test from \S\ref{secglobalanova} (applied to the data corresponding to $|\B|$, but using $\hat{\sigma}$ based on the full dataset, and using critical value $Q_{|\B|,\nu}^{(1-\alpha_p)}$, with $\nu=k(n-1)$ and $p=|\B|$). If nothing is rejected in this step, stop. Otherwise continue.
\item Consider all sets $\B$ that were rejected in the previous step and consider their subsets of size $m-3$. Keep on going like this, until no further subsets remain to be tested.
\end{enumerate}

 The method is implemented in the function \emph{tukeyWelsch()} of the R package \emph{asht} on CRAN. Some other methods related to Tukey-Welsch are discussed in \S\ref{secANOVA-CT}.

\subsection{k-FWER control} \label{seckFWER}
When we test many hypotheses, often we are willing to have a somewhat higher risk of false positives if this leads to much better power. (Recall that by power, we usually mean the expected fraction of the false hypotheses that is rejected.) One way to increase the power, is to control the \emph{k-familywise error rate} (k-FWER), where $1\leq k\leq m$. The k-FWER is defined as follows:
$$k\text{-FWER}=\pr(\text{reject at least } k \text{ true hypothess}).$$
Formulated differently, 
$$k\text{-FWER}=\pr(V\geq k).$$
Note that for $k=1$, this is the usual FWER, Thus, the k-FWER is a generalization of the FWER.
\emph{Controlling the k-FWER} means ensuring that
$ k\text{-FWER}\leq \alpha.$
A simple k-FWER controlling procedure is the following generalization of Bonferroni:  reject all hypotheses $\Hy_j$ with p-values $P_j$ that satisfy
\begin{equation}\label{eq:thresholdkFWER}
P_j\leq k\alpha/m.
\end{equation}

\begin{theorem} \label{thmkFWERss}
The procedure that rejects all hypotheses $\Hy_j$ with p-values satisfying inequality \eqref{eq:thresholdkFWER} controls the k-FWER, i.e., $\pr(V\geq k)\leq\alpha$.
\end{theorem}

\begin{proof}
Note that $$\mathbb{E}(V) =\mathbb{E}\sum_{j\in \N}\mathbbm{1}_{\{P_j\leq k\alpha/m\}} \leq N k\alpha/m$$
by Boole's inequality.

By Markov's inequality
$\mathbb{E}(V)\geq k\pr(V\geq k) $, so that 
$\pr(V\geq k) \leq (N k\alpha/m)/k = \alpha N/m\leq \alpha$, as was to be shown.
\end{proof}
In the above proof, Markov's inequality is used. This inequality may be  quite conservative; however, it can be shown that there are settings where (Assumption \ref{assvalidp} is satisfied and) the k-FWER is exactly $\alpha$ with this method \citp[][Thm 2.1(ii)]{lehmann2005generalizations}.

A different k-FWER method is defined as follows:
\begin{enumerate}
\item Reject all hypotheses that Bonferroni (or Bonferroni-Holm) rejects.
\item Consider the p-values that were not rejected in step one, and additionally reject the hypotheses corresponding to the smallest $k-1$ of these p-values (or reject all $m$ hypotheses if there were fewer than $k-1$ left).
\end{enumerate}
This method also controls the k-FWER, but it potentially rejects hypotheses with large p-values, so this method is not recommended.

Clearly, for $k=1$, the method from Theorem \ref{thmkFWERss} reduces to Bonferroni. Since Bonferroni is uniformly improved by Holm, one may wonder whether  the k-FWER method above can likewise be improved. This is indeed the case. For every $1\leq i \leq m$, consider the following critical value:

\begin{equation} \label{cvkFWERstepdown}
c_i=
 \begin{cases}
\frac{k\alpha}{m}& i\leq k, \\
\frac{k\alpha}{m+k-i} & i>k,
\end{cases}
\end{equation}
These critical values are a generalization of the Holm-critical values. Indeed, for $k=1$, we get  Holm's critical values.

\begin{theorem} \label{thmkFWERstepdown}
The step-down method with the critical values in \eqref{cvkFWERstepdown} controls the k-FWER.
\end{theorem}
\begin{proof}
If all hypotheses are false, the k-FWER is simply 0. Hence, suppose that there are one or more true hypotheses.
Consider the event
$$\E = \Big\{  |\{i\in \N: P_i\leq k\alpha/N\}|\leq k-1   \Big\}.$$

Note that by Markov's inequality,
\begin{align*}
\pr(\E^c)  = &  \pr\big(|\{i\in \N: P_i\leq k\alpha/N\}| \geq k\big) \\
\leq & k^{-1} \mathbb{E}\big(|\{i\in \N: P_i\leq k\alpha/N\}|\big)\\
\leq & k^{-1}N k\alpha/N =\alpha.
\end{align*}

Throughout the remainder, suppose $\E$ happens.
Note that the method compares the smallest $k$ p-values with $k\alpha/m$. 
If $p_{(k)}>k\alpha/m$, there are fewer than $k$ rejections in total.
If $p_{(k)} \leq k\alpha/m$,  then there are at most $k-1$ null p-values among $p_{(1)},...,p_{(k)}$ (since $\E$ holds). Hence, $N\leq m-1$. 
If $p_{(k+1)}>c_{k+1}$, the method does not go further.
Otherwise, note that $c_{k+1}\leq k\alpha/N$, so that there are at most $k-1$ null p-values among $p_{(1)},...,p_{(k+1)}$. Hence, $N\leq m-2$.
If $p_{(k+2)}>c_{k+2}$, the method does not go further.
Otherwise, note that $c_{k+2}\leq k\alpha/N$, so that there are at most $k-1$ null p-values among $p_{(1)},...,p_{(k+2)}$. 
We can continue with this reasoning, which means  that the number of false positives cannot be larger than $k-1$ under $\E$. Since $\pr(\E)\geq 1-\alpha$, this finishes the proof.

\end{proof}

The downside of k-FWER methods compared to FWER control, is that we no longer have $(1-\alpha)100\%$ confidence that there are no false positives. For example, suppose $k=4$ and we reject 5 hypotheses. Then we have no guarantee that it is likely that most of these findings are true discoveries. 

Note that if we do not require FWER control but k-FWER control with $k\geq 2$, then we could in principle always simply reject $k-1$ arbitrarily chosen hypotheses (even if all p-values are very large) and still control the k-FWER.
Thus we can uniformly improve the method defined by \eqref{eq:thresholdkFWER} by rejecting some additional hypotheses if that method rejects fewer than $k-1$ hypotheses.
However, if we reject $k-1$ (or fewer) hypotheses with a k-FWER method, then that is meaningless.

 At the beginning of this section (\S \ref{seckFWER}), we discussed that k-FWER control can be useful in case we want to have higher power. However, instead of  k-FWER control, we can then also use a different approach, such as false discovery exceedance (FDX) control (see \S\ref{secFDX}) or false discovery rate control (see \S \ref{seccriteriaFDP}, \S\ref{secFDR}, \S\ref{secknockoffs}). Such approaches are often more appealing than k-FWER control.  Some FDX controlling methods  rely on iteratively using k-FWER mehods (see \S\ref{secRW}).  

\subsection{Adjusted p-values} \label{secadj}
Consider hypotheses $\Hy_1,...,\Hy_m$ and some multiple testing procedure (MTP) which, given $\alpha\in(0,1)$, decides for each  hypothesis  whether it is rejected or retained. 
For every  $1\leq i \leq m$, the \emph{adjusted p-value} $p^{\text{adj}}_i$ is defined to be the smallest value $\alpha\in[0,1]$ for which the MTP still rejects $\Hy_i$. More precisely, 
$$p^{\text{adj}}_i  = \min\{\alpha: \Hy_i \text{ is rejected by the MTP}\}\wedge 1 ,$$
where $a\wedge b$ denotes the minimum of $a$ and $b$.
For example, note that for Bonferroni, $p^{\text{adj}}_i = \min\{p_i\cdot m,1\}$.
Clearly, $p_i\cdot m$ can be larger than 1. This illustrates why authors often take  the minimum with 1, although this does not really matter, since we do not reject anyway, if $\min\{\alpha: \Hy_i \text{ is rejected by the MTP}\}$ is larger than or equal to 1.

To obtain the adjusted p-values for step-down or step-up methods based on p-values, a general algorithm can be used. Consider a step-up or step-down method with critical values $c_1,...,c_m$. The critical values are typically of the form $\alpha/a_i$, where $a_i$ is an adjustment factor which does not depend on $\alpha$.
Thus, $a_i$ is the ratio between $\alpha$ and the $i$-th critical value.
The following algorithm can be used for obtaining the adjusted p-values \citp{goeman2014multiple}:

\begin{enumerate}
\item Obtain the sorted p-values $p_{(1)},...,p_{(m)}.$
\item Multiply each sorted p-values $p_{(i)}$ by the adjustment factor $a_i$.
\item  If this multiplication violates the ordering, adjust the values. For step-down methods, take $p^{\text{adj}}_{(i)} = \max_{1\leq j \leq i}a_jp_{(j)}$. For step-up methods, take $p^{\text{adj}}_{(i)} = \min_{i\leq j \leq m}a_jp_{(j)}$.
\item Set $p^{\text{adj}}_{(i)}=\min(p^{\text{adj}}_{(i)},1) $ for all $i$.
\end{enumerate}

For example, for Holm, which is a step-down method,  the critical values are $c_i=\alpha/(m+1-i)$.
Thus, to get the adjusted p-values for Holm, we first compute the values $p_{(i)} \alpha c_i^{-1} = p_{(i)}(m+1-i) $ and then possibly adjust them upwards to enforce monotonicity. 

The R function \emph{p.adjust()} outputs adjusted p-values for some FWER and FDR controlling methods, including Holm, Hommel  and Benjamini-Hochberg (the latter two will be discussed later). 
If one has computed adjusted p-values, then for any value of $\alpha$, one can immediately see which hypotheses are rejected by the MTP, namely the hypotheses with adjusted p-values below $\alpha$. Since the function \emph{p.adjust()} outputs adjusted p-values, it is not necessary to provide $\alpha$ as an argument and nevertheless, based on the adjusted p-values, one can easily  see which hypotheses are rejected for any value of $\alpha$.

There are also some important caveats when using adjusted p-values. The first one is that just like with regular p-values, it is required to specify $\alpha$  independently from the data, if we want FWER control (or e.g.  FDR control, which is discussed later). Choosing $\alpha$ based on the data or based on the adjusted p-values, \emph{may} invalidate a MTP. Also note that, as is clear from the algorithm above, the adjusted p-value of a hypothesis may depend on the other p-values. Thus, adjusted p-values have no independent meaning.




\subsection{Application to property sales data} \label{secappsales}
To illustrate the methods, we apply them  to a dataset on  residential property sales in Ames, Iowa
from 2006 to 2010 \citp{cock2011ames}. We use the cleaned data which are available through the R package \verb|AmesHousing| \citp{ameshousing}. The dataset contains 2930 rows, which correspond to sold individual properties. The data  contain the sale price and various variables that might affect it, such as the lot area, the number of bathrooms above grade, the build year and the roof style.

We removed variables that were near-constant and standardized all quantitative variables, including the sale price. We fitted a standard least squares linear regression model, where sale price was regressed on all other variables, including the dummy-encoded categorical variables. There were 31 numerical predictors and 248 dummy variables, so in total   279 predictors. For each predictor, we considered the hypothesis $\Hy_i$ that this predictor has coefficient 0. We computed the 279 two-sided p-values using the usual t-statistics. The smallest p-values are shown in Figure \ref{fig:50pvalues}.

\begin{figure}[ht!]
    \centering
    \includegraphics[width=0.85\textwidth]{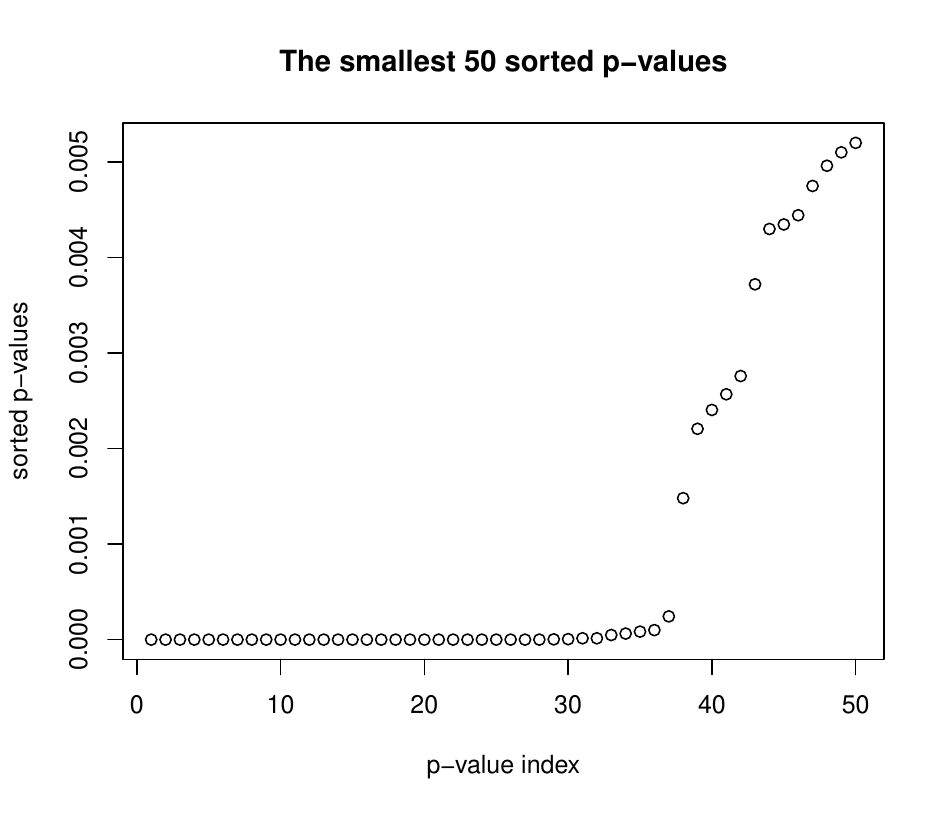}
    \caption{The sorted values of the smallest 50 p-values among all 279 p-values for the regression model predicting sale prices.}
    \label{fig:50pvalues}
\end{figure}

We take $\alpha=0.05$. Recall that Bonferroni rejects all hypotheses with p-values below $\alpha/m=\alpha/279\approx 0.0001792$. In this case, Bonferroni, rejects 36 hypotheses. Indeed, the 36-th smallest p-value, $p_{(36)}$, is 
$\approx 0.0000997$ and $p_{(37)}\approx 0.0002432$; clearly $0.0000997\leq \alpha/m$ and $0.0002432>\alpha/m$.

We now consider Holm. In the first step Holm rejects the same 36 hypotheses as Bonferroni, by definition. In the next step, Holm increases the rejection threshold to $\alpha/(m-36)=0.05/243\approx 0.0002057613$. This is still smaller than $p_{(37)}$, hence Holm rejects no further hypotheses in this case.

Next, we apply the k-FWER methods from \S\ref{seckFWER}. We take $k=5$ and apply the single-step method from Theorem
\ref{thmkFWERss}. This leads to 37 rejected hypotheses. Note that in this particular example, the result that  Bonferroni provides is more attractive, since it rejects one hypothesis less but it provides much stronger guarantees. 
The step-down k-FWER method (from Theorem \ref{thmkFWERstepdown}) also rejects 37 hypotheses.

We now focus on the p-values for the 31 numerical predictors, see Figure \ref{fig:31pvalues}.
With $\alpha=0.05$ as before, Bonferroni  and Holm both reject 10 hypotheses. The single-step 5-FWER method rejects 12 hypotheses and the sequential 5-FWER method (from Theorem \ref{thmkFWERstepdown}) rejects 14 hypotheses. Note that the 5-FWER methods again provide no additional information compared to Holm, in this case.

\begin{figure}[ht!]
    \centering
    \includegraphics[width=0.85\textwidth]{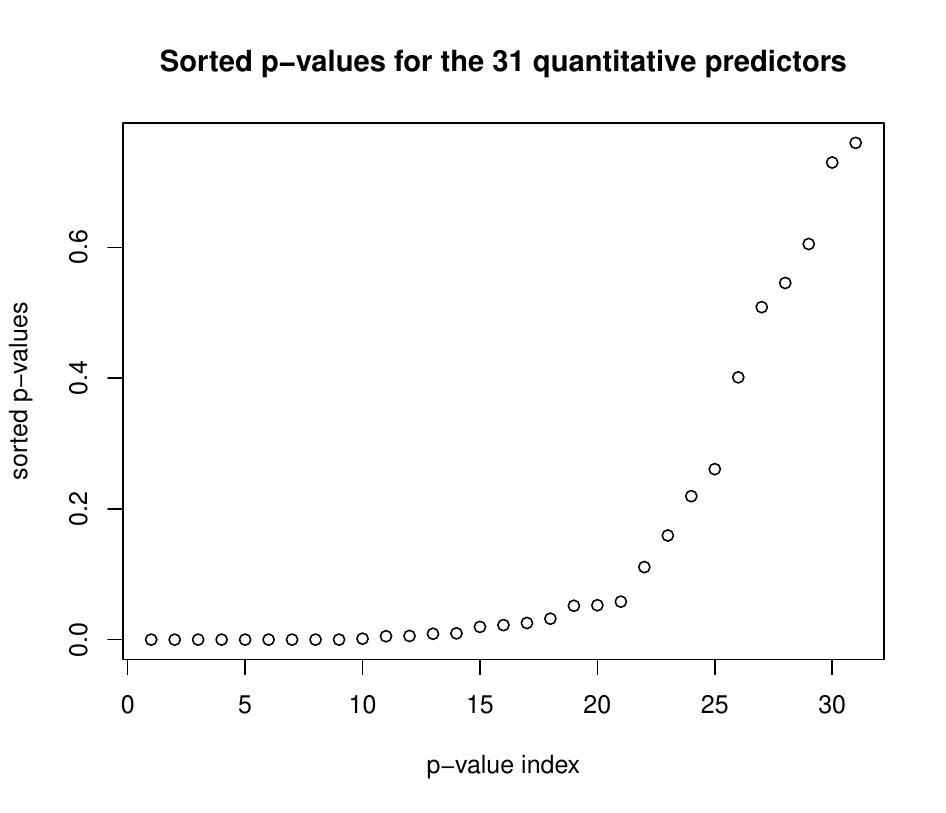}
    \caption{The sorted p-values corresponding to the numerical predictors of the regression model predicting sale prices.}
    \label{fig:31pvalues}
\end{figure}

\subsection{Exercises}
\begin{enumerate}

\item Consider 5 hypotheses and corresponding  p-values  $0.025$, $0.045$, $0.010$, $0.012$ and $0.245$. Take $\alpha=0.05$. How many hypotheses are rejected by Holm's method? And by the k-FWER method from Theorem \ref{thmkFWERstepdown}, with $k=2$?

\item
Consider Bonferroni's FWER controlling method applied to $m\geq 2$ valid p-values, which are \emph{independent} of each other. Prove that the FWER is then strictly smaller than $\alpha$, regardless of how many hypotheses are true.

\item\label{ex:bonsidak} Let $\alpha=0.05$ and (using a calculator) compute the ratio between the critical value  of Bonferroni's global test ($\alpha$/m) and of  \v{S}id\'{a}k's global test ($1-(1-\alpha)^{1/m}$) for $m\in\{2,3,4,5\}$. What (simple) pattern do you notice?

\item $\bigstar$ In exercise \ref{ex:bonsidak}, we considered the ratio between the critical values of Bonferroni and \v{S}id\'{a}k. Prove that as $m\rightarrow\infty$, this ratio converges to  $-\alpha/\log(1-\alpha)$.

\item Consider $m\geq 2$ hypotheses $\Hy_1,...,\Hy_m$ and $\alpha\in(0,1)$. Suppose all hypotheses are true and all p-values have the uniform distribution on $(0,1]$. Show that under a certain dependence structure of the p-values, the FWER of Bonferroni is exactly $\alpha$.

\item Assume  the p-values $P_1,...,P_m$ are valid (Assumption \ref{assvalidp}). Prove that Bonferroni's FWER-controlling method ensures that  the expected value of the number of false positives is at most $\alpha$, i.e.,  that $\mathbb{E}(V)\leq \alpha$. Show that this does not necessarily hold for Holm's method.

\item In \S\ref{secadj} a procedure is provided for computing adjusted p-values for step-down (and step-up) methods based on p-values and  fixed critical constants. For the case of step-down methods, prove that this algorithm indeed provides the adjusted p-values $p_{(i)}=\min\{\alpha\in[0,1]: \Hy_i \text{ is rejected by the MTP}\}\wedge 1.$

\end{enumerate}

\begin{sols}

\subsection{Solutions}
\begin{enumerate}

\item
The sorted p-values are $0.010$, $0.012$,  $0.025$, $0.045$ and $0.245$ and Holm's critical values are $c_1=\alpha/5=0.01$, $c_2=\alpha/4=0.0125$, $c_3=\alpha/3=0.0166...$, $c_4=\alpha/2=0.025$, $c_5=0.05$. Clearly $P_{(i)}\leq c_i$ for  $1\leq i \leq 2$, but not for $i=3$. Hence, 2 hypotheses are rejected.

The 2-FWER method, which is also a step-down method,  always rejects at least as many hypotheses as Holm.
For the 2-FWER method, the critical values are 

\begin{equation}
c_i=
 \begin{cases}
\frac{2\alpha}{m}& i\leq 2, \\
\frac{2\alpha}{m+2-i} & i>2,
\end{cases}
\end{equation}
Thus, $c_3=2\alpha/4=0.025$, $c_4=2\alpha/3=0.0333...$ and $c_5=\alpha$. Note that $P_{(3)}=0.025\leq c_3.$
Note that $P_{(4)}=0.045>c_4$. Thus,  the 2-FWER method rejects  3 hypotheses.

\item The FWER is $$\pr\big(\bigcup_{i\in \N}\{P_i\leq \alpha/m\}\big).$$
Note that this is \emph{strictly} smaller than 
$\sum_{i\in \N} \pr(P_i\leq \alpha/m),$ since the events $\{ P_i\leq \alpha/m \}$, $i \in \N$ are overlapping, due to the assumed independence. This sum equals $\alpha(N/m)\leq \alpha$. Thus, $FWER<\alpha$.

\item
Approximately 1, 0.987340, 0.983143,  0.981050 and  0.979795. The relative difference becomes larger as $m$ increases. (Apparently, as $m$ increases, Bonferroni ``pays'' more for allowing for negative dependencies. It can be shown that the ratio decreases to $-\alpha/\log(1-\alpha)$ as $m\rightarrow \infty$. This is approximately 0.974786.)

\item 
Note that the ratio is 
$$\frac{\alpha/m}{1-(1-\alpha)^{1/m}}=$$
$$\frac{\alpha/m}{1-\exp((1/m)\log(1-\alpha))}=$$

$$
\frac{\alpha/m}{
1  - \big(    \frac{(1/m\log(1-\alpha))^0}{0!}  + \frac{(1/m\log(1-\alpha))^1}{1!} +  \frac{(1/m\log(1-\alpha))^2}{2!}...\big)  }=$$
$$  \frac{\alpha}{
m  - \big(    m  +   \frac{\log(1-\alpha)^1}{1!} +  \frac{1/m(\log(1-\alpha))^2}{2!}...\big)  } =$$

$$  \frac{\alpha}{
  -      \frac{\log(1-\alpha)^1}{1!} -  \frac{1/m(\log(1-\alpha))^2}{2!}...  }=$$
  $$  \frac{\alpha}{
  -      \frac{\log(1-\alpha)^1}{1!} +o(m)  }.$$
  As $m\rightarrow\infty$, this converges to 
  $\alpha/(-\log(1-\alpha))$.

\item
We provide two possible solutions. (For the second approach, no complete proof is provided here.)

\emph{First solution:}
Let $i_1,...,i_m$ be a random permutation of $1,...,m$ (uniformly distributed on all possible permutations) and for every $1\leq j \leq m$ let $P_{i_j}$ be uniformly distributed on $((j-1)/m, j/m]$. Clearly, conditionally on $i_1,...,i_m$ these p-values do not have the standard uniform distribution, but marginally (i.e. unconditionally) they do. To see this, consider any $1\leq i \leq m$ and note that for every $1\leq j \leq m$, with probability $1/m$, the conditional distribution of $P_i$ is uniform on $((j-1)/m, j/m]$; hence marginally, $P_i$ is uniform on $(0,1]$.

We now show that for this dependence structure, the FWER of Bonferroni is $\alpha$.
It suffices to show that conditionally on $i_1,...,i_m$, the probability that Bonferroni rejects something is $\alpha$.
To see this, note that conditionally on $i_1,...,i_m$,  Bonferroni will reject at most one hypothesis, which is $\Hy_{i_1}$, and this happens when $P_{i_1}\leq \alpha/m$, which holds with probability $\alpha$.

\emph{Second solution:}
Let $P_1$ be uniformly distributed on $[0,1]$ and for every $2\leq i \leq m$ let $P_i=P_1+(i-1)/m - \lfloor P_1+(i-1)/m \rfloor$ (where $\lfloor P_1+(i-1)/m \rfloor$ is the ``integer part'' of  $P_1+(i-1)/m$).
All the p-values are then standard uniform, and neighbouring p-values  are a distance $1/m$ apart.

It can be checked that for this dependence structure of the p-values, the minimum of the p-values is uniform on $[0,1/m]$, so it is below $\alpha/m$ with probability  $\alpha$, so  the FWER of Bonferroni is then also exactly $\alpha$.

\item 
With Bonferroni's method, it holds that 
$$\mathbb{E}(V) = \mathbb{E}(\sum_{i\in \N}\mathbbm{1}_{\{P_i\leq \alpha m\}})= \sum_{i\in \N} \mathbb{E}\mathbbm{1}_{\{P_i\leq \alpha m\}}=\sum_{i\in \N} \pr\{P_i\leq \alpha m\} \leq N\alpha m\leq \alpha.$$

Note that if all hypotheses are true and all p-values are standard uniform, we get an equality: $\mathbb{E}(V) = \alpha$. In that setting (i.e. $m=N$ and all p-values standard uniform) Holm always rejects the hypotheses that Bonferroni rejects, and sometimes more (since Holm uses larger critical values $c_i$ than Bonferroni). Hence,  it then holds that $\mathbb{E}(V) > \alpha$ for Holm.

\item
It suffices to show that if $\Hy_{(i)}$ is rejected, then $p^{\text{adj}}_{(i)}\leq \alpha$, and vice versa. Indeed, that implies that $p^{\text{adj}}_{(i)}= \alpha$ is the borderline case where  $\Hy_{(i)} $ is barely rejected --- in other words, $p^{\text{adj}}_{(i)}$ is the smallest $\alpha$ for which the procedure rejects $\Hy_{(i)} $.

Note that $\Hy_{(i)}$ is rejected if and only if $p_{(j)}\leq c_j$ for all $j\leq i$, which is equivalent to 
$\max_{1\leq j \leq i}c_j^{-1}p_{(j)}\leq 1$, which is equivalent to $\max_{1\leq j \leq i}a_i p_{(j)}\leq \alpha $, which is equivalent to $p^{\text{adj}}_{(i)}\leq \alpha.$ This finishes the proof.

\end{enumerate}

\end{sols}

\newpage

\section{Closed testing for FWER control} \label{secCT}
In \S\ref{secmtbasics} we discussed some simple methods that control the FWER. One may wonder how we can construct a FWER controlling method that is as powerful as possible. However, one should realize that given some dataset, it is often not known which method will be most powerful. Moreover, not all methods require the same assumptions.
Thus, it often depends on the situation which FWER method can best be used.
Hence, it is useful to know several methods, each of which can be a good choice in certain specific situations.

\emph{Closed testing} is a general recipe for constructing multiple testing procedures that control the FWER or provide confidence on false discovery proportions. Moreover, it turns out that all admissible FWER controlling methods are closed testing procedures (CTPs); recall that ``admissible'' means that the procedure cannot be uniformly improved without violating FWER control. This follows from the fact that a multiple testing procedure that provides FWER control,  is a CTP or can be uniformly improved by one (see additional exercises).
In \S\ref{secCTPfdp}, CTPs also play an important role, although there we do not focus on the FWER, but on false discovery proportions.

The ingredients required for creating  a CTP are so-called \emph{local tests}. Once we have defined the local tests, the CTP that controls the FWER follows from them.
Hence, we will first discuss local tests and then closed testing.

\subsection{Local tests}
Let $\subs$ be the set of all nonempty subsets of $\{1,...,m\}$. Note that $|\subs| = 2^m-1$.
When we consider hypotheses $\Hy_1,...,\Hy_m$, we can also consider the corresponding \emph{intersection} hypotheses, which are all hypotheses $\Hy_{\I}$ with $\I\in \subs$, where we define
$$\Hy_{\I} = \bigcap_{i\in \I}\Hy_i.$$
Thus, $\Hy_{\I}$ is true if and only if  all hypotheses $\Hy_i$ with $i\in \I$ are true.  The set of all intersection hypotheses is sometimes called the \emph{closure} of the set of hypotheses $\Hy_1,...,\Hy_m$. This is because this set is closed under intersections in the sense that if $\Hy_{\I}$ and $\Hy_{\J}$ (with $\I$, $\J\in \subs$) are intersection hypotheses, then their intersection is also an  intersection hypothesis within this family.
The hypotheses $\Hy_1,...,\Hy_m$ that we ``start from'' are called the \emph{elementary} hypotheses. (Sometimes a condition is placed on the elementary hypotheses, e.g. that no elementary hypothesis is implied by another one.) Note that if $\I,\J\in \subs$, then perhaps confusingly,
$$\I\subseteq\J \Rightarrow  \Hy_{\I} \supseteq \Hy_{\J}.$$
For example, if $\I=\{1,2\}$ and $\J=\{1,2,3\}$, then $\Hy_{\J}$ corresponds to the darkest part of the Venn diagram in Figure \ref{fig:Venn}, while $\Hy_{\I}$ corresponds to a larger area.

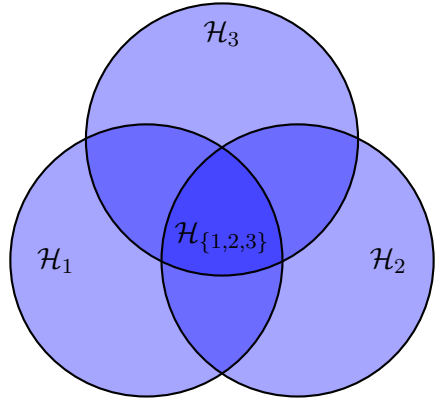
\begin{figure}[h]
\centering
\begin{tikzpicture}

\fill[blue, opacity=0.35] (-1,0) circle (1.8);
\fill[blue, opacity=0.35] ( 1,0) circle (1.8);
\fill[blue, opacity=0.35] ( 0,1.6) circle (1.8);

\draw[thick] (-1,0) circle (1.8);
\draw[thick] ( 1,0) circle (1.8);
\draw[thick] ( 0,1.6) circle (1.8);

\node at (-2.2,0) {$\mathcal{H}_1$};
\node at ( 2.2,0) {$\mathcal{H}_2$};
\node at ( 0,3) {$\mathcal{H}_3$};
\node at ( 0,0.3) {$\mathcal{H}_{\{1,2,3\}}$};

\end{tikzpicture}
\caption{A hypothesis is a set of distributions. This Venn diagram represents three hypotheses $\Hy_1$, $\Hy_2$, $\Hy_3$ and the corresponding intersection hypotheses. 
The intersection of all three hypotheses, $\Hy_{\{1,2,3\}}$, corresponds to the darkest part of the diagram.} \label{fig:Venn}
\end{figure}

We have already seen some examples of intersection hypotheses. For instance, if  we consider hypotheses $\Hy_1:\theta\leq \theta_0$ and $\Hy_2:\theta\geq \theta_0$, then the corresponding intersection hypothesis is $\Hy_{\{1,2\}}:\theta= \theta_0$. As another example, we considered the global null hypothesis corresponding to a collection of hypotheses $\Hy_1,...,\Hy_m$. This is the hypothesis that all $m$ hypotheses are true, i.e., it is the hypothesis $\Hy_{\{1,...,m\}}$.

A \emph{local test} is simply a hypotheses test  for an intersection hypothesis. We will assume throughout that each local test has level at most $\alpha$.
Of course the global hypothesis $\Hy_{\{1,...,m\}}$ is also an intersection hypothesis, so we will call a global test a ``local test'' as well. Suppose we reject the global hypothesis. Then we know (with $(1-\alpha)100\%$ confidence) that at least one of the $m$ hypotheses is false. Likewise, suppose we reject the intersection hypothesis $\Hy_{\{2,5,6\}}$ for example. Then at least one of the hypotheses $\Hy_2$, $\Hy_5$, and $\Hy_6$ is false (unless we made a false discovery).

In \S\ref{secglobaltests} we considered various global tests. Note that we can use these global tests as local tests. Indeed, a local test \emph{is} a global test applied to a subcollection of hypotheses.

For example, suppose that for every hypothesis we have a valid p-value $P_i$. Consider some $\I\in \subs$. 
By \S\ref{secbonglobal} a valid local test for $\Hy_{\I}$ is then the test that rejects $\Hy_{\I}$ if and only if $\min\{P_i:i\in \I\}\leq \alpha/|\I|$. It will be useful to view a test as an indicator function $\mathbbm{1}(\cdot)$, i.e., a function  that takes the values 0 and 1, where 1 indicates rejection and 0 indicates a non-rejection. Thus, we can write these global tests as $\phi_{\I}$, where $\phi_{\I}$ is the indicator function
\begin{equation} \label{ltBon}
\phi_{\I}=\mathbbm{1}\Big(\min\{P_i:i\in \I\}\leq \alpha/|\I|\Big).
\end{equation}

As another example, suppose that for all nonempty $\I\subseteq \N$, the $P_i$ with $i\in \I$ satisfy Simes' inequality (see \S\ref{secglobalsimes}), i.e., 
\begin{equation} \label{eq:Si}
\pr\Big\{\forall i \in \I: P_{(i)}^{\I}> \frac{i\alpha}{|\I|}\Big\}\geq1-\alpha,
\end{equation}
where we define $P_{(1)}^{\I},..., P_{(|\I|)}^{\I}$ to be  the sorted values of the $P_i$ with $i\in \I$.
 Then a valid local test for 
$\Hy_{\I}$ is 
\begin{equation} \label{lotestSimes}
\phi_{\I}'=\mathbbm{1}\Big(\exists i \in \I: P_{(i)}^{\I}  \leq \frac{i\alpha}{|\I|} \Big).
\end{equation}
Clearly, the local test $\phi_{\I}'$ is uniformly more powerful than $\phi_{\I}$ (but requires the extra assumption that \eqref{eq:Si} holds if $\I\subseteq \{1,...,m\}$). Note that if we consider an elementary hypothesis $\I=\{i\}$, then $\phi_{\I}$ and $\phi_{\I}'$ are both  simply $\mathbbm{1}(P_i\leq \alpha)$.

We are now ready to define the closed testing procedure based on the local tests. It turns out that the elementary hypotheses rejected by  the closed testing procedure based on the local tests \eqref{ltBon}, are the same as those rejected by the Holm method. The closed testing procedure based on the local tests $\phi_{\I}'$ potentially rejects more.

\subsection{The general closed testing method for FWER control} \label{secCTPgeneral}
The closed testing procedure  tests  all  intersection hypotheses $\Hy_{\I}$, $\I\in \subs$. Thus, to every intersection hypothesis, the procedure assigns a 1 (``reject'') or a 0 (``do not reject''). This means that the procedure in principle  considers $2^m-1$ tests. However, the procedure is often also powerful  for simply testing $\Hy_1,...,\Hy_m$, in case we are not interested in the intersection hypotheses. 

As explained, the ingredients of the closed testing procedure are the local tests. Thus, suppose that for every intersection hypothesis $\Hy_{\I}$ we have defined some local test  $\phi_{\I}$. Then the closed testing procedure is defined as follows.

\begin{definition} \label{defctp}
The closed testing procedure (CTP) based on the local tests $\phi_{\I}$, $\I\in \subs$ is the 
procedure that assigns a  decision $\psi_{\I}$ to every $\I\in \subs$ in the following way\footnote{An equivalent formulation is  that the CTP rejects $\Hy_{\I}$ if  and only if for all $I\subseteq \J\in \subs$,  $\phi_{\J}=1.$}:
$$
\psi_{\I} = \min\{\phi_{\J}: \I\subseteq \J\in \subs\}.$$
\end{definition}

Thus, the CTP rejects $\Hy_{\I}$ if and only if  $\phi_{\J}=1$ for al $I\subseteq \J\in \subs$ --- i.e., if all local tests $\phi_{\J}$ with $I\subseteq \J\in \subs$ reject their corresponding intersection hypotheses. 
 The $\psi_{\I}$ are sometimes called the \emph{effective local tests}, to distinguish them from the local tests. The CTP is  the collection of all  effective local tests.  Figure \ref{fig:ctp} and its caption provide further explanation.
Closed testing ensures that the set of rejected intersection hypotheses has  a logical structure: If  $\I\subset \J\in\subs$ and $\I$ is rejected by the CTP, then $\J$ is also rejected. This makes sense since if $\Hy_{\I}$ is false, then one of the $\Hy_i$ with $i\in \I$ is false, which implies that $\Hy_{\J}$ is false.


\begin{figure}[h!] 
\centering

\begin{tikzpicture}[
    node distance=1.8cm and 1.8cm,
    every node/.style={draw, minimum size=1cm, inner sep=4pt},
    crossed/.style={draw=black, line width=0.8pt, text=black},
    accepted/.style={draw=black, text=black},
    crossed out/.style={crossed, append after command={
        \pgfextra{
            \draw[red, line width=1pt] (\tikzlastnode.north west) -- (\tikzlastnode.south east);
            \draw[red, line width=1pt] (\tikzlastnode.north east) -- (\tikzlastnode.south west);
        }
    }},
    arrow/.style={->, black, thick}
]

\node[crossed out] (H123) {$H_1 \cap H_2 \cap H_3$};

\node[crossed out, below left=of H123] (H12) {$H_1 \cap H_2$};
\node[crossed out, below=of H123] (H13) {$H_1 \cap H_3$};
\node[accepted, below right=of H123] (H23) {$H_2 \cap H_3$};

\node[crossed out, below=of H12] (H1) {$H_1$};
\node[crossed out, below=of H13] (H2) {$H_2$};
\node[accepted, below=of H23] (H3) {$H_3$};

\draw[arrow] (H123) -- (H12);
\draw[arrow] (H123) -- (H13);
\draw[arrow] (H123) -- (H23);

\draw[arrow] (H12) -- (H1);
\draw[arrow] (H12) -- (H2);

\draw[arrow] (H13) -- (H1);
\draw[arrow] (H13) -- (H3);

\draw[arrow] (H23) -- (H2);
\draw[arrow] (H23) -- (H3);

\end{tikzpicture} 
\caption{This figure shows all intersection hypotheses in case there are three elementary hypotheses. Arrows indicate the logical relationships one level downwards. For example, if $\Hy_2\cap \Hy_3$ is true, then $\Hy_2$ and $\Hy_3$ must be true. The CTP only rejects a hypothesis if all hypotheses that imply it are rejected by their local tests. For example, suppose that the hypotheses indicated by crosses are rejected by their local tests. Which hypotheses are then rejected by the CTP?}     \label{fig:ctp}
\end{figure}
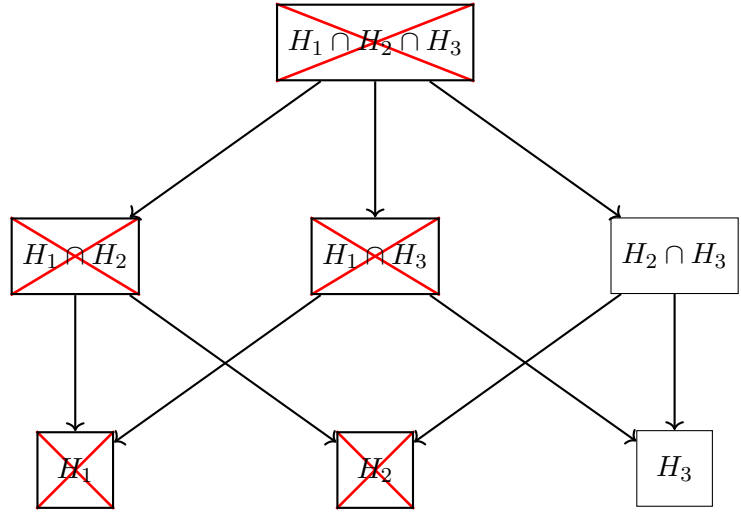

There are $2^m-1$ local tests, and evaluating them all is computationally undoable if $m$ is large. However, note that if an intersection hypothesis is not rejected, we know that all hypotheses implied by it (i.e. all hypotheses that can be reached by following the arrows in the closed testing tree) are also not rejected by the CTP. This can speed the computations up, but such a procedure can still be computationally feasible. Fortunately, in some cases, exact computational shortcuts are possible. There will be discussed later.

The reason why closed testing is called this way, is that the family of hypotheses $\Hy_{\I}$ with $\I\subseteq\{1,...,m\}$ is closed under intersections. The CTP tests the hypotheses in the closure  of the original hypothesis family $\Hy_1,...,\Hy_m$.
Even if we are  only interested in $\Hy_1,...,\Hy_m$ and not in the intersection hypotheses, CTPs are often powerful. We now check that the CTP indeed controls the FWER. In case we are interested in all intersection hypotheses, the FWER is $\pr(V^{\text{tot}}>0)$, where $V^{\text{tot}}$ is  the number of false positives among all rejected intersection hypotheses: 
$$V^{\text{tot}}= |\{\I\subseteq\N:  \I\neq\emptyset\text{ and }  \psi_{\I}=1\}|.$$
In case we are interested in the number of false positives among the rejected elementary hypotheses $\Hy_1,...,\Hy_m$, we consider
$$V^{\text{elem}}= |i\in \N:   \psi_{\{i\}}=1\}|.$$
\begin{theorem} \label{thmctp}
Consider hypotheses $\Hy_1,..,\Hy_m$ and suppose that for each $\I\in \subs$ we have a valid local test $\phi_{\I}$,  i.e., $\pr(\phi_{\I}=1)\leq \alpha $ if $\I\in \N$.
Then the corresponding CTP (see Definiton \ref{defctp}) controls the FWER, in the sense that $\pr(V^{\text{tot}}>0)\leq \alpha$. 
\end{theorem}
\begin{proof}
If all hypotheses are false, there is nothing to prove, so suppose $\N$ is not empty.
Let $\E = \{\phi_{\N}=0\}$ be the event that $\Hy_{\N}$ is not rejected by its local test. Note that $\pr(\E)\geq 1-\alpha$ since all local tests are valid.
Suppose $\E$ holds.
Consider any true hypothesis $\Hy_{\I}$, $\I\in \subs$. Since $\I\subseteq\N$, it follows by definition of the CTP that $\Hy_{\I}$ is not rejected by the CTP. 

We conclude that with probability at least $1-\alpha$, no true intersection hypotheses are rejected. Hence, $\pr(V^{\text{tot}}>0)\leq \alpha$.
\end{proof}

Clearly $V^{\text{elem}}\leq V^{\text{tot}}$. Even if the CTP does not reject any elementary hypotheses, it might still reject some other intersection hypotheses. 
However, it turns out that in several important cases, no intersection hypotheses can be rejected by the CTP in case no elementary hypotheses are rejected by the CTP. A CTP is called \emph{consonant} if it is always the case that for every $\Hy_{\I}$ that is rejected by the CTP, there is an $i\in \I$ such that $\Hy_{i}$ is rejected by the CTP. For such procedures, all interesting information is contained in the rejections (by the CTP) of the elementary hypotheses. Indeed, if we know that a CTP is consonant and we know which elementary hypotheses it rejects, then the intersection hypotheses that it rejects  are precisely those that are implied by the elementary hypotheses that are rejected by the CTP. For consonant procedures $V^{\text{elem}}=0\Leftrightarrow V^{\text{tot}}=0$.
If a CTP is nonconsonant, it may reject some intersection hypotheses without rejecting any elementary hypotheses. This idea will be important in \S\ref{secCTPfdp}. 

\begin{example}
Consider a consonant CTP with three elementary hypotheses. Suppose the elementary hypotheses that the CTP rejects are $\Hy_1$ and $\Hy_3$. Then we know that the intersection hypotheses that are locally rejected are the $\Hy_{\I}$ with $1\in \I$ or $3\in \I$. Those rejections are in a sense not interesting. For example, if $\Hy_1$ is false, then it immediately follows by elementary logic that e.g. $\Hy_{\{1,2\}},$ $\Hy_{\{1,3\}}$ and $\Hy_{\{1,2,3\}}$   must be false too.
\end{example}

An example of a consonant CTP is the CTP with local tests based on Bonferroni (this will be proved in \S\ref{secCTPfdp}). As already mentioned, the elementary hypotheses rejected by this CTP are the hypotheses that Holm rejects. 
 An example of a nonconsonant (or ``dissonant'') CTP is the CTP with local tests based on Simes' inequality (see exercises). We will return to the topic of nonconsonance in \S\ref{secCTPfdp}.

It follows from the proof of Theorem \ref{thmctp} that for the CTP to ensure that $\pr(V^{\text{tot}}>0)\leq \alpha$, it suffices that the local test $\phi_{\N}$ has type I error rate at most $\alpha$. Thus, in this sense, we do not require all local tests to be valid, but only $\phi_{\N}$. However, we do not know $\N$ of course; often it could be any subset of $\{1,...,m\}$, hence in a sense we need every local test to be valid.

\subsection{Example of a CTP: Holm} \label{secCTPHolm}
Recall that a local test is simply a global test applied to an intersection hypothesis $\Hy_{\I}$, $\I\in \subs$. Let us first consider the test from \S\ref{secbonglobal} by Bonferroni. This is the test
\begin{equation} \label{localtestHolm}
\phi_{\I} = \mathbbm{1}\big(\min\{P_i: i\in \I\}\leq \alpha /|\I|\big).
\end{equation}   
If we consider the corresponding CTP, then it turns our that it rejects exactly the same elementary hypotheses as Holm's method. Thus, to check which elementary hypotheses are rejected by the CTP, we do not need to perform $2^m-1$ local tests, but can simply use Holm's method, which is very fast. Thus, Holm's method is a computational \emph{shortcut} for the full CTP. (As has been mentioned and will be proved in \S\ref{secCTPfdp}, this CTP is consonant, so in a sense only the rejected elementary hypotheses are interesting. Holm tells us which these are.)

\begin{theorem} \label{thmHolmisCTP}
The  CTP with local tests  \eqref{localtestHolm} rejects exactly the same elementary hypotheses as Holm's method.
\end{theorem}
\begin{proof}
Without loss of generality, we can rename the hypotheses such that the  p-values are ordered: $p_{1}\leq ... \leq p_{m}$. (Thus, $\Hy_1$ is the hypothesis with the smallest p-value, etc.) Indeed, we can do this since clearly,  both methods simply reject a set of hypotheses corresponding to the smallest so-many p-values, and how many hypotheses they reject does not depend on the original order.

Note that  the methods reject at least one hypothesis if and only if Bonferroni's global test rejects $\Hy_{\{1,...,m\}}$.
Consequently,  the  CTP rejects at least $\Hy_{1}$ if and only if   Holm rejects at least $\Hy_{1}$.
\\
\\
\emph{Part 1: Showing that  Holm rejects at least as many as the CTP.}
Suppose that the CTP rejects $2\leq k \leq m$ hypotheses. We show that Holm does as well.
We give a proof by induction.
Let $1\leq j< k$ and suppose we have shown that Holm rejects $\Hy_1,...,\Hy_j$. We show that Holm also rejects $\Hy_{j+1}$.
Since Holm rejects  $\Hy_1,...,\Hy_j$, we know that $\forall 1\leq i\leq j:  p_i\leq \alpha/(m+1-i)$. We must show that also 
$p_{j+1}\leq \alpha/(m+1-(j+1))$. We know that the CTP rejects $\Hy_{j+1}$, i.e., 
$$ \forall \{{j+1}\} \subseteq \J \in \subs: \phi_{\J}=1  ,$$
i.e., 
$$ \forall\{{j+1}\} \subseteq \J \in \subs: \min\{p_i: i\in \J\}\leq \alpha/ |\J|.$$
Hence, in particular, for $\J=\{j+1,...,m\}$,
$$\min\{p_i: i\in \J\}\leq \alpha/ |\J|,$$
which means that
$$p_{j+1}\leq \alpha/ (m-j),$$ which means that indeed Holm rejects $\Hy_{j+1}$ as well.
\\
\\
\emph{Part 2: Showing that  the CTP rejects at least as many as Holm.}
Suppose that Holm rejects $2\leq k \leq m$ hypotheses. We show that the CTP does as well. We again use induction.
Let $1\leq j< k$ and suppose we have shown that the CTP  rejects $\Hy_1,...,\Hy_j$. We show that the CTP also rejects $\Hy_{j+1}$.
We know that 
$\forall 1\leq i\leq j+1:  p_i\leq \alpha/(m+1-i)$.
Further, we know that the CTP rejects $\Hy_j$, i.e.,
 $$\forall\{{j}\} \subseteq \J \in \subs: \min\{p_i: i\in \J\}\leq \alpha/ |\J|,$$
and we must show that 
$$ \forall\{{j+1}\} \subseteq \J \in \subs: \min\{p_i: i\in \J\}\leq \alpha/ |\J|.$$
Since the CTP rejects $\Hy_1,...,\Hy_j$, we know that once we include an index $l<j+1$ in $\J$, then $\psi_{\{l\}}=1$ and hence  $\phi_{\J}=1$ and hence  $\min\{p_i: i\in \J\}\leq \alpha/ |\J|$.
Hence, we only need to check that for all $\J\subseteq\{j+1,...,m\}$ that contain $j+1$, we have $\min\{p_i: i\in \J\}\leq \alpha/ |\J|$. But this clearly holds since for such $\J$ we have  $\min\{p_i: i\in \J\}=p_{j+1}\leq \alpha/(m+1-(j+1))\leq \alpha/|\J|$, since $\J\leq m-j$ for such $\J$.
\end{proof}

\subsection{Example of a CTP: using Simes-type local tests}\label{secHommel}
We will now consider Simes local tests, defined in \eqref{lotestSimes} already:
$$\phi_{\I} = \mathbbm{1}\Big( \bigcup_{i\in \I} \{P^{\I}_{(i)}\leq \frac{i\alpha}{|\I|} \}  \Big),$$
where $P^{\I}_{(1)}\leq ...\leq P^{\I}_{(|\I|)}$ are  the sorted p-values $P_i$ with $i\in \I$.

As explained, the CTP is  valid if the  local test $\phi_{\N}$ rejects with probability at most $\alpha$. Thus, we must assume that inequality \eqref{SimesProbEq} holds. As already discussed, this is a very mild assumption on the dependence structure of the null  p-values.

\citt{hommel1988stagewise} first discussed closed testing based on Simes local tests. 
When statisticians say ``Hommel's method'', note that they usually do not mean the procedure that tests all intersection hypotheses, but just the method for testing the  $m$ elementary hypotheses. In general, ``closed testing'' does not necessarily mean that we are interested in all $2^m-1$ intersection hypotheses: often we just use a closed testing method for testing the elementary hypotheses. (However, it turns out that the CTP based on Simes local tests is nonconsonant, so it can be useful to not merely look at the elementary rejections.)

The CTP based on Simes-type local tests is much more complex than Holm's method in the sense that it is less straightforward to come up with a computationally efficient algorithm (i.e. a shortcut) for evaluating which of the elementary hypotheses are rejected. 
If we naively performed the CTP, this would become computationally undoable for moderate values of $m$.
Hommel's original algorithm has a complexity that is quadratic in $m$, which is already much better than a completely naive implementation.  \citt{meijer2019hommel} were the first to come up with an algorithm with a complexity that is linear in $m$ (after sorting the p-values, which takes $O(m\log(m))$ time). 

\citt{hommel1988stagewise} showed that $\Hy_i$ is rejected by the CTP if and only if
$$h(\alpha)p_i\leq \alpha,$$
where 
$$ h(\alpha) = \max\big\{i\in\{0,...,m\}:  i p_{(m-i+j)}> j\alpha \text{ for all } j\in\{1,...,i\}\big\} .$$  
\citt{meijer2019hommel} compute  $h(\alpha)$ in linear time (after sorting the p-values).

The algorithm of  \citt{meijer2019hommel}  is implemented in the R package \emph{hommel} on CRAN \citp{hommel}. Its function of the same name, \emph{hommel()},  can be used as follows:
\verb|hommel(p, simes = TRUE)|, where \verb|p| contains the p-values. \verb|simes = TRUE| gives the method described above and \verb|simes = FALSE| gives a more conservative version of the method that does not require any assumptions on the dependence structure. The output of the function are the Hommel-adjusted p-values. Thus, all hypotheses with adjusted p-values at most $\alpha$ can be rejected.

\subsection{Example: one-way ANOVA model} \label{secANOVA-CT}
In \S\ref{secglobalanova} we considered a global test for the one-way balanced ANOVA model. In the section we explore closed testing for  this model. Some of these tests also have versions for non-balanced ANOVA models.

\subsubsection{Closed testing  method}
The test from \S\ref{secglobalanova} rejects when the statistic \eqref{eqttTW},
$$ \sqrt{2}\max_{1\leq i < i'\leq k} |T_{ii'}| =\frac{\max_{1\leq i \leq k} \bar{Y}_i- \min_{1\leq i \leq k}\bar{Y}_{i'}}{\hat{\sigma}/\sqrt{n}},$$
 is large. 
As discussed, under the global null hypothesis, this test statistic is a Studentized range variable with parameter $k$ and $\nu$ degrees of freedom, where $\nu=k(n-1)$.
We rejected the global null hypothesis when 
$$   \sqrt{2}\max_{1\leq i < i'\leq k} |T_{ii'}| \geq Q_{k,\nu}^{(1-\alpha)}.$$

To see the problem through the lens of closed testing, we must first define elementary hypotheses. These are the hypotheses
$\Hy_{\A}$ where $\A\subseteq\{1,...,k\}$ and $|\A|=2$. \emph{Here we interpret the notation} $\Hy_{\A}$ \emph{not as earlier in this section, but as in}  \S\ref{secFWERanova}: \emph{it  is the hypothesis that the} $\theta_i$ \emph{with} $i\in \A$ \emph{are equal to each other.} Clearly, the number of elementary hypotheses is $\binom{k}{2}$. We will call the number of elementary hypotheses $m$, like before. 

We now consider the intersection hypotheses, i.e., all possible intersection of the $m$ elementary hypotheses. At first sight, there seem to be $2^m-1$ intersection hypotheses. However, we are then counting some intersection hypotheses multiple times. For example, consider the elementary hypotheses $\Hy_{\{1,2\}}$,  $\Hy_{\{2,3\}}$ and  $\Hy_{\{1,3\}}$. The intersections  $\Hy_{\{1,2\}}\cap \Hy_{\{2,3\}}$, $\Hy_{\{2,3\}}\cap\Hy_{\{1,3\}}$ and $\Hy_{\{1,2\}}\cap\Hy_{\{1,3\}}$ are in fact the same: they are all equal to $\Hy_{\{1,2,3\}}$. Further, note that  all intersection hypotheses can be written as intersections of the form $$\Hy_{\A_1}\cap\Hy_{\A_2}\cap...$$ with  $\A_1, \A_2,...$ disjoint. 

\subsubsection{Shortcuts}
As always, naive closed testing is computationally infeasible when $m$ is large, so it is useful to look for shortcuts. 
\begin{itemize}
\item
\textbf{Tukey-Welsch.} One important shortcut is the Tukey-Welsch procedure from \S\ref{secFWERanova}. The advantage of Tukey-Welsch is that most of the intersection hypotheses do not need to be individually considered, since we only need to consider some of the hypotheses $\Hy_{\A}$.
Tukey-Welsch is considered to have good power, although its power is not equal to that of the full CTP (the full CTP is a uniform improvement of Tukey-Welsch).

\item
\textbf{Newman-Keuls and Begun-Gabriel.}
In case studentized range test statistics are used, a small modification of the Tukey-Welsch procedure gives a method that is even faster and is actually uniformly more powerful than the CTP. This method is referred to as the Newman-Keuls procedure \citp[][p.66]{hochberg1987multiple}. This method is not known to control the FWER.
However, since it is uniformly more powerful than the CTP, and since the CTP is uniformly more powerful than Tukey-Welsch,  we can use Newman-Keuls as an upper bound for the CTP's decision and Tukey-Welsch as a lower bound for the CTP's decision.
 This can be helpful since we know that all hypotheses that are rejected by Tukey-Welsch will also be rejected by the CTP and all hypotheses that are not-rejected by Newman-Keuls will not be rejected by the CTP. Hence, after running Tukey-Welsch and Newman-Keuls, there is  only a limited set of hypotheses left that is contentious, i.e., for which it should be further checked whether they can be rejected or not. (This is exploited by a method by Begun and Gabriel, see  \citealp[][p.122]{hochberg1987multiple}.)

The Newman-Keuls method is defined as folows. We let $\alpha_p:=\alpha$ for every $2\leq p \leq k$.
We then re-name the  parameters such that $\hat{\theta}_{1}\leq ... \leq \hat{\theta}_{k}$.
Then we consider Tukey-Welsch, but with this simpler definition of $\alpha_p$. Note that if this method rejects
$\Hy_{\A}$ with $\min(\A)=i$ and $\max(\A)=j$, then it follows that all hypotheses implied by $\Hy_{\A}$ and containing $i$ and $j$ are rejected, if studentized range statistics are used.  Thus, when we reject $\Hy_{\A}$, we can automatically reject all subsets of size $|\A|-1$ as well, except $|\A|\setminus \{i\}$ and $|\A|\setminus \{j\}$. Thus, if we reject $\Hy_{\A}$, then we can proceed by checking  $\Hy_{|\A|\setminus \{i\}}$ and $\Hy_{|\A|\setminus \{j\}}$, etcetera. In this way, Newman-Keuls saves time compared to Tukey-Welsch.

\item
\textbf{Adaptation of Tukey-Welsch.}
Finally, we can consider the following procedure, which is somewhat more conservative than Tukey-Welsch but faster \citp[][p.115]{hochberg1987multiple}. This adaptation can be used when studentized range test statistics are used.
First note that Newman-Keuls exploits the fact that $Q_{p,\nu}^{(1-\alpha_p)}$ decreases as $p$ decreases, when $\alpha_p:=\alpha$ for every $2\leq p \leq k$. This is not guaranteed is $\alpha_p$ is defined as in Tukey-Welsch, i.e., we may then sometimes have $Q_{p,\nu}^{(1-\alpha_{p-1})} > Q_{p,\nu}^{(1-\alpha_p)}$. To resolve this, we can enforce montonicity by making each  $Q_{p,\nu}^{(1-\alpha_{p+1})}$ at least as large as the previous critical value  $Q_{p,\nu}^{(1-\alpha_{p})}$. We then obtain a method that is somewhat more conservative than Tukey-Welsch, but for which we can use an algorithm analogous to that of Newman-Keuls.

\item
\textbf{Other methods for ANOVA models.}
There are also other methods for testing in ANOVA models.
A well-known method is Fisher's LSD (``Least Significant Difference'') procedure, but it only provides weak FWER control (as defined in \S\ref{secdefFWER}). It proceeds by first running a global F-test. If the global null is rejected, the method then tests all pairwise hypotheses without any further multiple testing correction.

Faster but less powerful alternatives for Tukey-Welsch are Tukey's HSD (``Honestly Significant Difference'') method and  Tukey-Kramer \citp{fay2022statistical}. These do provide strong FWER control.

\end{itemize}

\subsection{Exercises}
\begin{enumerate}

\item Remove and add some crosses in the ``closed testing tree'' from Figure \ref{fig:ctp} (the crosses represent the hypotheses rejected by the their local tests) and decide which hypotheses are rejected by the closed testing procedure. Additionally, draw the closed testing tree in case there are 4 elementary hypotheses.

\item  A multiple testing procedure --- which decides for each of the hypotheses in some family whether it is rejected or not --- is called \emph{coherent} if for each pair of hypotheses $\Hy^a$ and $\Hy^b$ in the family the following always holds: if $\Hy^a\subseteq \Hy^b$\footnote{Recall that a hypothesis is a set of distributions, so $\Hy^a\subseteq \Hy^b$ means that if the true distribution is in $\Hy^a$, it is also in $\Hy^b$. In other words, if $\Hy^a$ is true, then $\Hy^b$ is true, i.e., $\Hy^a$ implies $\Hy^b$.} and $\Hy^b$ is rejected, then $\Hy^a$ is rejected.

Consider hypotheses $\Hy_1,...,\Hy_k$ (these need not be elementary hypotheses) and consider some incoherent MTP for these $k$ hypotheses (i.e. a procedure that ``accepts'' or rejects each hypothesis). Assume the procedure controls the FWER at level $\alpha$. Now consider a new MTP, which rejects each hypothesis $\Hy_i$ ($1\leq i \leq k$) that implies a hypothesis that is rejected by the original procedure. Thus, the new procedure rejects the same hypotheses as the old one, and possibly more. Prove that the new procedure is coherent   and that it also controls the FWER at level $\alpha$  (the proof is short).

\item  

Consider some family $\Hy_1,...,\Hy_k$ of hypotheses (not necessarily elementary hypotheses).  Recall that a multiple testing procedure for this family  is called \emph{coherent} if for each pair of hypotheses $\Hy_i$ and $\Hy_i$ in the family the following holds: if 
$\Hy_i\subseteq \Hy_j$  
 and $\Hy_j$ is rejected, then $\Hy_i$ is rejected.

 Consider a coherent MTP for the family  $\Hy_1,...,\Hy_k$ and suppose the procedure has FWER at most $\alpha$. 
We will now prove that this MTP is the result of applying closed testing based on a collection of local tests with size at most $\alpha$. Thus, every coherent procedure can be seen as a closed testing procedure.

Consider the closure (under intersections) of the  family $\Hy_1,...,\Hy_k$ and define a local test for each hypothesis $\Hy$ in this closure as follows. If $\Hy$ is in the original family, reject $\Hy$ locally if and only  the original MTP rejects it. Otherwise, reject  $\Hy$ locally if and only if at least one of the locally rejected hypotheses (by the new tests) among $\Hy_1,...,\Hy_k$ is implied by $\Hy$.

(a) Now consider the CTP based on these local tests.  Prove that this CTP has FWER (over all intersection hypotheses) at most $\alpha$.

(b)
Prove that among $\Hy_1,...,\Hy_k$, the hypotheses rejected by this CTP are the same as those rejected by the original MTP.

\item 
Let $\alpha\in(0,1)$. Consider $m\geq 3$ elementary hypotheses with corresponding p-values $P_1,...,P_m$. Use a counterexample to show that the corresponding CTP based on Simes local tests is not necessarily consonant. (\emph{Hint:} Consider e.g. $p_1=...=p_{m-1}=\alpha\frac{m-1}{m}$ and $p_m>\alpha$.)

\item
Prove that if $m=2$, the CTP based on Simes local tests is consonant.

\item $\bigstar$ Explain that in the specific context of \S\ref{secANOVA-CT}, the number of distinct intersections of elementary hypotheses equals the number of partitions of the set $\{1,...,k\}$, minus 1. (A partition is a collection of subsets of $\{1,...,k\}$, such that each element of  $\{1,...,k\}$ is contained in exactly  one of these subsets.)

\end{enumerate}

\begin{sols}

\subsection{Solutions}
\begin{enumerate}

\item In case there are 4 hypotheses, the closed testing tree consists of 4 layers. The top layer consists of 1 hypothesis, $\Hy_1\cap....\cap\Hy_4$. The second layer consists of the four hypotheses $\Hy_1\cap \Hy_2\cap\Hy_3$, $\Hy_1\cap \Hy_2\cap\Hy_4$, $\Hy_1\cap \Hy_3\cap\Hy_4$, $\Hy_2\cap \Hy_3\cap\Hy_4$. The third layer consists of $\binom{4}{2}=6$ hypotheses and the fourth layer consists of the four elementary hypotheses $\Hy_1,....,\Hy_4$. From each node, you should draw an arrow to the nodes directly below it corresponding to the implied hypotheses. For example, from $\Hy_1\cap \Hy_2\cap\Hy_3$ there should be arrows to $\Hy_1\cap \Hy_2$, $\Hy_1\cap\Hy_3$ and $\Hy_2\cap\Hy_3$.

\item Consider one of the hypotheses. By construction, all hypotheses that imply it are rejected by the new procedure. Hence, the new procedure is coherent.
Let $\E$ be the event that the old procedure makes no false discoveries. Thus, all hypotheses rejected by the original procedure are false. The new procedure rejects those hypotheses, plus hypotheses which imply them, which must als be false (If $\Hy_i$ implies $\Hy_j$ and  $\Hy_j$ is false, then $\Hy_i$ must be false.)

\item (a)

With probability at least $1-\alpha$, the original MTP makes no false discoveries. Suppose it makes no false discoveries.   Thus, for the new procedure, all local rejections of hypotheses among $\Hy_1,...,\Hy_k$ are correct. The new procedure then also rejects every hypotheses $\Hy$ that implies one of  the rejected hypotheses among $\Hy_1,...,\Hy_k$. Clearly, if $\Hy_i$ is correctly rejected, then it is also correct to reject all hypotheses that imply $\Hy_i$. Thus, with probability at least $1-\alpha$, none of the local tests of the new procedure make a false discovery. Consequently, the new CTP has FWER at most $\alpha$.

(b) Consider $1\leq j \leq k$ and suppose $\Hy_j$ is rejected by the original procedure. We will show that it is rejected by the defined closed testing procedure. Note that the local test rejects $\Hy_j$. Moreover, by coherence, all hypotheses among $\Hy_1,...,\Hy_k$ that imply $\Hy_j$, are rejected by the original procedure. Consider an intersection hypothesis $\Hy_{\I}$ with $j\in \I$.  $\Hy_j$ is rejected locally and is implied by $\Hy_{\I}$, so by definition of the local tests, $\Hy_{\I}$ is rejected locally. Thus, all $\Hy_{\I}$ with $j\in \I$ are rejected locally, so $\Hy_j$ is rejected by the defined CTP.

Conversely, suppose $\Hy_j$ is rejected by the defined closed testing procedure. Then $\Hy_j$ is rejected locally, and by construction this means the original procedure rejects it. This finishes the proof.

\item Suppose  $p_1=...=p_{m-1}=\alpha\frac{m-1}{m}$ and $p_m>\alpha$. 
Note that $\Hy_{\{1,...,m\}}$ is rejected by its local test, since $p_{m-1}\leq\alpha\frac{m-1}{m}$.
Since $\Hy_{\{1,...,m\}}$ is at the top of the closed testing tree, it directly follows that it is rejected by the CTP.

Now consider any $\I\subseteq \{1,...,m\}$ with $|\I|=m-1$ and $m\in \I$. Let $p_{(1)}^{\I}\leq...\leq p_{(m-1)}^{\I} $ be the sorted p-values corresponding to $\I$.
$\Hy_{\I}$ is rejected by its local test if and only if there is an $1\leq i \leq m-1$ such that $p_{(i)}^{\I}\leq c^\I_i$, where $c^\I_i =\frac{\alpha i}{|\I|}$.
Note that $\Hy_{\I}$ is not rejected by its local test, since $p^{\I}_{|\I|-1}= \alpha\frac{ m-1}{m}$, which exceeds the critical value $c^{\I}_{|\I|-1}=\alpha\frac{ |\I|-1}{|\I|}=\alpha\frac{ m-2}{m-1}$, and likewise the other p-values corresponding to $\I$ exceed their critical values.

We have seen that none of the hypotheses $\Hy_{\I}$ with $\I\subseteq \{1,...,m\}$ and $|\I|=m-1$ and $m\in \I$ are  rejected by the their local tests. For every $1\leq i \leq m$ there is such an $\I$ which contains $i$. It follows that none of the elementary hypotheses are rejected by the CTP. However, $\Hy_{\{1,...,m\}}$ is rejected by the CTP. Hence, the procedure is nonconsonant.

\item
Apart from the elementary hypotheses $\Hy_1$ and $\Hy_2$, say, there is only one intersection hypothesis, $\Hy_{12}$. Let $p_1, p_2$ be the p-values corresponding to $\Hy_1$ and $\Hy_2$.
Suppose $\Hy_{12}$ is rejected by the CTP. 
That means that either $p_{(1)}\leq \alpha\frac{1}{2}$ or $p_{(2)}\leq \alpha$ (or both). Thus, at least one of the p-values $p_1, p_2$ is at most $\alpha$. 
For $i\in\{1,2\}$, the local test for $\Hy_i$ simply rejects when $p_i\leq \alpha$. Hence, at least one of the elementary hypotheses is rejected locally, and hence also by the CTP.
Thus, if $\Hy_{12}$ is rejected by the CTP, then at least one of the elementary hypotheses is rejected. We conclude that the method is consonant when $m=2$.

\item 
Let $\mathcal{P}$ be the set of all possible partitions of $\{1,...,k\}$ exluding the partition that consists of $m$ singletons.
It suffices to show that there is a one-to-one mapping (bijection) from $\mathcal{P}$ to the set of all intersection hypotheses, as follows. We define this mapping as follows: we map every partition in $\mathcal{P}$ to the intersection hypothesis $\Hy_{\A_1}\cap \Hy_{\A_2}\cap ....$, where $\A_1,\A_2,...$ are the elements of the partition that contain at least 2 elements. To see that this is a one-to-one mapping, consider a specific intersection hypothesis and note that it can be uniquely written as $\Hy_{\A_1}\cap \Hy_{\A_2}\cap ....$, where $\A_1,\A_2,...$ are disjoint and have size at least 2. Clearly, there is only one element in $\mathcal{P}$ that maps to this intersection hypothesis --- namely, the partition that contains $\A_1,\A_2,...$ and that contains all elements of $\{1,...,k\}$ that are not in $\A_1\cup\A_2,...$ as singletons.

\end{enumerate}

\end{sols}

\newpage

\section{Resampling-based FWER control} \label{secresamplingFWER}
In \S\ref{secgroupinv} and \S\ref{secboot} in particular, we discussed resampling-based nonparametric and semiparametric testing (e.g. permutation and bootstrap testing) in case there is a single hypothesis. There, we discussed an advantage of such tests: they often require fewer assumptions than fully parametric tests. Among multiple testing methods, there are also some important resampling-based methods. Besides the mentioned robustness, they have another advantage: they take into account the dependence structure of the test statistics, without requiring strong assumptions on this dependence structure. This allows these methods to often have 
much better power than methods not based on resampling. An example of a method that does not account for the dependence structure, is Holm. That method is very conservative when there are strong positive dependencies between the p-values. The resampling-based maxT method discussed below, is then potentially a much more powerful alternative.

\subsection{Definition of the maxT method} \label{secdefmaxT}
For FWER control, one resampling-based method is by far the most popular: this is the \emph{maxT} method by \citt{westfall1993resampling}. A related method is the \emph{minP} method, which we do not cover here \citp{dudoit2008multiple,goeman2014multiple}.  MaxT can be based on permutation testing or bootstrap testing. In case of bootstrap testing, the method is sometimes referred to as the Romano-Wolf method \citp{romano2005stepwise}. Here we will refer to the approach as \emph{maxT}. In a general sense, the maxT method need not be based on resampling, and e.g. Holm could then be seen as a special case of maxT. However, maxT tends to be most powerful when combined with resampling, and here we will focus on resampling-based versions of maxT.

Both in \S\ref{secgroupinv} and \S\ref{secboot}, we considered ``resampled test statistics'': 
\begin{itemize}
\item In \S\ref{secgroupinv} we considered the statistics $(T(gX):g\in\G)$, where $\G$ is a group of transformations (e.g. permutation maps), or we use a collection of random transformations from $\G$.
\item In \S\ref{secboot} we considered the test statistics $(T(X^j): 1\leq j \leq b)$, where $X^1,...,X^b$ are bootstrap samples.
\end{itemize}

Let $w$ denote the number of transformations used (in case we use   group invariance testing, e.g. permutation testing) or the number of bootsrap samples (in case bootstrap testing is used). 
The maxT method creates resampled test statistics  as in \S\ref{secgroupinv} and \S\ref{secboot}, except that it does this for every hypothesis, and it applies \emph{simultaneous permutation} or \emph{simultaneous bootstrapping} (further explained below). Since the method resamples test statistics for every hypothesis, it does not merely end up with a vector of length $w$, but with an $m$-by-$w$ matrix, for which we use the following notation:
\begin{equation} \label{eq:matrixT}
\bm{T} = (T_i^j)_{1\leq i \leq m, 1\leq j \leq w}=
\begin{bmatrix}
T_1^1 & T_1^2 & \ldots & T_1^w \\
T_2^1 & T_2^2  &  \ldots & T_2^w \\
\vdots & \vdots & \ddots & \vdots \\
T_m^1 & T_m^2  & \ldots & T_m^w
\end{bmatrix}
,
\end{equation}
where we will assume that the first column corresponds to the original data. Thus, if we use random permutations,  we use the identity map plus $w-1$ random permutations.
When constructing this matrix, the key point is the simultaneous resampling. In case permutations are used this means that for every $1\leq j \leq w$, $T_1^j,...,T_m^j$ are obtained using the same permutation map.  In case bootstrapping is used, it means that for every $1\leq j \leq w$, $T_1^j,...,T_m^j$ are obtained using the same bootstrap sample of the data. This means that in some particular sense, these transformations preserve the dependence structure of the test statistics. More precisely, let the $N$-by $w$ matrix $\bm{T}_{\N}$ contain the rows of $\bm{T}$ with indices in $\N$. A key idea behind the method is that all rows of $\bm{T}_{\N}$ have the same joint distribution  (for finite $n$ or asymptoticaly).

There is a single-step maxT procedure and a more powerful sequential, ``step-down'' version. The single-step maxT method is simply  the first step of the sequential procedure --- just like Bonferroni is the first step of Holm.
The (sequential maxT procedure is in fact equivalent to a CTP, as discussed later.
The  sequential maxT procedure takes the matrix $\bm{T}$ as input and is defined as follows:


\begin{enumerate}
\item For each $1\leq j \leq w$, compute the maximum $M^j:=\max\{T_i^j: 1\leq i \leq m\}$ and let $q$ be the $(1-\alpha)$-quantile of these maxima, i.e., $$q= \min\{x\in \reals: \frac{|\{1\leq j \leq w: M^j\leq x \}|}{w}\geq 1-\alpha  \}.$$
Reject all hypotheses $\Hy_i$ with test statistics $T_i>q$. (Here $T_i$ means $T_i^1$.)
\item 
If nothing was rejected in the previous step, stop. Otherwise, delete all rows from $\bm{T}$ corresponding to the rejected hypotheses and repeat the above. This  leads to a potentially smaller threshold, $q_2$ say, which may lead to additional rejections.
\item 
Continue like this until a step is reached where no additional hypotheses are rejected.
\end{enumerate}

MaxT tends to be more powerful than Bonferroni or Holm when there are positive dependencies between the tested variables. The reason is that the stronger the positive dependence is between the test statistics in a column of $\bm{T}$, the smaller the maximum of these test statistics tends to be, and hence the smaller the rejection threshold ($q$) tends to be.

As we will discuss further in  \S\ref{secmaxTCTP}, the sequential maxT method is related to Holm's method. In the same way as Holm is the sequential version of Bonferroni, maxT is the sequential version of single-step maxT. Essentially the only difference between Bonferroni and single-step maxT, is that Bonferroni uses a ``worst-case'' rejection threshold $\alpha/m$ that uses no knowledge on the dependence structure, while maxT uses a potentially less restrictive threshold. Another difference is that maxT uses test statistics rather than p-values, but that is not a fundamental difference; for example, we could use $-P_1,...,-P_m$ as test statistics within maxT.

\begin{remark} \label{maxT1by1}
 Just like Holm can be formulated as a procedure that rejects hypotheses one by one (starting with the smallest p-value), sequential maxT can also be formulated as a procedure that rejects hypotheses one by one (starting with the largest test statistic). Indeed, we can perform maxT by starting with the largest test statistic and checking whether it exceeds the quantile $q$. Then, if that is the case, we move on to the second largest statistic and check whether it exceeds a quantile that has been computed in the same way, but after removing one row from $\bm{T}$, etcetera.
\end{remark}

\subsection{Example data analysis} \label{secexmaxT}

To illustrate the maxT method, we revisit the dataset from \S\ref{secexcortest} about car models. An analysis of a larger dataset with maxT is in \S\ref{secribo}.
For illustration purposes, we consider the same limited set of car models as in \S\ref{secexcortest} ($n=5$), but now we consider all the variables, see Figure \ref{fig:cardata_allcols}. We define the \emph{mpg} variable to be $X$ and we define the other 10 variables to be $Y^1,...,Y^{10}$. Suppose we want to test the hypotheses $\Hy_i: X \indep Y^i$, $1\leq i \leq m=10$. We discussed in \S\ref{secexptest} how we can perform a permutation test to test such a hypothesis.
Now we proceed in a similar way, except that we simultaneously permute the columns corresponding to the $Y^i$. There are $5!=120$ ways in which this can be done.  Permuting the columns corresponding to the $Y^i$ will be equivalent to  simply permuting the first column as done in Figure \ref{fig:permutedX}.

\begin{figure}[htbp]
    \centering
    \includegraphics[width=0.9\textwidth]{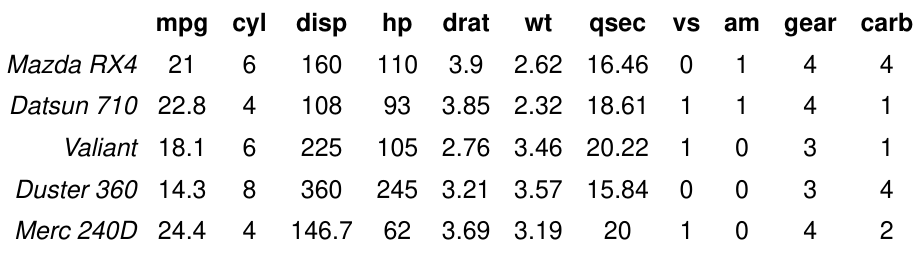}
    \caption{Example data on car models.}
    \label{fig:cardata_allcols}

\vspace{0.5cm} 


    \includegraphics[width=0.9\textwidth]{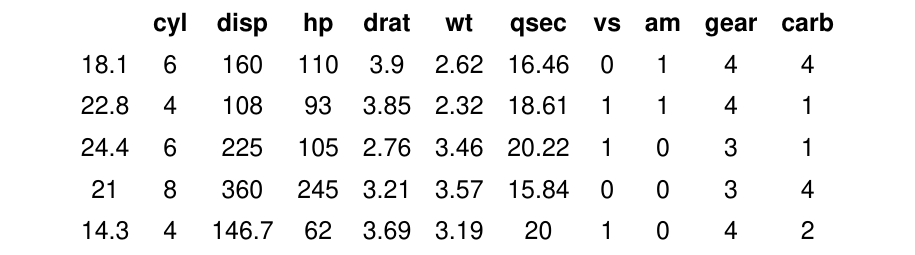}
    \caption{    
    Simultaneously permuting the columns corresponding to all variables except \emph{mpg},  will give the same results as simply permuting the column of \emph{mpg}. In this figure we only permuted the column of \emph{mpg} and kept the other columns as in Figure \ref{fig:cardata_allcols}.}
    \label{fig:permutedX}
\end{figure}

As test statistics we will use the absolute values of the correlations, i.e., $T_i = |\rho(X,Y^i)|$, $1\leq i \leq m$.
 For each of the 120 permuted versions of the data we compute the test statistics, and we collect the test statistics in the matrix $\bm{T} = (T_i^j)_{i,j}$ as in \eqref{eq:matrixT}. For each column of $\bm{T}$ we compute the maximum of the statistics. A histogram of the 120 resulting maxima is shown in Figure \ref{fig:max-stats}. We take $\alpha=0.1$ for illustration. The $(1-\alpha)$-quantile of the maximum-statistics $M^j$ was $q\approx 0.931$. Two of the test statistics exceeded $q$. By coincidence, these were the first two test statistics: $T_1\approx 0.953$ and $T_2\approx0.941$. This finishes the first step of the maxT method. For the second step, we remove the columns corresponding to $Y^1$ and $Y^2$ and essentially apply the same procedure again. The $(1-\alpha)$-quantile of the maximum-statistics is now $0.929$. None of the remaining test statistics $T_3,...,T_{10}$ exceed this value, so we get no further rejections and stop the procedure. Thus the sequential maxT was able to reject 2 hypotheses in this example.

\begin{figure}[htbp]
    \centering
    \includegraphics[width=0.9\textwidth]{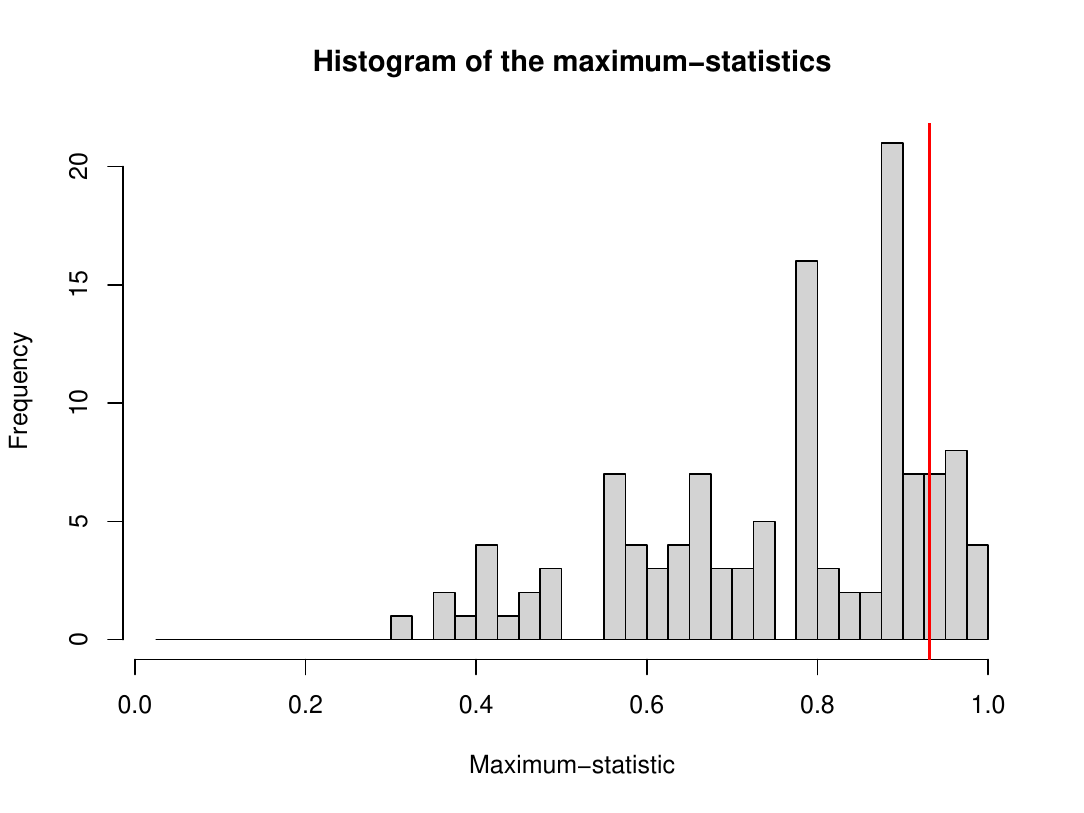}
    \caption{For each of the 120 permutations we computed the maximum of the $m=10$ test statistics. This results in the histogram above. The red line indicates $q$, which is the empirical $(1-\alpha)100\%$ quantile of these maxima (here $\alpha=0.1$). The value $q$ is the rejection threshold of the single-step maxT method. The sequential maxT method then continues by finding a potentially smaller threshold.}
    \label{fig:max-stats}
\end{figure}

We have seen that maxT rejects 2 hypotheses. Alternatively, we could compute individual p-values for each of the 10 hypotheses as in \S\ref{secexcortest}, and then apply Bonferroni (or Holm). The smallest (non-adjusted) p-value that we then obtain is $2/120=1/60$. The smallest Bonferroni-adjusted p-values is thus $(1/60)\cdot m=1/6$, which is larger than $\alpha$. Thus we see that if we perform classical permutation tests and use Bonferroni (or Holm), we do not reject anything. 

If more permutations are available, so that p-values can be very tiny, maxT often still has much better power, when there are strong dependencies among the test statistics.  This will be illustrated with a different dataset in \S\ref{secribo}. There Holm rejects 53 hypotheses, while maxT rejects 74 hypotheses. When there are many strong positive dependencies between the test statistics, the relative difference between Holm and maxT can be even much bigger.

\subsection{Theory for the maxT method} \label{secmaxTtheory}
Here we make more precise in which case the maxT method controls the FWER, for finite samples or asymptotically. Asymptotic FWER control means that $\liminf_{n\rightarrow\infty}FWER_n\leq \alpha$, where the substript $n$ indicates the sample size. We first briefly discuss using bootstrapping and then discuss using permutations (or other groups of transformations).
\\
\\
\emph{Bootstrap sampling.}
Recall that we defined $\bm{T}_{\N}$ to be the matrix containing the rows of $\bm{T}$ with indices in $\N$. Recall that we assume in this section that the first row of $\bm{T}$ is based on the original data.
In case of bootstrapping, a sufficient condition for asymptotic FWER control is that asymptotically, the rows of $\bm{T}_{\N}$ are i.i.d.. This will often be the case if we proceed analogously to \S\ref{secboot}, where now instead of bootstrap-sampling individual test statistics, we simultaneously bootstrap-sample whole columns of $\bm{T}$ (i.e., all values in the column are based on the same bootstrap sample of the data). For finite samples, bootstrap methods are typically not exact.
\\
\\
\emph{Permutations.}
In case we use permutations, we may even have finite-sample FWER control. Consider the following assumption.
Let $X$ be the full dataset and let $\G$ be the group of transformations (e.g. permutation maps) applied to the data.
\begin{assumption} \label{assjointd}
The joint distribution of the test statistics $T_i(gX)$ with $i \in \N$ and $g\in \G$ is invariant under all transformations $g\in \G$ of $X$. 
\end{assumption}
This assumption is satisfied in particular if the joint distribution of the part of the data corresponding to $\N$ is unchanged by the transformations.
In the example from \S\ref{secexmaxT}, a sufficient condition for the assumption to be satisfied,  is that permuting the first column of the dataset does not change the joint distribution of the first column and the columns corresponding to the true hypotheses --- or, equivalently, that permuting column 2 to 11 does not change that joint distribution.

\begin{theorem} \label{thmmaxTvalid}
Under Assumption \ref{assjointd}, maxT controls the FWER, i.e., $\pr(V>0)\leq\alpha$.
\end{theorem}
\begin{proof}
If all hypotheses are false, $V=0$. Now suppose there is at least one true hypothesis.

For every $g\in \G$ consider the maximum $M_{\N}^g:=\max\{T_i(gX):i\in \N\}$ and let $q^{\N}$ be the $(1-\alpha)$-quantile of these maxima. Let $\E$ be the event that $M_{\N}^{id}\leq q^{\N}$. Here $id$ is the identity transformation, i.e., $id(X)$ are the original data. 
It follows from Assumption \ref{assjointd} that the joint distribution of the values $M_{\N}^{g}$, $g\in \G$ is invariant under all transformations  $g\in \G$. Hence, by Theorem \ref{basic}, $\pr(M_{\N}^{id}> q^{\N})\leq \alpha$, i.e., $\pr(\E)\geq 1-\alpha$.

 In the rest of the proof, suppose $\E$ holds.
Note that if $\E$ happens, then step 1 of the  maxT method rejects no true hypotheses, since for every $i\in \N$, we then have $T_{i}(X)\leq M_{\N}^{id}\leq q^{\N}\leq q$.

In the second step of maxT, the test statistics corresponding to hypotheses that were rejected in step 1 are ignored, but we will still have $q^{\N}\leq q_2$ and hence $M_{\N}^{id}\leq q_2$, which means that again no true hypotheses are rejected. 

We can continue this reasoning, which means that  no true hypotheses are rejected when $\E$ happens. Since $\pr(\E)\geq 1-\alpha$, this finishes the proof.
\end{proof}
In case random permutations are used (resampled with or without replacement, but such that the first transformation is fixed to be the identity) then maxT also controls the FWER (in a marginal sense). In practice, often random transformations are used, to limit the computational burden.

\subsection{MaxT is a closed testing procedure} \label{secmaxTCTP}
Note that maxT is  similar to Holm, in the sense that it starts with a first step, and then if there are any rejections, it ignores the hypotheses that have been rejected and repeats the first step, etcetera. We saw that Holm is an exact shortcut for a CTP. Likewise, maxT is an exact shortcut for a (consonant) CTP,
namely the CTP based on the following local tests for the intersection hypotheses $\Hy_I$:
\begin{equation} \label{eqltmaxT}
\phi_{\I}= \mathbbm{1}\big( \max\{T_i:i\in\I\} > q^{\I}  \big),
\end{equation}
where 
$q^{\I}$ is the $(1-\alpha)$-quantile  of the maxima  $\max\{T_i(gX):i\in \I\}$, $g\in \G$.
To show this, we can reason in a way that is similar to how we showed that Holm is a CTP (Theorem \ref{thmHolmisCTP}). 

\begin{theorem}
The elementary hypotheses rejected by the CTP based on the local tests defined in \eqref{eqltmaxT}, are exactly the same as those rejected by maxT.
\end{theorem}

\begin{proof}
We must show that maxT rejects at least as many hypotheses as the CTP, and vice versa. 
To show that maxT rejects at least as many hypotheses as the CTP, we proceed as follows. 
We first observe that if the CTP rejects at least one hypothesis, then $ \max\{T_i:1\leq i \leq m\} > q^{\{1,...,m\}\}}  $, which means that maxT rejects at least one hypothesis. Number the hypotheses in such a way that
$T_{1}\geq ... \geq T_{m}$ (thus, $\Hy_1$ is the hypothesis with the largest test statistic).
 Next, suppose that 
the CTP rejects $2\leq k \leq m$ hypotheses. We show that maxT  does as well. We give a proof by induction.
Let $1\leq j< k$ and suppose we have shown that maxT rejects $\Hy_1,...,\Hy_j$. We show that maxT also rejects $\Hy_{j+1}$.
As mentioned in Remark \ref{maxT1by1}, maxT can be seen as a procedure that rejects hypotheses one by one.
Since maxT rejects  $\Hy_1,...,\Hy_j$, we know that $\forall 1\leq i\leq j:  T_i> q^{\{i,...,m\}}$. We must show that also 
$T_{j+1}> q^{\{j+1,...,m\}}$. We know that the CTP rejects $\Hy_{j+1}$, i.e., 
$$ \forall \{{j+1}\} \subseteq \J \in \subs: \phi_{\J}=1  ,$$
which means that
$$ \forall\{{j+1}\} \subseteq \J \in \subs: \max\{T_i: i\in \J\}> q^{\J}.$$
Hence, in particular, for $\J=\{j+1,...,m\}$,
$$\max\{T_i: i\in \J\}> q^{\J},$$
which means that
$$T_{j+1}> q^{\{j+1,...,m\}},$$ which means that indeed maxT rejects $\Hy_{j+1}$ as well.

The above reasoning is analogous to the first half of the proof of Theorem \ref{thmHolmisCTP}, which says that  Holm rejects at least as many hypotheses as the CTP based on Bonferroni. To show that maxT does not reject more than the CTP, we can proceed analogously to the second half of the proof of that theorem.
\end{proof}

\subsection{Software implementations of maxT}
One can easily program maxT oneself, although implementations may vary in their computational efficiencies.
Some  efficient implementations of the maxT method are contained in the R package \emph{multtest} \citp{pollard2005multiple,multtest}. This package also provides maxT-adjusted p-values (see \citealp{romano2016efficient} for a formula). Recall that the \emph{adjusted p-value} for a hypothesis is the smallest $\alpha$ for which the hypothesis is rejected by the multiple testing procedure.

As mentioned, in case bootstrapping is used, the maxT method is sometimes called the Romano-Wolf method. The method is implemented under that name in Stata \citp{clarke2020romano}.

\subsection{Extension to k-FWER control} \label{secmaxTkFWER}
Just like the methods by  Bonferroni and Holm, the maxT method can be generalized to k-FWER control. There is a single-step version and a sequential version. The  single-step method is again simply the first step of the sequential version. Here we only provide the single-step version. For the sequential version, see Algorithm 2.1 of \citt{romano2007control}. The sequential version (for $k>1$) tends to be much more computationally intensive than  sequential maxT. This has to do with the fact that after performing e.g. the first step, we cannot be confident that the rejected hypotheses are all false (since we allow for $k-1$ false positives).

For any collection of numbers, define the \emph{k-maximum} as the $k$-th largest number, i.e., the smallest real number $x$ such that at most $k-1$ of the numbers exceed $x$.
The single-step k-FWER version of maxT is exactly the same as maxT, except that we now use the $k$-maxima where before we would take the maxima. Thus, $q$ is now the $(1-\alpha)$-quantile of the $k$-maxima 
$$k\text{-}\max\{T_i^j: 1\leq i \leq m\},\quad 1\leq j \leq w.$$
Then, under Assumption \ref{assjointd}, we obtain  k-FWER control, and in case of bootstrapping we can obtain asymptotic k-FWER control.

\subsection{Exercises}
\begin{enumerate}
\item  Show that maxT cannot reject any hypotheses if $w<\alpha^{-1}$.
\item 
In practice one would take $w$ large, but here we take $w=5$ for simplicity.
Suppose we are interested in testing $m=3$ hypotheses and have computed the follow matrix $\bm{T}$:
\begin{equation} 
\bm{T} =
\begin{bmatrix}
6.2 & 6.0 & 4.1 & 5.7 & 4.4 \\
6.3 & 4.0  &  6.2& 4.0 & 3.7 \\
3.5 & 5.4 & 4.4 & 2.9 & 3.5\\
\end{bmatrix}
,
\end{equation}
As usual, the first column contains the test statistics corresponding to the original data.
Take $\alpha=0.2$. Which hypotheses do the single-step and  the sequential version of maxT reject in this case? How many hypotheses are rejected when $\alpha<0.2$?

\item
$\bigstar$ Consider data $X$, hypotheses $\Hy_1,...,\Hy_m$ and corresponding test statistics $T_1(\cdot),...,T_m(\cdot)$. Suppose Assumption \ref{assjointd} holds. Consider some rejection threshold $t\in \reals$. 
Thus, we reject the hypotheses with indices in $\R(X):=\{1\leq i \leq m: T_i(X)>t\}$.

For every $g\in \G$, let 
$$ R^g(X):=|\{1\leq i \leq m: T_i(gX)>t\}|.$$
Let $$R^{(1)}\leq ...\leq R^{(|\G|)}$$ be the sorted values $R^g$, $g\in \G$.

Let $\alpha\in(0,1)$ and let 
$\bar{V}(X):= R^{(k)}$, where $k=\lceil (1-\alpha)|\G|\rceil$.

Prove that
$$\pr(V(X)\leq \bar{V}(X))\geq 1-\alpha,$$ i.e., $ \bar{V}(X)$ is a $(1-\alpha)100\%$-confidence upper bound for the number of false positives $V(X)=|\N\cap\R(X)|$. (This approach is discussed and  extended in \citealp{hemerik2018false}.)

\end{enumerate}

\begin{sols}

\subsection{Solutions}
\begin{enumerate}
\item  For each $1\leq j \leq w$, single-step  maxT computes the maximum $M^j:=\max\{T_i^j: 1\leq i \leq m\}$ and lets $q$ be the $(1-\alpha)$-quantile of these maxima, i.e., $q= \min\{x\in \reals: \frac{|\{1\leq j \leq w: M^j\leq x \}|}{w}\geq 1-\alpha  \}.$
Here $\frac{|\{1\leq j \leq w: M^j\leq x \}|}{w}$ takes values in $\{1/w, 2/w, .... w/w\}$. If $w<\alpha^{-1}$, then $(w-1)/w<1-\alpha$. It follows that $q$ is actually the maximum of the values $M^j$, $1\leq j \leq w$. Thus, $q$ is the maximum of all values in $\bm{T}$. The single-step maxT method only rejects hypotheses $\Hy_i$ with $T_i>q$. It follows that single-step maxT rejects no hypotheses. Consequently, the sequential maxT method rejects no hypotheses either.

\item We start with the first step of the maxT procedure. For each permutation or bootstrap sample, we must compute the maximum of the $m$ teststatistics. These maxima are 6.3, 6.0, 6.2, 5.7 and 4.4. We must then find the empirical $(1-\alpha)100\%=80\%$-quantile of these maxima, which is 6.2. We can then reject all hypotheses with test statistics strictly larger than 6.2. There is only one such hypothesis, namely $\Hy_2$, which has a statistic of 6.3. Thus, in the first step, we reject $\Hy_2$.

We now perform the second step. After removing the second row, the maxima are 6.2, 6.0, 4.4, 5.7 and 4.4. The $80\%$-quantile of these maxima is 6.2. We can then reject all hypotheses with test statistics strictly larger than 6.2. This leads to no additional rejections. Thus, maxT only rejects $\Hy_2$. 

In case $\alpha<0.2$, the  $(1-\alpha)100\%$-quantile of the maxima is simply the maximum of the maxima. No test statistic can exceed the maximum of the maxima, so no hypotheses are rejected when $\alpha<0.2$.

\item
Let
$$V^{(1)}\leq ...\leq V^{(|\G|)}$$ be the sorted values $V(gX)$, $g\in \G$. Here $V(gX)= |\N\cap\R(gX)|$, where 
$$ \R^g(X):=\{1\leq i \leq m: T_i(gX)>t\}.$$

Note that if we see the values $V(gX)$ as the ``test statistics'', then Assumption \ref{assgroupinv} holds, i.e., the joint distribution of the 
values $V(gX)$, $g\in \G$, is invariant under all transformations in $\G$ of $X$. Hence, by Theorem \ref{basic}, 
$$\pr(V(X)\leq V^{(k)}(X))\geq 1-\alpha.$$
Finally note that $V^{(k)}(X)\leq \bar{V}(X)$, so that 
$$\pr(V(X)\leq  \bar{V}(X))\geq 1-\alpha,$$
as was to be shown.

\end{enumerate}

\end{sols}

\newpage

\section{False discovery exceedance control} \label{secFDX}

\subsection{Multiple testing criteria involving false discovery proportions} \label{seccriteriaFDP} 
When we control the familywise error rate, we are very strict in the sense that we require that with large probability, there are no false positives at all. In many cases however, we would be willing to incur a small percentage of false positives if this leads to more power. In that case it can be useful to consider multiple testing methods that somehow control the \emph{false discovery proportion} (FDP), which is defined as the fraction of true hypotheses among all the rejections. We define the FDP  to be 0 when there are no rejections.

When we apply a multiple testing method, we will usually denote by $\R\subseteq\{1,...,m\}$ 
the set containing the indices of the rejected hypotheses and we write $R=|\R|$.
As usual, we denote the number of false positives by $V=|\R\cap \N|$.
the FDP then satisfies $$FDP=\frac{V}{R\vee 1},$$
where  ``$\vee 1$'' is added to avoid dividing by 0.

There are several possible approaches to ``keeping the FDP small''. A very popular approach is to consider the \emph{false discovery rate} (FDR), which is  the expected value of the FDP,
$$FDR=\mathbb{E}(FDP).$$
Controlling the FDR means ensuring that the FDR is below some value, which we typically denote by $\alpha\in(0,1)$. Typical choices are $\alpha=0.05$ and $\alpha=0.1$. Likewise we could control the median of the FDP, which means ensuring that 
$$\pr(FDP\leq \gamma)\geq 0.5$$
for some small value $\gamma\in[0,1)$.

If we control the mean or median of the FDP, then we control a central tendency of the FDP, but we do not guarantee that with e.g. $95\%$ confidence, the FDP lies below some small value $\gamma$. Ensuring that with large probability the FDP is at most $\gamma$, is called \emph{false discovery exceedance} (FDX) control (or sometimes simply ``FDP control''). Thus, we then ensure that
$$\pr(FDP\leq \gamma)\geq 1-\alpha$$
for some small value $\alpha$. Note that if we take $\alpha=0.5$, we control the median of the FDP.
Note that a more ``positive formulation'' of FDX control, is that we ensure that with probability at least $1-\alpha$, at least $(1-\gamma)100\%$ of the findings are true discoveries:
$$\pr(TDP\geq 1-\gamma)\geq 1-\alpha,$$
where $TDP:=1-FDP$ is the \emph{true discovery proportion}.

FDX control is the topic of this section.
Providing confidence statements on many FDPs simultaneously, is discussed in \S\ref{secCTPfdp}.
There we will allow picking any set of hypotheses post hoc and obtaining a confidence statement on the corresponding FDP (or equivalently, on the TDP). FDR controlling methods are discussed in \S\ref{secFDR} and \S\ref{secknockoffs}.

\subsection{Lehmann-Romano and related methods}
We will discuss several FDX methods, but we will omit proofs, which  can be found in the cited papers. We first define an FDX method by  \citt{lehmann2005generalizations}.
Suppose that at least one of the following assumptions holds: 
\begin{itemize}
\item the null p-values satisfy $\pr(P_i\leq c)\leq c$ for every $c\in[0,1]$, conditional on the non-null p-values.
\item the null p-values satisfy Simes' inequality. 
\end{itemize}
Note that the first assumption is satisfied if the  $(P_i:i\in \N)$ are independent of the $(P_i:i\in \N^c)$.
If at least one of the above two assumptions holds, then  the p-value-based step-down method (see \S\ref{secupdown})  based on critical values 
\begin{equation} \label{eq:LRpowerful}
c_i = \frac{(\lfloor \gamma i \rfloor +1 )\alpha}{m+\lfloor \gamma i \rfloor +1-i }
\end{equation} 
 provides FDX control 
\citp[][p.1148,1150]{lehmann2005generalizations}.
Recall that in \S \ref{secupdown} we defined what we mean by a step-down method: we need to find the smallest sorted p-value that exceeds its critical value, and then reject all hypotheses with smaller p-values.
The Lehmann-Romano procedure that we just defined is also the  Lehmann-Romano procedure studied in \citt{dohler2024FDX}, who make some improvements for the case that the data are discrete.
The procedure of \citt[][pp.1148, 1150]{lehmann2005generalizations} coincides with the FDX method of \citt{romano2007control} (see \S\ref{secRW}) if within that procedure the step-down k-FWER method from 
\citt[][p.1143]{lehmann2005generalizations} is used. The procedure of \citt[][p.1148,1150]{lehmann2005generalizations} is implemented in the function \emph{continuous.LR()}  from  the R package \emph{FDX} on CRAN. The improvement by \citt{dohler2024FDX} for discrete data is available through the function \emph{discrete.LR()}. Take note that the argument \verb|alpha| required by those R functions is our $\gamma$, and their argument \verb|zeta| is our $\gamma$.
\footnote{The reason why our $\gamma$ is called \emph{alpha} in that package, is that the literature sometimes uses $\alpha$ to denote the ``target FDP''.}

\citt[][p.1152]{lehmann2005generalizations} also provide a method that does not require dependence assumptions at all.
This procedure is the same as their other step-down method just discussed, except that we should decrease the critical values $c_i$ by dividing them by $C_{(\lceil \gamma m \rceil +1)}$, where $C_j := \sum_{i=1}^j (1/i)$. 
Thus the critical values are 
\begin{equation} \label{eq:LRgeneral}
c_i' = \frac{(\lfloor \gamma i \rfloor +1 )\alpha}{(m+\lfloor \gamma i \rfloor +1-i)C_{(\lceil \gamma m \rceil +1)} }
\end{equation} 
For example,
for $j=10$, it holds that  $C_j\approx 2.92$ and 
 for $j=100$, it holds that $C_j\approx 5.187$. 
 
 We have seen that there are two ``Lehmann-Romano (2005)'' methods.
 The second one that we mentioned (with ciritical values $c_i'$) was later uniformly improved in \citt[][Theorem 3.4]{romano2006stepdown}. That method makes no assumptions on the dependence either, and works by decreasing the critical values \eqref{eq:LRpowerful} not by $C_{(\lceil \gamma s \rceil +1)}$, but by a number $D(\gamma, m)\geq 1$, which depends on $\gamma$ and $m$. 
 Thus, the critical values are 
 \begin{equation} \label{eq:RS}
c_i'' = \frac{(\lfloor \gamma i \rfloor +1 )\alpha}{(m+\lfloor \gamma i \rfloor +1-i)D(\gamma, m)}.
\end{equation} 
 For example, for $\gamma=0.1$ and $m=1000$, it holds that  $D(\gamma, m)\approx 3.4179.$ \citt[][p.42]{romano2006stepdown} contains a table that shows $D(\gamma, m)$ for several other values of $\gamma$ and $m$. Clearly, the method from  \citt[][Theorem 3.4]{romano2006stepdown} (which makes no assumptions on the dependence) is still less powerful than \citt[][p.1148,1150]{lehmann2005generalizations} (which does require assumptions on the dependence).
 
 The method from \citt[][Theorem 3.4]{romano2006stepdown} is not implemented in the R package \emph{FDX}. However, note that if we multiply all p-values by $D(\gamma, m)$ and then apply the function \emph{continuous.LR()}, then this is equivalent to applying the method from \citt[][Theorem 3.4]{romano2006stepdown}. This is because  multiplying all p-values by $D(\gamma, m)$ has the same effect as dividing the critical constants by $D(\gamma, m)$.

\subsection{Romano-Wolf} \label{secRW}
The FDX controlling method of \citt{romano2007control} is based on iteratively applying k-FWER methods: the method starts with applying a 1-FWER method, then applies a 2-FWER method, etcetera until a step is reached where some criterion is not met. Romano-Wolf is thus a kind of step-down approach. The k-FWER method that is applied within Romano-Wolf could e.g. be one of the methods discussed in  \S\ref{seckFWER} and \S\ref{secmaxTkFWER}. Thus, Romano-Wolf is a general method, in the sense that there are multiple options for the k-FWER method that is used within Romano-Wolf.

Given some k-FWER method (for $k=1,2,...$), let $\R^k$ be the set containing the indices of the hypotheses rejected by the k-FWER method, and let $\R^0=\emptyset$. The Romano-Wolf method is defined as follows:
\begin{enumerate}
\item Let $k^*$ be the smallest $k\geq 1$ for which $k/(|\R^k|+1)>\gamma$ (if there is no such $k$, reject all hypotheses ). 
\item Reject all hypotheses with indices in $\R^{k^*}$.
\end{enumerate}
An equivalent formulation is in  Algorithm \ref{alg:RW}.
The intuition behind the method is as follows: if we reject the top $|\R^k|+1$ hypotheses, then (with probability at least $1-\alpha$) there are at most $k$ false positives, so the FDP is at most  $k/(|\R^k|+1)$. If this value is larger than $\gamma$, we stop, because we want the FDP to be at most $\gamma$.


\begin{algorithm}[H] 
\caption{Algorithm for the Romano-Wolf method. Algorithm 4.1 from \citt{romano2007control} is the same, except that it checks whether $|\R^k|<k/\gamma -1$, which is equivalent to checking whether $k/(|\R^k|+1)>\gamma$.} \label{alg:RW}
\begin{algorithmic}[1]
\STATE $k\leftarrow0$
\FOR{$k=1,\ldots,m$}   
      \STATE Compute $\R^k$
    \IF{$k/(|\R^k|+1)>\gamma$}
      \STATE \textbf{break}
    \ENDIF
\ENDFOR
	\STATE \textbf{return} $\R^k$
\end{algorithmic}
\end{algorithm}

It turns out that if in Romano-Wolf we use the k-FWER method from Theorem \ref{thmkFWERstepdown}, then we obtain the Lehmann-Romano method discussed above (the most powerful one, with critical values defined in \eqref{eq:LRpowerful}).

In \citt{romano2007control} it is proved that Romano-Wolf provides FDX control in an asymptotic sense. Further theoretical results are discussed in \citt{delattre2015new}. In practice, the Romano-Wolf method tends to be valid if the underlying k-FWER method is valid.

 \subsection{Application  to  property sales data}
 We continue our analysis from \S\ref{secappsales}, where we computed 31 p-values for the numerical predictors of the linear regression model predicting sales prices, see Figure \ref{fig:31pvalues}.  We took $\alpha=0.05$ there and saw that Bonferroni and Holm both rejected 10 hypotheses. If we keep $\alpha$ and let $\gamma=0.1$ and apply the Lehmann-Romano method with critical values defined in \eqref{eq:LRpowerful}, we also reject 10 hypotheses. This is because we require the FDP to be at most $\gamma$, which means that we allow almost no false hypotheses, given the fact that we are merely testing 31 hypotheses. If we increase $\gamma$ to 0.2, we reject 12 hypotheses. This is not a real improvement, given that  we already rejected 10 hypotheses with Bonferroni. 
 
 We see that the results based on Lehmann-Romano are somewhat underwhelming. Part of the reason is ``bad luck'': in \S\ref{secribo} the improvement with Lehmann-Romano is more impressive.
Another reason is that in sense Lehmann-Romano is quite strict, since it needs to account for  worst-case scenarios regarding the dependence structure of the p-values. (It makes some  assumptions on the dependence, but these are limited.) 
This is somewhat comparable to the way in which Bonferroni and Holm tend to strict compared to maxT.
The Romano-Wolf method can be potentially much more powerful than Lehmann-Romano, if a resampling-based k-FWER method is used within Romano-Wolf. Based on the existing literature however, it is not fully clear how such resampling-based methods can be used for testing coefficients in a model.
 Hence, we now move to a different dataset, where we can use a resampling-based version of Romano-Wolf.

  \subsection{Application  to  riboflavin production data} \label{secribo}
We analyze a dataset about riboflavin (vitamin B2) production with  bacteria called \emph{B. subtilis}. This dataset
is freely available \citp{buhlmann2014high}. It contains normalized measurements
of expression rates of 4088 genes from $n = 71$ samples. 
Further, the dataset contains the corresponding 71 measurements of the logarithm of the riboflavin production
rate.
For each $1 \leq i \leq  4088$, we are interested in the hypothesis  $\Hy_i$ that the riboflavin production
rate was stochastically independent
of the expression level of gene $i$. 
For the parametric methods, we computed p-values using two-sided correlation tests as in \S\ref{secexcortest}. 

 We took $\alpha=0.05$ again. Bonferroni and Holm rejected 53 hypotheses, while Hommel rejected 54 hypotheses. Next, we kept $\alpha=0.05$ and took $\gamma=0.1$ and applied the FDX method of Lehmann Romano (based on critical values \eqref{eq:LRpowerful}). This method resulted in 103 rejections. Since $\gamma$ was 0.1, we know that with probability at least $1-\alpha=0.95$, at most 10 of these 103 rejections are false positives. Thus  with probability at least $1-\alpha=0.95$, at least  93 of these 103 rejections are true positives. 
 
 We can have even more power than Lehmann-Romano if we use a permutation-based multiple testing method. Indeed, we can permute the  observations of the roboflavin production rate and construct a matrix of resampled test statistics as in \S\ref{secdefmaxT}. We can then perform permutation-based methods such as maxT and the k-FWER verision of maxT from \S\ref{secmaxTkFWER}. The test statistics that were used were the absolute values of the empirical correlations between the production rates and gene expressions.
  The maxT method rejected 74 hypotheses. As expected, this is more than what Holm rejects.   
 We can also use the FDX method of Romano and Wolf, where as the k-FWER method we take the k-FWER version of maxT from \S\ref{secmaxTkFWER}.  We took $w=1000$, i.e, we used 999 random permutations plus the original data.
 The Romano-Wolf method based on the single-step  k-FWER method from \S\ref{secmaxTkFWER} rejected 201 hypotheses. Note that it was quite expected that the resampling-based Romano-Wolf method would reject more than Lehmann-Romano: the k-FWER method used within Romano-Wolf is relatively powerful, because it takes into account the dependence structure of the gene expression levels.
 From Romano-Wolf we conclude that with probability at least $1-\alpha=0.95$, among the 201 hypotheses with the largest test statistics, at least $90\%$, so at least 181 hypotheses, are true discoveries.

 \subsection{Exercises}
\S\ref{secFDR} contains some exercises linking FDX control and FDR control.

 \begin{enumerate}
 
 \item Provide a simple expression for  the last critical constant, $c_m$, of the most powerful  method among  the two FDX methods from \citt{lehmann2005generalizations}.

 \item
Suppose $m=7$ and $\alpha=\gamma=0.2$. Suppose the sorted p-values are $0.01$, $0.02$, $0.04$, $0.06$, $0.14$, $0.36$ and $0.74$.
Determine how  many hypotheses are rejected by the most powerful method of the two FDX methods from \citt{lehmann2005generalizations}.


 \item Note that the  FDP is a random variable, which takes values in a subset of $[0,1]$.
 Assume for simplicity that the FDP is symmetric about its mean (in practice it is often right-skewed). Let $\gamma\in(0,1)$ and $ \alpha\in(0,1/2]$.
 Prove that $\pr(FDP\leq \gamma)\geq 1-\alpha$ implies that $\mathbb{E}(FDP)\leq \gamma$. 
 Discuss informally that that is not generally the case  the other way around.
 Further, discuss informally that the first statement is not generally true either when the FDP is right-skewed and $\alpha$ is close to 0.

 \item Suppose that all p-values are independent of each other. Suppose that the p-values corresponding to the true hypotheses all have the standard uniform distribution. Suppose that the p-values corresponding to the false hypotheses all have the same distribution as well (not necessarily uniform). Let $t\in (0,1)$ and suppose we simply reject all hypotheses with p-values at most $t$. Now suppose we keep adding hypotheses to the problem (so that $m\rightarrow\infty$), and that  the fraction $\pi_0$ of true hypotheses converges to some constant in $(0,1)$ as $m\rightarrow\infty$.
Write $FDP=FDP(m)$ to ephasize the dependence of the FDP on $m$. 

Prove that as $m\rightarrow\infty$, $FDP(m)$ converges to an (unknown) constant $c\in[0,1]$. Discuss that this  means that if we consider $\alpha, \gamma\in(0,1)$ and  choose $t$ such that $\lim_{m\rightarrow\infty} \mathbb{E}(FDP(m))< \gamma$, then we have $\lim_{m\rightarrow\infty} \pr(FDP(m)\leq \gamma)\geq 1-\alpha$.
 (Recall that ensuring that $\mathbb{E}(FDP)$ is bounded,  is called false discovery rate control.)
Discuss informally whether  these statements are  still necessarily true if the p-values are dependent on each other. 
 \end{enumerate}

\begin{sols}
\subsection{Solutions}

\begin{enumerate}
\item
We have $c_m = \frac{(\lfloor\gamma m\rfloor +1 )\alpha}{m+ \lfloor\gamma m\rfloor +1-m } =
\frac{(\lfloor\gamma m\rfloor +1 )\alpha}{\lfloor\gamma m\rfloor +1} =\alpha.$ Thus, the last critical value is simply $\alpha$, just like in the methods of Bonferroni and Holm and in Simes' global test.
  
\item
For every $1\leq i \leq 7$, the critical value $c_i$ satisfies
$$c_i = \frac{(\lfloor \gamma i \rfloor +1 )\alpha}{7+\lfloor \gamma i \rfloor +1-i }.$$
For $1\leq i \leq 4$, $\lfloor \gamma i \rfloor=0$, and for $5\leq i \leq 7$, $\lfloor \gamma i \rfloor=1$. Thus, we have
\begin{align*}
c_1=& \frac{0.2}{7+1-1} = 0.2/7 =1/35,\\
c_2=& \frac{0.2}{7+1-2} = 0.2/6 =1/30,\\
c_3=& \frac{0.2}{7+1-3} = 0.2/5 =1/25,\\
c_4=& \frac{0.2}{7+1-4} = 0.2/4 =1/20,\\
c_5=& \frac{0.4}{7+1+1-5} = 0.4/4 =1/10,\\
c_6=& \frac{0.4}{7+1+1-6} = 0.4/3 =2/15,\\
c_7=& \frac{0.4}{7+1+1-7} = 0.4/2 =1/5.
\end{align*} 
Note that $p_{(i)}\leq c_i$ holds for all $1\leq i \leq 3$, but not for $i=4$. Thus, 3 hypotheses are rejected, since the method is a step-down procedure. (Note that we could have stopped computing the critical values after computing $c_4$.)

\item
$\pr(FDP\leq \gamma)\geq 1-\alpha \geq  0.5$ means that the median of the FDP is at most $\gamma$. By symmetry, $\mathbb{E}(FDP)$ must also be at most $\gamma$.
 
If $\mathbb{E}(FDP)\leq \gamma$, there can be a $c>\alpha$ such that with probability at least $c$,  the FDP is strictly larger than $\gamma$. (When the FDP is right-skewed but $\alpha$ is small, this is often true as well.)
Then $\pr(FDP\leq \gamma)\leq 1-c< 1-\alpha.$
  
\item 
We have $FDP(m)=V(m)/(V(m)+S(m))$, where $S(m)$ is the number of correct rejections $|\{1\leq i \leq m: \Hy_i\text{ is false and } P_i\leq t \}|$.
 
Because of the independence of the p-values and because $\pi_0$ is asymptotically proportional to $m$, $V(m)/m$ converges to a constant $c_1$, say, and likewise $S(m)/m$ converges to some constant $c_2$. We have 
$$FDP(m) = \frac{V(m)/m}{V(m)/m+S(m)/m}.$$
As $m\rightarrow\infty$, the numerator converges to $c_1$ and the denominator converges to $c_1+c_2$, so that  $\lim_{m\rightarrow\infty}FDP(m) = c_1/(c_1+c_2)$, which meas that $FDP(m)$ converges to a constant $c$.
 
Since $\lim_{m\rightarrow\infty} \mathbb{E}(FDP(m))< \gamma$, it follows that $c<\gamma$.
Hence,  $\lim_{m\rightarrow\infty} \pr(FDP(m)< \gamma)=1 \geq 1-\alpha$.
 
When there is dependence among, say, the null p-values, then  the variance of $V$ does not necessarily vanish relative to its mean, so $V(m)/m$ does not converge to a constant. In general, when there is dependence among the p-values,  $FDP(m)$ may not converge to a constant but may vary substantially about its mean regardless of $m$. In that case, if $\alpha$ is small, requiring that  the tail probability $\pr(FDP>\gamma)$ is at most $\alpha$ is often a stricter requirement than requiring that the mean $\mathbb{E}(FDP)$ is at most $\gamma$ (also for large $m$).

\end{enumerate}
\end{sols}

 \newpage
 
 \section{Simultaneous confidence on all FDPs} \label{secCTPfdp}

\subsection{Motivation}
In \S\ref{seccriteriaFDP} we discussed a few  multiple testing criteria that  involve false discovery proportions (FDPs). FDX and FDR methods (the latter are discussed in \S\ref{secFDR}) have the limitation that the user must make all data analysis decisions before seeing the data, and these methods only provide a statement on a single set of hypotheses. FDX and FDR methods provide the user with a fixed set of rejected hypotheses, and the user has no freedom to explore other sets of hypotheses after that. For example, if an FDX method rejects 15 hypotheses, then it provides information on these hypotheses only. The method then provides no information on the top 5 hypotheses (i.e., the 5 hypotheses with the most extreme test statistics or p-values) or the top 20 hypotheses --- apart from any information that might be deduced from the method's statement about the 15 rejected hypotheses.

For example, suppose $\alpha=0.1$ and $\gamma=0.2$ and the FDX method rejects 15 hypotheses. Then we know with $90\%$ confidence that among the top 15 rejections, there are at most 3 false positives. This suggest that it is likely that among the top 5 rejections there are strictly fewer than 3 false positives. However, the FDX method does not provide any statement on that. Likewise, if we look at the top 20 hypotheses instead of the rejected 15 hypotheses, then the method tells us nothing about these 5 additional hypotheses. Thus, we can only say that with $90\%$ confidence, there are at most 3+5 true hypotheses among the top 20 hypotheses.
As another example, suppose that among the 15 rejections, 4 of them belong together in some way, or that for some other reason we are interested in these 4. Then we cannot say anything about these 4 hypotheses except that with large probability at most 3 of them are true. As a final example, suppose that we are interested in some subset of 3 hypotheses among the rejected 15. Then we cannot say anything about the number of false hypotheses among these 3. The general problem here is that FDX methods only provide a confidence statements on a single set of hypotheses.

Recall the notation $\C=\{\text{nonempty subsets of }\{1,...,m\}\}$.
In this chapter, we define a general approach that provides confidence statements on the FDPs of \emph{all} $|\C|=2^m-1$ sets of hypotheses. By the FDP of a set $\I\in \C$ we will simply mean the fraction of true hypotheses in $|\I|$, i.e., $|\N\cap\I|/|\I|$. Thus, we view the FDP as a property of a set of hypotheses; there are thus $2^m-1$ FDPs in total.
 The methods in this section  provide a $(1-\alpha)$-confidence upper bound $t_{\alpha}(\I)$ for the the number of true hypotheses $|\N\cap\I|$ in every set $\I$. Dividing by $|\I|$ gives a confidence bound for  $|\N\cap\I|/|\I|$, the FDP  for the set $\I$. 
In fact, we will guarantee  much more, because the bounds $t_{\alpha}(\I)$ are all simultaneously valid with probability at least $1-\alpha$, i.e., 
\begin{equation} \label{GSresult}
   \mathbb{P} \Bigg[\bigcap_{\I\in\subs}\Big\{     |\N\cap \I| \leq t_{\alpha}(\I)   \Big\} \Bigg]\geq 1-\alpha.
   \end{equation}

   Because with large probability the bounds $t_{\alpha}(\I)$ are all valid, we can in fact pick any $\I$ --- even after seeing the data --- and obtain a valid confidence statement on the number of true hypotheses in $\I$. This allows users the freedom to keep looking at various sets of hypotheses after looking the data, without invalidating the confidence statements. Thus, users then have much more flexibility than with FDX and FDR methods, which tell the user which hypotheses to look at.
   
   Note that since $t_{\alpha}(\I)$ is a $(1-\alpha)$-confidence upper bound for $|\N\cap\I|$, it follows that 
   \begin{equation} \label{eq:dalpha}
   d_{\alpha}(\I):=|\I|- t_{\alpha}(\I)
   \end{equation}   
    is a $(1-\alpha)$-confidence lower bound for $|\N^c\cap \I|$, the number of false hypotheses (true discoveries) among $\I$. ($t_{\alpha}$ stands for ``true'' and $d_{\alpha}$ stands for ``discoveries''.) Thus, the following is an equivalent formulation of \eqref{GSresult}:
   \begin{equation} \label{GSresultTDP}
   \mathbb{P} \Bigg[\bigcap_{\I\in\subs}\Big\{     |\N^c\cap \I| \geq  d_{\alpha}(\I)   \Big\} \Bigg]\geq 1-\alpha,
   \end{equation}
Such a formulation in terms $d_{\alpha}(\I)$ is sometimes preferred since it gives a more positive perspective: we ensure that there are at least a certain number of true discoveries among $\I$. Thus, instead of focussing on what goes wrong, we then focus on the correct rejections. Nevertheless, this formulation is completely equivalent to \eqref{GSresult}.

\subsection{General construction of the simultaneous bounds}
We now define bounds $t_{\alpha}(\I)$ that have the property \eqref{GSresult}. The construction is closely related to the general closed testing principle that we covered in  \S\ref{secCTPgeneral}; in fact, we will start by considering a closed testing procedure. Suppose that for every $\I\in\C$ we have defined a valid local test $\phi_{\I}$. Recall that the corresponding closed testing procedure (CTP) rejects all hypotheses $\Hy_{\I}$ for which $\psi_{\I}=1$, where  
\begin{align*}
\psi_{\I} = &  \min\{\phi_{\J}: \I\subseteq \J\in \subs\}\\
= & \mathbbm{1}\big(\forall  \I\subseteq \J\in \subs:  \phi_{\J}=1 \big).
\end{align*}
 As a shorthand, write
$$\X := \{\I\in \C: \Hy_{\I}   \text{ is rejected by the CTP}\} =  \{\I\in \C:  \psi_{\I}=1\}.$$
For every $\I\in \C$  we now define $t_{\alpha}(\I)$ as follows \citp{goeman2011multiple}:
\begin{equation} \label{eq:talpha}
t_{\alpha}(\I):=\max\{|\J|: \J \subseteq \I \text{ and }  \J\not \in \X\}.
\end{equation}
If there is no nonempty $\J \subseteq \I$ with  $\J\not \in \X$, we consider this maximum to be 0. 

\begin{theorem}\label{thmsimCTPbounds}
For every $\I\in \C$, define $t_{\alpha}(\I)$ as in \eqref{eq:talpha}. Then property \eqref{GSresult} --- and equivalently, property \eqref{GSresultTDP} --- is satisfied.
\end{theorem}
\begin{proof}
With probability at least $1-\alpha$, $\Hy_{\N}$ is not rejected by its local test. Suppose $\Hy_{\N}$ is not rejected by its local test. Then,  for every $\I \in \C$ it holds that $\N\cap \I\not\in \X$, so that $|\N\cap \I| \leq t_{\alpha}(\I)$. This proves inequality \eqref{GSresult}.
\end{proof}
We see that the proof is  short, just like the proof from \S\ref{secCTPgeneral} of the fact that a closed testing procedure controls the FWER. Both proofs are completely based on the fact with probability $1-\alpha$, $\Hy_{\N}$ is not rejected by its local test. As long as that is the case, all statements of the closed testing procedure and all the bounds $t_{\alpha}(\I)$ are correct. 

The following formulation of $t_{\alpha}(\I)$ (by \citealp{genovese2006exceedance}) can also be useful.

\begin{fact} \label{altformtalpha}
An equivalent  definition of the bounds $t_{\alpha}(\I)$ is the following:
\end{fact}
\vspace{-0.4cm}
\begin{equation} \label{eq:talpha2} 
t_{\alpha}(\I)=\max\{|\I\cap \K|: \K\in\C  \text{ and }\phi_{\K}=0\}.
\end{equation}
\begin{proof}
We first show that the quantity in    \eqref{eq:talpha} is at most the quantity in    \eqref{eq:talpha2}. To do this, we must show that for every $ \J \subseteq \I$ with $ \J\not\in \X$, there is a $\K\in \C$ for which $\phi_{\K}=0$ and such that $|\J|  \leq  |\I\cap \K|$.
To prove that, note that if $ \J \subseteq \I$ and $ \J\not\in \X$, then apparently there is some $ \J\subseteq \K\in \C$ for which $\phi_{\K}=0$, and then $|\J| =|\J\cap \K| \leq  |\I\cap \K|$. 

We now  show that the quantity in    \eqref{eq:talpha2} is at most the quantity in    \eqref{eq:talpha}.
To do this, we must show that for every $\K\in \C$ for which $\phi_{\K}=0$, there is a $ \J \subseteq \I$ with $ \J\not\in \X$ and $|\J|  \leq  |\I\cap \K|$.
To prove that, note that if $\K\in \C$ is such that $\phi_{\K}=0$, then we can take $\J=\I\cap \K$ and then we have $\J\subseteq\I$, $\J\not\in \X$ and $|\I\cap\K|= |\J|$.
\end{proof}

Which formulation of $t_{\alpha}(\I)$ is more useful or insightful, might depend on the situation. The first formulation, \eqref{eq:talpha}, connects the method to the large literature on closed testing. The second formulation, \eqref{eq:talpha2},  is in a sense the simplest, since it does not refer to the closed testing procedure. However, note that if we  use  the second formulation, we need to check for each $\K\in \C$ whether $\phi_{\K}=0$, which is very computationally expensive. In the first formulation, we  `only' need to check something for the $\J$ with $\J \subseteq \I$. This might be useful  in case we have a fast algorithm for checking whether $\J\in \X$. 

\subsection{Consonant procedures, e.g. Bonferroni-based local tests} \label{secexBonftalpha}
As a first example, suppose the local tests $\phi_{\I}$ are Bonferroni tests, i.e., 
\begin{equation}\label{localtBonf}
\phi_{\I} = \mathbbm{1}( \min_{i\in \I}P_i\leq \alpha/|\I|).
\end{equation}
We saw in \S\ref{secCTPHolm} that the CTP  based on these local tests is the Holm method; more precisely, they reject the same elementary hypotheses (and since the CTP is consonant, it follows which intersection hypotheses are rejected by the CTP).
Now, let us see what simultaneous FDP bounds $t_{\alpha}(\I)$ we get based on these local tests.


\begin{proposition} \label{proptalphaisHolm}
For every $\I\in \C$, the bound $t_{\alpha}(\I)$ is simply the number of  hypotheses in $\I$ that are not rejected by Holm's method.
\end{proposition}

\begin{proof}
This result is a corollary of Propositions \ref{boundscons} and \ref{Bonfcons}. We additionally provide a direct proof here.
Renumber the hypotheses such that $P_1\leq ...\leq P_m.$ Let $\J$  be the set of all hypotheses in $\{1,...,m\}$ that are not rejected by Holm. We must show that $t_{\alpha}(\I)=|\I\cap\J|.$
By  formula \eqref{eq:talpha2}, we have
\begin{align} 
t_{\alpha}(\I) = & \max\{|\I\cap \K|: \K\in\C  \text{ and }\phi_{\K}=0\}   \notag  \\ 
= & \max\{|\I\cap \K|: \K\in\C  \text{ and }\min_{i\in \K}P_i> \alpha/|\K|\}. \label{eq:TDPBonfproof}
\end{align}


In case $\J=\emptyset$, all hypotheses are rejected by Holm, and then $t_{\alpha}(\I)=0$. Indeed, pick some  $\K\in\C $ and let $k=\min\K.$ Then  
\begin{equation} \label{eq:trick}
\min_{i\in \K}P_i =P_k \leq\alpha/(m+1-k) =\alpha/|\{k,...,m\}| \leq  \alpha/|\K|.
\end{equation}
Thus, in case $\J=\emptyset$, we see that $t_{\alpha}(\I)=0=|\I\cap\J|.$

In case $\J\neq\emptyset$,   choose $j$ such that $\J=\{j,...,m\}$. 
Note that  $\min_{i\in \J}P_i = P_j> \alpha/(m+1-j)= \alpha/|\J|$. 
Thus, \eqref{eq:TDPBonfproof} is at least 
$ |\I\cap \J|$.

We are done if we show that \eqref{eq:TDPBonfproof} is also at most $|\I\cap \J|$. If $\I\cap \J=\I$, then  that immediately follows. Otherwise consider any $\K\in \C$ such that $|\I\cap \K|> |\I\cap \J|$. We must show that $\min_{i\in \K}P_i\leq \alpha/|\K|$.
To do this, note that at least one hypothesis in $\K$ is rejected by Holm. Write  $k=\min\K$, so that $\Hy_k$ is rejected by Holm. Then exactly as in \eqref{eq:trick} we find $\min_{i\in \K}P_i \leq \alpha/|\K|,$ as was left to show.
\end{proof}

We thus see that in case Bonferroni-based local tests are used, the bounds $t_{\alpha}(\I)$ follow directly from knowing which elementary hypotheses are rejected. It turns out that this is always the case if  the closed testing procedure is consonant. Recall that the CTP is called consonant if for every $\I\in \X$, there is an $i\in \I$ with $\{i\}\in \X$. 
Note that for every CTP, if $\{i\}\in \X$ then all $\I$ that contain $i$ are in $\X$. Thus, if the CTP is consonant, all intersection hypotheses that are rejected, are those that logically have to be false if the rejected elementary hypotheses are false. It may not be surprising that $t_{\alpha}(\I)$ is then simply the number of $i\in\I$ with $i\in \X$. This is true  for all consonant procedures.
We now prove that.
Then, we show that the CTP based on Bonferroni-based local tests is  consonant. These two results together imply  Proposition \ref{proptalphaisHolm}.

\begin{proposition} \label{boundscons}
Consider a consonant CTP. Then $t_{\alpha}(\I)$ is simply the number of elementary hypotheses in $\I$ that are not rejected by the CTP.
\end{proposition}
\begin{proof}
By   formula \eqref{eq:talpha},
$$t_{\alpha}(\I) = \max\{| \J|: \J\subseteq \I \text{ and }\J\not\in \X\}.$$
By definition of a CTP,  if $\J\not\in \X$, then for all $i\in \J$, $\{i\}\not\in \X$. The other way around, if for all $i\in \J$ it holds that $\{i\}\not\in \X$, then $\J\not\in \X$ due to consonance. Thus, the above equals 
$$\max\{|\J|: \J\subseteq \I \text{ and } \forall i\in \J:\{i\}\not\in \X\}.$$
This is simply the number of elementary hypotheses in $\I$ that are not rejected by the CTP.
\end{proof}

\begin{proposition} \label{Bonfcons}
The CTP based on Bonferroni local tests \eqref{localtBonf} is consonant.
\end{proposition}
\begin{proof}

Suppose $\I\in \X$.  
Pick $j\in \I$ such that $P_j=\min\{P_i:i\in \I\}$. We will be done if we show that $\{j\}\in \X$.

We must  show that for all  $\J\in \C$ with $j\in \J$, $\phi_{\J}=1$.
Consider  a $\J\in \C$ with $j\in \J$. Firstly,  note that $\phi_{\I\cup\J}=1$, since $\I\in\X$. 
Further, 
$\min\{P_i: i\in \J\} =\min\{P_i: i\in (\I\cup \J)\}  \leq \alpha /|\I\cup \J| \leq \alpha /| \J| $, which means that $\phi_{\J}=1$ as we needed.
\end{proof}

Another example of a consonant CTP, is the CTP that corresponds to the maxT method (see \S\ref{secmaxTCTP}). In the next section, we will look at a  CTP that is not consonant.

\subsection{Nonconsonant procedures, e.g. Simes local tests}
For nonconsonant (``dissonant'') methods, the  bounds $t_{\alpha}(\I)$ do not generally follow trivially from the set of rejected elementary hypotheses. 
An example of a nonconsonant CTP (see exercises) is the one based on Simes local tests 
\begin{equation}\label{localtSimes}
\phi_{\I} = \mathbbm{1}\Big( \bigcup_{i\in \I} \{P^{\I}_{(i)}\leq \frac{i\alpha}{|\I|} \}  \Big).
\end{equation}
We mentioned in \S\ref{secHommel} that the CTP based on these local tests is  Hommel's method --- more precisely, Hommel finds the elementary hypotheses that are rejected by this CTP. Since the CTP is not consonant, we can potentially make statements that do not trivially follow from Hommel's rejections of elementary hypotheses. For example, it may happen that Hommel rejects nothing, but we can still infer that $|\N^c\cap\I|>0$ for some $\I\in \C$.

Consider the CTP based on the  Simes local tests. Naively computing the $t_{\alpha}(\I)$ for all $\I\in \C$ is computationally infeasible for moderate or large $m$. 
A fast, exact computational shortcut for this CTP is presented in \citt{goeman2019simultaneous}. The method is implemented in the R package \emph{Hommel} on CRAN \citp{hommel}.
The R package  can be used as follows. First create an object of class \verb|hommel| using
$$\verb|hom <- hommel(p)|$$ where \verb|p| contains the $m$ p-values. Then run
$$\verb|discoveries(hom, I, alpha=0.05)|$$
where \verb|I| is the set of (indices of) hypotheses $\I$ for which we want to compute the bound $d_{\alpha}(\I)=|\I|-t_{\alpha}(\I).$ We then obtain $d_{\alpha}(\I)$.

If we apply this method (with $\alpha=0.05$ again) to the 4088 p-values from \S\ref{secribo}, we clearly obtain valuable information on top of what Hommel's FWER method provided. Indeed, we saw that Hommel's method rejected 54 hypotheses. However, if we compute $d_{\alpha}(\I)$ with $\I=\{1,...,m\}$, we get 104 rejections. Thus, we know that with probability at least $0.95$, there are at least 104 false hypotheses. Further, if we compute $d_{\alpha}(\I)$ where $\I$ contains the hypotheses with the 100 smallest p-values, then we get $d_{\alpha}(\I)=90$; thus,  among the 100 hypotheses with the smallest p-values, at least 90 are false. Importantly, all such statements are \emph{simultaneously} true with probability at least $0.95$.

If we consider the p-values from \S\ref{secappsales}, then unfortunately the bounds $d_{\alpha}(\I)$ provide no additional information in addition to the set of elementary hypotheses rejected by Hommel. Thus, we see that the information that we get on top of the rejected elementary hypotheses can sometimes be a lot, and sometimes nothing.

Other examples of nonconsonant CTPs are procedures with local tests based on sums of test statistics, including Fisher combinations.  As is the case for FWER and FDX control, sometimes we can gain more powerful simultaneous FDP statements by using approaches  based on permutations or bootstrapping. Examples are the methods in \citt{hemerik2019permutation, blain2022notip, andreella2023permutation, vesely2023permutation}.

\subsection{Exercises}

\begin{enumerate}

\item Consider Figure \ref{fig:ctp}, where crosses indicate the intersection hyptheses that are rejected by their local tests. Suppose the local tests have significance level $\alpha=0.05$.

 (a) What can we  say with $95\%$ confidence about the total number of false hypotheses among the three?

Now suppose that $\Hy_2\cap\Hy_3$ is also  rejected by its local test, in addition to the hypotheses indicated with a cross.

(b) What can we then say about the total number of false hypotheses in $\Hy_2\cap\Hy_3$?

(c) What can we then say about the total number of false hypotheses?

\item
Suppose we take $\alpha=0.01$ and are interested in four hypotheses $\Hy_1,...,\Hy_4$ and the $2^4-1$ intersections. Consider some CTP for these hypotheses. Suppose that all intersection hypothes are rejected by the CTP, except $\Hy_{12}$, $\Hy_1$ and $\Hy_2$.

(a) What can we say with $99\%$ confidence about the total number of false hypotheses among $\Hy_1,...,\Hy_4$? 

(b) What can we say with $99\%$ confidence about the total number of false hypotheses among $\Hy_3$ and $\Hy_4$?

\item 
Consider elementary hypotheses $\Hy_1,...,\Hy_m$ and a corresponding CTP. Prove that the following statements are equivalent.
\begin{itemize}
\item The procedure is consonant.
\item For every $\I\in\C$, $d_{\alpha}(\I)>0$ implies there exists an $i\in\I$ for which $d_{\alpha}(\{i\})=1$.
\end{itemize}


\end{enumerate}


\begin{sols}
\subsection{Solutions}
\begin{enumerate}
\item (a) For the first question, we must find $t_{\alpha}(\{1,2,3\})$, which is defined by 
$$\max\{|\J|: \J \subseteq \{1,2,3\} \text{ and }  \J\not\in \X\}.$$
Note that $\{1,2,3\}\in \X$, so this maximum is smaller than 3. Note that  $\{2,3\}\not\in \X$, so this maximum is 2. Thus,  with probability at least $0.95$ the number of true hypotheses among the three is at most 2, i.e., there is at least one false hypothesis among these three.

(b)  Note that in this case $\Hy_1\cap\Hy_2$, $\Hy_1\cap\Hy_3$ and $\Hy_2\cap\Hy_3$ are all rejected by the CTP. It follows that $t_{\alpha}(\{1,2,3\})$ is at most 1. It is not 0, since  $\{3\}\not\in \X$. Hence 
with probability at least $0.95$ the number of true hypotheses among the three is at most 1, i.e., there is are least 2 false hypotheses among these three.

(c) For the third question, we look at $t_{\alpha}(\{2,3\})$, which is 1. Hence, with probability at least $0.95$, the number of true hypotheses among these two is also at most 1.

\item (a) 
We must compute $$d_{\alpha}(\{1,...,4\})= 4- t_{\alpha}(\{1,...,m\})=$$
$$4-\max\{|\J|: \J \subseteq \{1,...,4\} \text{ and }  \J\not \in \X\}=$$
$$4-|\{1,2\}|=2.$$
We see that with probability at least $0.99$, there are at least two false elementary hypotheses.

(b)
Note that $$t_{\alpha}(\{3,4\})=$$
$$2-\max\{|\J|: \J \subseteq \{3,4\} \text{ and }  \J\not \in \X\}.$$
There is no $\J \subseteq \{3,4\}$ for which $\J\not \in \X$. Hence we know that with probability at least $0.99$,  $\Hy_3$ and $\Hy_4$ are both false.

\item 
We first show that the first statement implies the second one.
$d_{\alpha}(\I)>0$ means $t_{\alpha}(\I)<|\I|$. Hence, apparently, $\I\in \X$, since otherwise we would have $t_{\alpha}(\I)=|\I|$. Since the procedure is consonant, this means that there must be an elementary hypothesis $i\in \I$ that is rejected by the CTP, i.e., $\{i\}\in \X$. This means that $t_{\alpha}(\{i\})=0$, so $d_{\alpha}(\{i\})=1-0=1.$

We now show that the second statement implies the first one. Consider $\I\in \X$. We must show there is an $i\in \I$ that is rejected by the CTP. Note that $t_{\alpha}(\I)<|\I|$, so $d_{\alpha}(\I)>0$, so there is an $i\in \I$ with $d_{\alpha}(\{i\})=1$, so $t_{\alpha}(\{i\})=0$, which implies that  $\{i\}\in \X$. We conclude that the procedure is consonant.

\end{enumerate}

\end{sols}


\newpage

\section{False discovery rate control with Benjamini-Hochberg} \label{secFDR}

Recall that the false discovery rate (FDR) is defined as $\mathbb{E}(FDP)$.
The by far  most popular and well-known FDR controlling method is the Benjamini-Hochberg (BH) method. The method was proposed in \citet{benjamini1995controlling}. There the authors prove that the method controls the FDR when the null p-values are independent of each other. In fact, they show that the FDR is then at most $\pi_0\alpha$, where $\pi_0:=|\N|/m$ is the fraction of true hypotheses among all hypotheses. (Unfortunately,  $\pi_0$ is usually not exactly known.) In \citet{benjamini2001control} it is shown that this still holds under certain dependence structures of the p-values (see below). That paper also provides a variant of the method, which controls the FDR regardless of the dependence structure, but is less powerful.

BH is defined as follows. Let $k$ be the largest number $1\leq i \leq m$ for which $P_{(i)}\leq \frac{i}{m}\alpha$, if there is such an $i$, and reject the hypotheses corresponding to $P_{(1)},...,P_{(k)}$. Otherwise, reject nothing. Thus, BH is simply a step-up method with critical values $c_i=\frac{i}{m}\alpha$,  $1\leq i \leq m$. This means that BH-adjusted p-values can be computed as explained in \S\ref{secadj}. 

\subsection{Theory}

When $m$ is large and a substantial proportion of the hypotheses are false, BH can be much more powerful than FWER methods such as Bonferroni, Holm and Hommel. For example, in BH, the $(m/10)$-th sorted p-value is compared with $\alpha/10$ (assuming $m$ is a multiple of 10), while Bonferroni compares all p-values with $\alpha/m$.
Of course, ensuring that $FDR\leq\alpha$ is a much weaker guarantee than ensuring that $FWER\leq\alpha$.
Also, ensuring  that the FDR is below $0.05$ is weaker than ensuring that the FDP is below $0.05$ with high probability (FDX control), see some of the exercises.

BH often controls the FDR in practice, especially when two-sided p-values are used.
BH is known to control the FDR  when the \emph{PRDS} assumption is satisfied.  PRDS stands for \emph{positive regression dependence on the subset} $\N$. The PRDS assumption roughly  says that all p-values depend on the null p-values in a positive way. 
To define the PRDS assumption, we need the following concept. A set $\A\subseteq\reals^m$ is called \emph{increasing} if $a\in \A$ implies that $b\in \A$ for all $b\in \reals^m$ satisfying $b\geq a$ (meaning $b_i\geq a_i$ for all $1\leq i \leq m$).

\begin{definition}
Consider p-values $P_1,...,P_m$. They satisfy \emph{positive regression dependence on the subset} $\N$ (\emph{PRDS}) if for every $i\in \N$ and increasing set $\A\subseteq \reals^m$, the function $x\mapsto \pr((P_1,..,P_m)\in \A|P_i\leq x)$ is nondecreasing.
\end{definition}

\begin{theorem}
Suppose that the null p-values are i.i.d. and uniform on $[0,1]$ and independent of the non-null p-values. Then the FDR of BH is exactly $\pi_0\alpha$.

 If we only assume that the p-values satisfy the PRDS assumption, then the FDR of BH is at most $\pi_0\alpha$.
\end{theorem}

We will only prove the first part of the theorem here.
\begin{proof}
Multiple proofs can be found in \citt{wang2022elementary}. Here we give the proof based on martingale theory.
For $t\in[0,1]$ write 
$$R(t)= |\{1\leq i \leq m: P_i\leq t\}|,$$ 
$$V(t)= |\{i\in \N: P_i\leq t\}|.$$ 
A main idea in the proof will be  that $\mathbb{E}V(t)/t$ is equal to $|\N|$ for all $t\in (0,1]$
and is a martingale, so that we can use martingale theory.

Let $$t_{\alpha} = \sup\{t\in [0,1]: \frac{mt}{R(t)\vee 1}\leq \alpha\}$$ (not to be confused with earlier notation). Here $mt$ can be seen as a conservative estimate of $V(t)$, so that $t_{\alpha}$ has the interpretation of being the largest rejection threshold for which the estimated FDP is still below $\alpha$.

Note that the function $t\mapsto \frac{mt}{R(t)\vee 1}$ is piecewise continuous and has no jumps upwards. It follows that  $\frac{mt_{\alpha}}{R(t_{\alpha})\vee 1}$ is exactly $\alpha$, hence
\begin{equation} \label{eqR}
R(t_{\alpha})\vee 1 = \frac{mt_{\alpha}}{\alpha}.
\end{equation}

Consider any $1\leq i \leq m$ and suppose $P_i\leq t_{\alpha}$. On the interval $[P_i, t_\alpha]$ the function $t\mapsto \frac{mt}{R(t)\vee 1}$ takes its minimum at some $t\in\{P_1,...,P_m\}$.
Hence, there is a $1\leq j\leq m$ with $P_i\leq P_j$ and $mP_{j}/R(P_j)\leq \alpha$, i.e.,  $\Hy_i$ is rejected by BH\footnote{Here we use that $\Hy_i$ is rejected by BH if and only if   there is a $1\leq j\leq m$ with $P_i\leq P_j$ and $P_j\leq \frac{R(P_j)}{m}\alpha$.}.
If $P_i> t_{\alpha}$, there is no such $j$. Thus, each $\Hy_i$ is rejected by BH if and only if $P_i\leq t_{\alpha}$. 

For every $t\in[0,1]$, consider the $\sigma$-algebra generated by the variables $(\mathbbm{1}_{\{P_1\leq s\}},...,\mathbbm{1}_{\{P_m\leq s\}}:s\in [t,1])$:
$$\mathcal{F}_t = \sigma(\mathbbm{1}_{\{P_1\leq s\}},...,\mathbbm{1}_{\{P_m\leq s\}}:s\in [t,1]).$$
Conditioning on $\mathcal{F}_t$ means conditioning on the information which p-values have been oberved (and their values) before time $t$, if we let time run backwards.

We now show that  $t\mapsto V(t)/t$ is a backward martingale with respect to the filtration $(\mathcal{F}_t)_{t\in [0,1]}$.
For $s\leq t$,
$$ \mathbb{E} (V(s)/s|\mathcal{F}_t) = $$
$$s^{-1}\mathbb{E} (\sum_{i\in \N}\mathbbm{1}(P_i\leq s)|\mathcal{F}_t)=$$
$$ s^{-1}\sum_{i\in \N}\mathbb{E}(\mathbbm{1}(P_i\leq s)|\mathbbm{1}(P_i\leq t)). $$
We have $\pr(P_i\leq s|P_i>t)=0$ and $\pr(P_i\leq s|P_i\leq t)=s/t$.
Hence, the above equals
$$ s^{-1}\sum_{i\in \N} (s/t)  \mathbbm{1}(P_i\leq t)  =V(t)/t,$$
so $t\mapsto V(t)/t$ is a backward martingale.

Note that $t_{\alpha}$ is a stopping time with respect to the filtration $(\mathcal{F}_t)_{t\in[0,1]}$. (We start at timepoint 1 and let $t$ decrease until we observe that  $\frac{mt}{R(t)\vee 1}\leq \alpha$, which is known  conditional on $\mathcal{F}_t$.) The optional stopping theorem tells us that $\mathbb{E}(V(t_{\alpha})/t_{\alpha})= \mathbb{E}(V(1)/1)=|\N|$.

Hence, by \eqref{eqR},
$$\mathbb{E}\frac{V(t_{\alpha})}{R(t_{\alpha})\vee 1} =  \mathbb{E}\frac{V(t_{\alpha})\alpha}{mt_{\alpha}} =\frac{|\N|}{m}\alpha.$$
\end{proof}


The proof is by \citt{finner2009false} and is also provided in  \citt{wang2022elementary}. (\citealp{benjamini2001control} prove a slightly different result.) 
Note that the critical values of BH are identical to the critical values of Simes' global test. From the fact that BH provides weak FWER control under the PRDS assumption (see exercises) it follows that Simes' inequality \eqref{SimesProbEq} is valid under  the PRDS assumption. The reverse is not always true, i.e., Simes may hold while PRDS does not hold. When there are strong positive dependencies in the data, BH and Simes's global test tend to be conservative.

BH is implemented in the standard R function \emph{p.adjust()}. The mentioned  conservative variant that controls the FDR regardless of the dependence structure \citep{benjamini2001control} is also available in that function.
Applying BH (with $\alpha=0.05$) to the 279 p-values from  \S\ref{secappsales} results in     58 rejections (Holm rejected 36).  Applying it to the 31 p-values   corresponding to the numerical covariates results in 17 rejections (Holm rejected 10).

\subsection{Exercises}
\begin{enumerate}

\item Suppose $m=5$, $\alpha=0.05$ and the p-values corresponding to $\Hy_1,...,\Hy_5$ are $0.200$, $0.015$, $0.025$, $0.030$,  and $0.720$ respectively.

(a) Which hypotheses are rejected by Benjamini-Hochberg?

(b) Compute the  Benjamini-Hochberg-adjusted p-value for each hypothesis. See \S\ref{secadj} on how to compute adjusted p-values for step-up methods.

\item 
Consider $m$ hypotheses. Below, ``hypotheses'' always refers to these $m$ hypotheses and not to any of their intersections.

For each of the following general statements, indicate whether they are true. Explain your claims.
\begin{itemize}
\item
Statement 1: Whenever Benjamini-Hochberg rejects 0 hypotheses, Hommel also rejects 0 of the hypotheses.
\item  Statement 2: Whenever Hommel rejects 0 hypotheses,  Benjamini-Hochberg also rejects 0 hypotheses.
\end{itemize}

\item Consider a multiple testing method that controls the FDR at level $\alpha$, i.e., $FDR\leq\alpha$, where $\alpha\in(0,1)$. Prove that the method provides weak FWER control at level $\alpha$.

\item Consider a multiple testing method that controls the FDR, i.e., $FDR\leq\alpha$, where $\alpha\in(0,1)$. Let $\gamma\in (0,1)$.
Prove that $\pr(FDP\geq \gamma) \leq \alpha/ \gamma$, i.e., the method provides FDX control.

\item Suppose a method provides FDX control, i.e., $\pr(FDP\leq \gamma) \geq 1-\alpha$. 
Show that then $FDR\leq (1-\alpha)\cdot \gamma  + \alpha$.

\item Let $\E$ be the event that Benjamini-Hochberg rejects at least 1 hypothesis. 
Consider the following statement: ``If the p-values are valid and satisfy the PRDS assumption, BH controls the FDR conditionally on $\E$, i.e., $\mathbb{E}(FDP|\E)\leq \alpha$''. Is this statement generally true?

\end{enumerate}

\begin{sols}
\subsection{Solutions}
\begin{enumerate}

\item 

(a) The critical constants are $0.01, 0.02,...,0.05$. Note that $p_{(3)}\leq 0.03$, so at least 3 hypotheses are rejected, since Benjamini-Hochberg is a step-up method. The p-values $p_{(4)}$ and $p_{(5)}$ do not exceed their critical values. Hence, exactly three hypotheses are rejected. The hypotheses with the smallest p-values are $\Hy_2, \Hy_3, \Hy_4$, so these hypotheses are rejected.

(b) First we must compute $p_{(i)}\alpha c_i^{-1} =p_{(i)}\alpha ((\alpha i/m))^{-1} = p_{(i)}m/i$ for every $1\leq i \leq 5$. This gives the values $0.075,  0.062, 0.050, 0.250, 0.720$. Finally, we must enforce monotonicity. 
Since Benjamini-Hochberg is a step-up method, this means decreasing the values 0.075 and 0.062 to 0.050. Thus, the \emph{sorted} adjusted p-values are 
 $0.050,  0.050, 0.050, 0.250, 0.720$. Thus, for $\Hy_1,...,\Hy_5$, the respective adjusted p-values are $0.250, 0.050,  0.050, 0.050,  0.720$.

\item The critical values of the step-up method of Benjamini and Hochberg are identical to the critical values of Simes' global test. 
Consequently, Benjamini-Hochberg rejects something if and only if Simes' global test rejects.

Hommel is the CTP based on Simes local tests. A CTP will reject none of the  elementary hypotheses if the global hypothesis $\Hy_1\cap...\cap \Hy_m$ is not rejected by its local test. Hence,  Hommel will reject nothing if Simes' global test does not reject.
Thus, if Benjamini-Hochberg rejects nothing, then Hommel rejects nothing, so Statement 1 is true. 

If Simes' global test rejects something, then that does not imply that Hommel rejects one of the elementary hypotheses. This is because the CTP based on Simes local tests is not consonant. Thus, it may happen that Benjamini-Hochberg rejects something, but Hommel rejects none of the $m$ hypotheses. Thus, Statement 2 is false.

\item Suppose all hypotheses are true. We must show that $\pr(V>0)\leq \alpha$. Note that if all hypotheses are true, then $V>0\Leftrightarrow FDP=1$ and $V=0\Leftrightarrow FDP=0$.
Hence, 
 $\pr(V>0) =\pr(FDP=1) = \mathbb{E}(FDP)\leq\alpha$.

\item Since $FDP$ is a nonnegative random variable, we can use Markov's inequality, which says that $\mathbb{E}(FDP)\geq \gamma \pr(FDP\geq \gamma)$. Hence $$\alpha \geq FDR=\mathbb{E}(FDP)\geq \gamma\pr(FDP\geq \gamma) $$ and the result follows.

\item Note that by the law of total probability,
$\mathbb{E}(FDP)$ equals
$$\pr(FDP\leq \gamma)\mathbb{E}(FDP|FDP\leq \gamma) + \pr(FDP> \gamma)\mathbb{E}(FDP|FDP>\gamma) $$
\begin{equation} \label{eqproofFDRFDX}
\leq \pr(FDP\leq \gamma) \cdot \gamma  + \pr(FDP> \gamma) \cdot 1.
\end{equation}
Here we know that $\pr(FDP> \gamma)$ lies in the interval $[0,\alpha]$.
\eqref{eqproofFDRFDX} is maximized when $\pr(FDP\leq \gamma)$ is minimized and $\pr(FDP> \gamma)$ is maximized, i.e., when $\pr(FDP> \gamma)=\alpha$.
Thus, \eqref{eqproofFDRFDX} cannot exceed
$$ (1-\alpha)\cdot \gamma  + \alpha,$$
as was to be shown.

Alternative solution: Note that \eqref{eqproofFDRFDX} equals 
$$ (1-\pr(FDP> \gamma)) \cdot \gamma  + \pr(FDP> \gamma)$$
$$=\pr(FDP> \gamma)(1-\gamma)+\gamma $$
$$\leq \alpha (1-\gamma)+\gamma $$
$$= (1-\alpha)\cdot \gamma  + \alpha.$$

\item
No. Suppose for example that all hypotheses are true.  Then conditional on $\E$, the FDP is always 1. 
Thus, conditional on finding something, the error rate is not generally controlled. This is also true for many other  multiple testing methods, and also for tests of a single hypothesis.
\end{enumerate}

\end{sols}

\newpage

\section{FDR control with knockoffs} \label{secknockoffs}
\subsection{Motivation and setting}
Consider covariates $X_1,...,X_m$ and a response variable $Y$.
Often we want to know for every $1\leq j \leq m$ if  $Y$ depends on  $X_j$ \emph{conditional} on the other $m-1$ covariates. For every $1\leq j \leq m$, let $\Hy_j$ be the hypothesis that $Y$ is conditionally independent on $X_i$. If $\Hy_j$ is true, then that means that given the other covariates, $X_j$ does not provide any information on $Y$; in short, $Y\indep X_j|X_{-j}$.\footnote{This means in particular that conditional on $(X_j:j\in \N^c)$, $Y$ is independent of each individual $X_j$ with $j\in \N$.
However, it does not always mean that conditional on $(X_j:j\in \N^c)$, $Y$ is independent of the whole vector $(X_j:j\in \N)$. \citt[][pp.557-558]{candes2018panning} give an example to illustrate that.
However, if we exclude some pathological cases, these two notions are equivalent.}

To test $\Hy_j$, we could for instance assume a generalized linear model where $\mathbb{E}(Y)= g^{-1}(\beta_1X_1+...+\beta_mX_m)$, and test whether the coefficient of  $X_i$ is 0. Of course, we then make assumptions, since we assume a certain distributional shape for $Y$ and we assume that $g(\mathbb{E}(Y))$ depends linearly on $\beta_1X_1+...+\beta_mX_m$. We could make the model potentially more realistic by adding interactions, but still some substantial assumptions remain, and we should also be careful that the number of parameters does not become too large compared to the sample size.

The \emph{knockoffs} methodology from \citt{candes2018panning} provides an alternative method that gets rid of such assumptions, replacing it by assumptions on the joint dependence structure of $(X_1,...,X_m)$. The method tests the conditional-independence hypotheses $\Hy_1,...,\Hy_m$ in such a way that the FDR is controlled (for the \emph{knockoffs+} method, to be precise). The method puts no restrictions on $m$: it may be larger than $n$. The only assumption that the method makes, is that the joint dependence structure of $(X_1,...,X_m)$ is known. This is of course a big assumption, but in some settings this dependence structure is (approximately) known. \citt[][p.554-555]{candes2018panning} give these examples where the dependence structure is (approximately) known:
\begin{itemize}
\item Sometimes the covariate distribution is known because we control it, e.g. in some experiments.
\item Sometimes we very accurately know the covariate distribution, due to the availability of a large amount of additional unlabeled data.
\item A third example is when we have  pre-existing knowledge about aspects of the joint distribution of the covariates.
\end{itemize}

 Moreover, if the covariate distribution is unknown, it can be estimated based on the data at hand; there are some limited results on knockoffs methods for the situation where the dependence structure of $(X_1,...,X_m)$ is unknown a priori.

 The knockoffs method is a multiple testing method, but it can of course be used as a variable selection method.  We could therefore denote the set of hypotheses that the method selects by $\mathcal{S}$ as in \citt{candes2018panning}, but we will use the notation $\R$ in line with our usual notation. Thus, the (knockoffs+) method guarantees that $FDR=\mathbb{E}(|\N\cap \R|)/(|\R|\vee 1)\leq \alpha$.
 
 \subsection{Definition of knockoffs and test statistics}
 
 \begin{definition}\label{defknockoffs} \emph{(Model-X) knockoffs} for the variables $X=(X_1,...,X_p)$ are a new vector of random variables $\tilde{X}=(\tilde{X}_1,...,\tilde{X}_p)$ that is constructed to have the following properties:
 \begin{enumerate}
 \item\label{it:swapinv} for any subset $\A\subset \{1,...,m\}$, 
 $$  (X, \tilde{X})_{\text{swap}(\A)}   \,{\buildrel d \over =}\, (X, \tilde{X}), $$ where the swapping operation is defined below;
 \item\label{it:XindY} $\tilde{X}\indep Y | X$. 
 \end{enumerate}
 \end{definition}
 Note that property \ref{it:XindY} is automatically satisfied if $\tilde{X}$ is constructed based on $X$ without using $Y$.
 The vector $(X,\tilde{X})_{\text{swap}(\A)}$ is obtained from $(X,\tilde{X})$ by swapping the $X_i$ with $i\in \A$ with the corresponding $\tilde{X}_i$. For example, if $m=3$ and $\A=\{2,3\}$, we have
 $$  (X_1,X_2,X_3,\tilde{X}_1,\tilde{X}_2,\tilde{X}_3)_{\text{swap}(\A)} =  (X_1,\tilde{X}_2,\tilde{X}_3,\tilde{X}_1,X_2,X_3).$$ By property \ref{it:swapinv},  the distribution of $(X,\tilde{X})$ is invariant under these swapping operations. Thus the knockoff variables are interchangeable with the corresponding original covariates, except that the knockoff variables  are conditionally independent of $Y$. We will discuss the construction of knockoff variables in \S\ref{secconknock}.
 
 We now define the data in detail. We consider i.i.d. random vectors $(X_{i1},...,X_{im}, Y_i)\in \reals^m\times\reals$, which we collect in an $n\times m$-matrix  $\mathbf{X}$ and a column vector $y$ of length $n$. The knockoffs matrix $\tilde{\mathbf{X}}$ is constructed in such a way that for each $1\leq i \leq n$, 
 $(\tilde{X}_{i1},...,\tilde{X}_{im})$ is a knockoff of $(X_{i1},...,X_{im})$ as explained above.
We define $(\mathbf{X}, \tilde{\mathbf{X}})_{\text{swap}(\A)}$ to be the matrix for which the swapping operation defined above is applied to each row. 
We need the following lemma.
 
 \begin{lemma}\label{lemmaswap}
 If $\A\subseteq\N$, then 
 $$(\mathbf{X}, {\tilde{\mathbf{X}}})|y   \,{\buildrel d \over =}\, (\mathbf{X}, {\tilde{\mathbf{X}}})_{\text{swap}(\A)}|y$$
 \end{lemma}
 \begin{proof}
 
 We will prove the result for $|\A|=1$, so $\A=\{i\}$ for some $1\leq i\leq m$; then the general result follows from repeated swapping.
 Swapping does not change the distribution of $(\mathbf{X}, {\tilde{\mathbf{X}}})$, by Property \eqref{it:swapinv}. Further, the conditional distribution of $Y$ given $(\mathbf{X}, {\tilde{\mathbf{X}}})=(x,\tilde{x})$ is
 the same as the conditional distribution of $Y$ given $(\mathbf{X}, {\tilde{\mathbf{X}}})_{\text{swap}(\A)}=(x,\tilde{x})$, since $Y$ is conditionally independent on the covariate $X_i$ corresponding to $\A$. Thus, the joint distribution of $(\mathbf{X}, {\tilde{\mathbf{X}}}, Y)$ is unchanged by the swapping. Hence, given any $y$, the distribution of $(\mathbf{X}, {\tilde{\mathbf{X}}})|y$ is unchanged by the swapping.
 \end{proof}
 
 The knockoffs method takes certain statistics $W_j$, $1\leq j \leq m$, as input.
We now explain the construction of these $W_j$, which are a function ($w_j(\cdot)$, say) of the data and the knockoffs:
 $$W_j=w_j\big\{(\mathbf{X}, \tilde{\mathbf{X}}),y\big\}.$$
 In order for the knockoffs method to work, we need to define the test statistics $W_j$ in such a way that they have the \emph{sign flip property}, which means that swapping the $j$-th
 variable with its knockoff changes the sign of $W_j$:
 \begin{equation}\label{eq:signflip}
 w_j\big\{(\mathbf{X}, \tilde{\mathbf{X}})_{\text{swap}(\A)},y\big\} = 
 \begin{cases}
w_j\big\{(\mathbf{X}, \tilde{\mathbf{X}}),y\big\}, & j\not\in \A, \\
-w_j\big\{(\mathbf{X}, \tilde{\mathbf{X}}),y\big\},  & j\in \A,
\end{cases}
 \end{equation}
 for every $\A\subseteq \{1,...,m\}$.
 
Apart from requiring the sign flip property, how should we define these statistics $W_j$ to have good power? The key idea is that we define them in such a way that if $\Hy_j$ is false, then $w_j\big\{(\mathbf{X}, \tilde{\mathbf{X}}),y\big\}$ tends to be (much) larger than 0. The reason why we want that, is that the procedure will reject all hypotheses with test statistics exceeding some value. For example \citp[][p.561]{candes2018panning}, consider a linear regression model with $2m$ coefficients, corresponding to our $m$ covariates and their knockoffs --- where $Y_i$ depends on the $i$-th row of $(\mathbf{X}, \tilde{\mathbf{X}})$. (Note that do not require this model to be valid. We will only need the assumptions stated above to achieve FDR control.) For this model, consider the lasso,

$$ \min_{b\in \reals^{2m}} \frac{1}{2} \lVert y - (\mathbf{X}, \tilde{\mathbf{X}}) b\rVert^2_2 +\lambda  \lVert b \rVert_1 $$
and let $\hat{b}$ be the solution. Here $\lambda>0$ can be prespecified or based on the data, as long as we define it in such a way that permuting the columns of $(\mathbf{X}, \tilde{\mathbf{X}})$ does not change its value.
Then we can define the test statistics $W_j$ based on the lasso-estimated coefficients in the following way:
$$W_j  = |\hat{b}_j| - |\hat{b}_{j+m}|.$$
With this definition in place, the sign flip property \eqref{eq:signflip} is satisfied. If instead of the lasso we use ridge or elastic net to get $\hat{b}$, then likewise property \eqref{eq:signflip} is satisfied. These are common ways to define the test statistics, but other definitions satisfying property \eqref{eq:signflip} are also possible. For example, the test statistics may be based on a generalized linear model instead of a linear model. They may also be based on variable importance measures from some machine learning model.

The following lemma \citp[][p.561]{candes2018panning}  states the key property that we will use to prove FDR control. Note that a \emph{Rademacher variable} takes the values $-1$ and $1$, each with probability $0.5$.
\begin{lemma}\label{lemsignsradem}
Conditionally on $(|W_1|,...,|W_m|)$, the signs of the $W_j$ with $j\in \N$ are i.i.d. coin flips, i.e., Rademacher variables.
\end{lemma}

 \begin{proof}
 Write $W=(W_1,...,W_m)$.
Let $\varepsilon = (\varepsilon_1,\ldots,\varepsilon_m)$ be  independent random variables such that
$\varepsilon_j $ is Rademacher if $j \in \N$, and $\varepsilon_j = 1$ otherwise.
To prove the claim, it suffices to show that
\begin{equation} \label{eqepsW}
W \,{\buildrel d \over =}\, W',
\end{equation}
where $W':= \varepsilon \odot W$, with 
 $\odot$ denoting pointwise multiplication.  Indeed, $W'$ has the property stated in the lemma:
$(W_j':j\in \N)$ has i.i.d. Rademacher-distributed signs independent from $(W_1,...,W_m)$,  so independent from $(|W_1'|,...,|W_m'|)$.

To prove property \eqref{eqepsW}, let $\A = \{j: \varepsilon_1=-1\}$.
We have
$$W' = (\varepsilon_1w_1\big\{(\mathbf{X}, \tilde{\mathbf{X}}),y\big\},\ldots,\varepsilon_m w_m\big\{(\mathbf{X}, \tilde{\mathbf{X}}),y\big\} ). $$
By property  \eqref{eq:signflip}, this equals
$$ (w_1\big\{(\mathbf{X}, \tilde{\mathbf{X}})_{\text{swap}(\A)},y\big\},\ldots, w_m\big\{(\mathbf{X}, \tilde{\mathbf{X}})_{\text{swap}(\A)},y\big\} )  $$
By Lemma \ref{lemmaswap}, conditional on $y$ this random vector has the same distribution as 
$$ (w_1\big\{(\mathbf{X}, \tilde{\mathbf{X}}),y\big\},\ldots, w_m\big\{(\mathbf{X}, \tilde{\mathbf{X}}),y\big\} )  $$ and hence also marginally over $y$. But the above is $W$, which finishes the proof of property \eqref{eqepsW}.
\end{proof}
 
 \subsection{FDR control}
 Consider some threshold $t>0$. By Lemma \ref{lemsignsradem},
 $$|\{1\leq j \leq m: W_j\leq -t\}|\geq |\{j\in \N: W_j\leq-t\}| \,{\buildrel d \over =}\,  |\{j\in \N: W_j\geq t\}|.$$
 Note that the above equality in distribution holds because conditionally on the indices of ``extreme'' null statistics $\{j\in \N: |W_j|\geq t\}$, the signs of these extreme null statistics are i.i.d. Rademacher variables.
 Thus, if we reject all hypotheses $\Hy_j$ with $W_j\geq t$, then we see that $|\{1\leq j \leq m: W_j\leq -t\}|$ is an upward biased estimate of $|\{j\in \N: W_j\geq t\}|$, which is the number false discoveries.
 Thus, we can view
 $$\widehat{FDP}(t):= \frac{|\{1\leq j \leq m: W_j\leq-t\}|}{|\{1\leq j \leq m: W_j\geq t\}|}$$ as an estimate of  the FDP when we use threshold $t$.
 The idea of the knockoffs procedure is to choose the rejection threshold $t$ in a data-dependent way. Namely, we take $t$ to be the smallest threshold for which the estimated $FDP$ is still below $\alpha$.
 
 The following theorem says that if we use the estimate $\widehat{FDP}(t)$ defined above then we do not control the FDR but something which is similar (if $m$ is large). To ensure real FDR control, we have to  add $+1$ in the numerator of $\widehat{FDP}(t)$, which makes the estimate more conservative \citp[][p.562]{candes2018panning}. The proof only makes use of the coin-flip property from  Lemma \ref{lemsignsradem}  and is given in \citt{barber2015controlling}. Here we give a shortened version of the proof, sometimes referring to \citt{barber2015controlling} for details.
 
 \begin{theorem} \label{thmKOmethods}
Define the rejection  threshold $\tau > 0$ by 
$$
\tau
= \min\left\{
t\in \mathcal{W}^+ :
\frac{|\{j : W_j \le -t\}}{|\{j : W_j \ge t\}|} \le \alpha
\right\}
\qquad
\text{(knockoffs method),}
$$
where $\mathcal{W}^+:=\{|W_j|: |W_j|>0\}$ (or $\tau = \infty$ if the set above is empty). Here, any number divided  by $0$ is interpreted as $\infty$.
Then the procedure selecting the variables
$$
\R = \{ j : W_j \ge \tau \}
$$
controls  a ``modified FDR'', defined as
\[
\mathrm{mFDR}
= \mathbb{E}\!\left[
\frac{|\R \cap \N|}{|\R| + 1/\alpha}
\right]
\le \alpha.
\]

\medskip

Consider the slightly more conservative threshold
$$
\tau_+ = 
\min\left\{ t \in \mathcal{W}^+ :
\frac{1 + |\{j : W_j \le -t\}}{|\{j : W_j \ge t\}|} \le \alpha
\right\}.
\qquad
\text{(knockoffs$+$ method),}
$$
Setting $\R = \{ j : W_j \ge \tau_+ \}$, controls the usual FDR,
\[
\mathbb{E}\!\left[
\frac{|\R \cap\N|}{|\R| \vee 1}
\right]
\le \alpha.
\]

\medskip

These results are non-asymptotic and hold no matter the dependence between the response
and the covariates—in fact, they hold \emph{conditionally} on the response $y$.

\end{theorem}

\begin{proof}
\phantom{.}
\\
\noindent \emph{Overview of the proof.}
We will only give the proof for the second  method, the \emph{knockoffs+} method. For a proof of  the result on the  \emph{knockoffs} method, see \citt{barber2015controlling}.

To prove the validity of the knockoffs+ method, \citt{barber2015controlling} first consider a multiple testing procedure based on  p-values, and prove that this procedure controls the  FDR, 
$\mathbb{E}\big(V/(R\vee1) \big)   \leq \alpha.$
Then, they show that the knockoffs+ method can be formulated as a special case of that multiple testing method, thus proving that the knockoffs+ method controls the FDR. Here we will first prove validity of the general multiple testing method. This is not completely necessary, but the advantage is that we closely stick to the reasoning of \citt{barber2015controlling}.
\\
\\
\emph{Step 1: Definition of the general multiple testing method.}
The general multiple testing method based on  p-values is called \emph{Selective Seqstep+} \citp[][p.2081]{barber2015controlling} and is defined as  follows. Consider $m$ hypotheses with corresponding p-values $p_1,..,p_m$ and assume the null p-values are valid (Assumption \ref{assvalidp}). We do \emph{not} sort these p-values based on their values.
Fix $c\in(0,1)$\footnote{In fact it is sufficient to consider the case $c=0.5$, for proving FDR control of knockoffs+. Then $(1-c)/c$ is simply $1$.} and define
$$\hat{k}_1 = \max\big\{ 1\leq k \leq m: \frac{1+|\{j\leq k: p_j>c\}|}{|\{j\leq k: p_j\leq c\}|\vee 1} \leq \frac{1-c}{c}\alpha    \big\},$$
with the convention that $\hat{k}_1 =0$ if the set is empty. The Selective Seqstep+ procedure rejects all hypotheses with indices $j$ satisfying $j \leq \hat{k}_1 $ \emph{and} $p_j\leq c$.
\\
\\
\emph{Step 2: Proving FDR control for the  general multiple testing method.}
We now prove that the Selective Seqstep+ method controls the FDR. This is stated in \citt[][Theorem 3]{barber2015controlling} and proved in their supplement.
Note that $R = |\{j\leq\hat{k}_1: P_j\leq c  \}|$ and hence $V = |\{\text{null } j\leq\hat{k}_1: P_j\leq c  \}|$. We have 
$$\mathbb{E}(\frac{V}{R\vee 1}) = \mathbb{E}(\frac{V}{R\vee 1}\cdot \mathbbm{1}_{\{\hat{k}_1>0\}})=$$
$$ \mathbb{E}\Bigg[  \frac{|\{\text{null } j\leq\hat{k}_1: P_j\leq c  \}| }{1+ |\{j\leq\hat{k}_1: P_j> c  \}|}                  \cdot     \Bigg(\frac{1+ |\{j\leq\hat{k}_1: P_j> c  \}|}{|\{j\leq\hat{k}_1: P_j\leq c  \}|\vee 1}  \mathbbm{1}_{\{\hat{k}_1>0\}}\Bigg) \Bigg] \leq $$
\begin{equation}\label{eqsupermar}
\mathbb{E}\Bigg[  \frac{|\{\text{null } j\leq\hat{k}_1: P_j\leq c  \}| }{1+ |\{j\leq\hat{k}_1: P_j> c  \}|}                    \Bigg]  \cdot  \frac{1-c}{c}\cdot  \alpha 
\end{equation}
by definition of  $\hat{k}_1$.

That the above is at most $\alpha$, is a consequence of Lemma 1 of the supplement of \citt{barber2015controlling}. The idea is as follows. Consider the following stochastic process indexed by $k$:
$$  M(k) =   \frac{|\{\text{null } j\leq k: P_j\leq c  \}| }{1+ |\{\text{null }j\leq k: P_j> c  \}|}, $$
with $k=m, m-1,...$ running backwards. It turns out that this is a backward supermartingale with respect to a certain  filtration, with respect to which $\hat{k}_1$ is a stopping time. 
Note that at the first timepoint $k=m$, we have 
$$
\mathbb{E}M(m) =  \frac{|\{\text{null } j: P_j\leq c  \}| }{1+ |\{\text{null }j: P_j> c  \}|}.
$$
As shown at the end of the proof of Lemma 1 in the supplement of \citt{barber2015controlling}, this expected value is bounded from above by $c/(1-c)$. $\hat{k}_1$ is a stopping time with respect to the filtration, so that the optional stopping theorem for super-martingales  implies that $\mathbb{E}M(\hat{k}_1)\leq c/(1-c)$. Note that clearly, the expectation in \eqref{eqsupermar} is at most  $\mathbb{E}M(\hat{k}_1)$, so  the quantity \eqref{eqsupermar} is at most $(c/(1-c))\cdot ((1-c)/(c))\alpha =\alpha$.
\\
\\
\emph{Step 3: Noting that knockoffs+ is a special case of the  multiple testing method.}
Index the hypotheses  in such  a way that $|W_1|\geq ...\geq |W_m|>0$ (ignore statistics that are 0 and let $m$ be the number of nonzero statistics). Note that even with this new indexation, the statistics corresponding to the true null hypotheses have i.i.d. Rademacher signs, by Lemma \ref{lemsignsradem}.

To be able to use the Selective Seqstep+ method, we must define p-values $P_1,...,P_m$, which we do in the following way:
 \begin{equation}\label{defbinpv}
P_j = 
 \begin{cases}
1 & W_j<0, \\
0.5,  &W_j>0.
\end{cases}
 \end{equation}
 Note that the $P_j$ with $j\in \N$ are then valid and independent p-values, since the signs of the $W_j$ with $j\in \N$ are i.i.d. Rademacher variables.

The knockoffs+ method computes
$$
\tau_+
= \min\left\{
t \in \mathcal{W}^+ :
\frac{1 + |\{j : W_j \le -t\}|}{|\{j : W_j \ge t\}|} \le \alpha
\right\}
$$
and rejects the hypotheses $\R = \{ j : W_j \ge \tau_+ \}$.

This leads to the same rejections as defining 
$$
\hat{k}_1
= \max\left\{
1\leq k \leq m :
\frac{1 + |\{j : P_j=1 \}|}{|\{j : P_j=0.5\}|} \le \alpha
\right\}
$$
and rejecting the hypotheses $\R = \{ j\leq \hat{k}_1  : P_j=0.5\}$.
But that is exactly the set that Selective Seqstep+ rejects, if we take $c=0.5$.
Thus, the knockoffs+ method has FDR at most $\alpha$.

\end{proof}

 \subsection{Constructing knockoffs} \label{secconknock}
 How can we construct knockoff variables, i.e. variables $\tilde{X}$ such that the properties in Definiton \ref{defknockoffs} hold?
 
 \subsubsection{The case that $X$ is multivariate normal} \label{secXmvn}
First consider the case  that $X$ has a known multivariate normal distribution $\N(0,\Sigm)$. Then an example of a  joint distribution for $(X, \tilde{X})$ that satisfies property \eqref{it:swapinv} is 
$\N(0,\mathbf{G})$,
where 
\begin{equation} \label{refG}
\mathbf{G}  = 
\begin{pmatrix}
\Sigm & \Sigm -\text{diag}(s) \\
\Sigm -\text{diag}(s) & \Sigm,
\end{pmatrix}
\end{equation}
  where $\text{diag}(s)$ is the diagonal matrix with diagonal $s=(s_1,...,s_m)'$ and $s\succeq 0$ (meaning all elements of $s$ are nonnegative) is chosen in such a way that the matrix is positive semidefinite. Note that by defining $\mathbf{G}$ in this way, we ensure that for every $1\leq j \leq m$, $X_j$ and $\tilde{X}_j$ have the same covariance with all the other variables\footnote{For example, for $i\neq j$, $\Cov(X_j,X_i)=\Sigma_{ji}=(\Sigma-\text{diag}(s))_{ji} = \Cov(\tilde{X}_j,X_i)$.}, while  their covariance with each other is smaller than their variance. That is essential, because if we would have $\Cov(X_j,\tilde{X}_j) = \Var(X_j)=\Var(\tilde{X}_j)$, then $X_j$ and $\tilde{X}_j$ would be perfectly correlated and  $X_j$ and $\tilde{X}_j$ would be (a.s.) identical. That would mean that the procedure usually has no power for rejecting $\Hy_j$.
  
  To sample $\tilde{X}$ such that $(X, \tilde{X})\sim \N(0,\mathbf{G})$, we can use classical formulas for the conditional distribution of a subvector of a multivariate normal distribution\footnote{
  The well known result that we use is the following. If a random vector $Z$ has a $N(\mu_z,\mathbf{G}_z)$ distribution with 
  $$
  Z=\begin{bmatrix}
{Z_1}\\
{Z_2} 
\end{bmatrix},\quad \mu_z=\begin{bmatrix}
{\mu_1}\\
{\mu_2} 
\end{bmatrix},\quad \mathbf{G}_z=
\begin{bmatrix}
\mathbf{G}_{11}& \mathbf{G}_{12} \\
\mathbf{G}_{21} &\mathbf{G}_{22}
\end{bmatrix}
,$$ 
then the conditional distribution of $Z_2$ given $Z_1=z_1$ is
$$ N\big(\mu_2+\mathbf{G}_{21}\mathbf{G}_{11}^{-1}(Z_1-\mu_1), \mathbf{G}_{22}- \mathbf{G}_{21}\mathbf{G}_{11}^{-1}\mathbf{G}_{12}\big).$$
}. 
These tell us that
  $$\tilde{X}|X \,{\buildrel d \over =}\, \N(\mu,\mathbf{V}),$$
  where 
  \begin{align*}
   \mu =& \Sigm^{-1}(\Sigm-\text{diag}(s)) X \\
         =& X-\Sigm^{-1}\text{diag}(s)X,\\
         \mathbf{V} =&  \Sigm -  (\Sigm-\text{diag}(s))\Sigm^{-1} (\Sigm-\text{diag}(s)) \\
         =&  \Sigm -  \Sigm +           2\text{diag}(s)   -\text{diag}(s)\Sigm^{-1}\text{diag}(s) \\
         =& 2\text{diag}(s) - \text{diag}(s)\Sigm^{-1}\text{diag}(s).
  \end{align*}
  
  It turns out that $\mathbf{G}$ is positive semidefinite if and only if $s\succeq 0$ and $2\Sigm\preceq \text{diag}(s)$ \citp[][pp. 2057, 2062]{barber2015controlling}. ($\mathbf{A} \succeq \mathbf{B}$ is defined to mean that $\mathbf{A} - \mathbf{B}$ is positive semidefinite.) 
  How can we choose $s$ such that this is satisfied and the method has good power? In order to have good power, we want $\Cov(X_j,\tilde{X}_j)$ to be small, which means we want $s_j$ to be large. Thus we wish to find large $s_j$'s in such a way that the constraints $s\succeq 0$ and $2\Sigm\preceq \text{diag}(s)$ are satisfied. Especially when $m$ is large, this quicky becomes computationally demanding: in practice we will need to choose a tradeoff between statistical power and computational efficiency. Details are in \citt[][pp.564-565]{candes2018panning}.
  
\subsubsection{The general case}
Above we assumed that $X$ has a known multivariate normal distribution $\N(0,\Sigm)$. Now suppose we do not have such knowledge. One possible situation is that  we  know that $X$ is multivariate normal with mean 0, but we do not know $\Sigm$. Then we can estimate $\Sigm$ and proceed as above, which leads to approximate FDR control.
Another possible situation is that we do not even know whether $X$ is multivariate normal. Then a proposal by \citt{candes2018panning} is that we do not require $(X,\tilde{X})_{\text{swap}(\A)}$ and $(X,\tilde{X})$ to have exactly the same distribution, but require them to have the same first two moments. (If we have to estimate the covariance matrix $\Sigm$, then this moment matching can of course only be done approximately.)

Requiring that $(X,\tilde{X})_{\text{swap}(\A)}$ and $(X,\tilde{X})$ have the same covariance matrix, means requiring that $$\Cov(X,\tilde{X})=\mathbf{G},$$
with $\mathbf{G}$ as in \eqref{refG}. Hence, \citt[][pp.564-565]{candes2018panning} propose to use the same approaches for finding $s$ as in the case where $X$ is multivariate normal (see \S\ref{secXmvn}). 

Note that when $\Sigm$ is unknown and $m$ is large, estimating $\Sigm$ is no small task. This means that the knockoffs(+) method can be quite inaccurate, in the sense that the FDR may substantially exceed $\alpha$. However, it is worth adding that when $m$ is large and $\Sigm$ is unknown, we might be able to use some pre-existing knowledge on $\Sigm$, for example that entries far from the diagonal are 0. And again, sometimes we do know $\Sigm$ quite accurately a priori. For example, \citt{candes2018panning} study a genetic dataset with $m=377749$ predictor variables (single nucleotide polymorphisms, SNPs). There they use existing knowledge on the correlation structure of such variables. For example, SNPs on different chromosomes are roughly independent of each other (depending on the type of population the sample is from). Further, scientists have some knowledge of the correlations between the SNPs within a chromosome.
Alternative methods for sampling approximate knockoffs have also been proposed \citp{romano2020deep,blain2024knockoffs}. 

To conclude, knockoff methods do not always work well (i.e., they do not always lead to accurate FDR control), but that is mainly because  conditional independence testing with many variables is often nearly impossible. Knockoff methods can potentially be a good approach in this context, although research is ongoing. The discussed methods are implemented in the R package \emph{knockoff}, available on CRAN \citp{knockoff}.

\subsection{Exercises}
\begin{enumerate}
\item
Suppose $m=6$, $\alpha=0.2$ and $(W_1,...,W_6)=(0.7, 1.1 ,-0.8 ,0.5 ,1.2,0.6)$. 

(a) Which hypotheses are rejected by the knockoffs method as defined in Theorem \ref{thmKOmethods}? Which hypotheses are rejected by the knockoffs+ method?

(b) For each of these two methods, briefly discuss what it guarantees  in this situation.
\end{enumerate}

\begin{sols}
\subsection{Solutions}
\begin{enumerate}
\item
(a) We first consider the  knockoffs method. For $t=0.5$, we have 
$$\frac{|\{j: W_j\leq -t\}|}{|\{j: W_j\geq t\}|} = 1/5 \leq \alpha.$$ For smaller $t$ the fraction becomes $\infty$. Hence, $\tau=0.5$, so all hypotheses except the third one are rejected.

We now consider the  knockoffs+ method. For $t\in [0.5,0.8]$, we have 
$$\frac{1+ |\{j: W_j\leq -t\}|}{|\{j: W_j\geq t\}|} \geq 2/5 > \alpha.$$ For smaller $t$ the fraction becomes $\infty$. 
For  $t>0.8$, we have 
$$\frac{1+ |\{j: W_j\leq -t\}|}{|\{j: W_j\geq t\}|} = \frac{1}{|\{j: W_j\geq t\}|} \geq 1/2 > \alpha.$$ Hence, $\tau+=\infty$ and the knockoffs+ method rejects no hypotheses.

(b) The knockoffs method guarantees that 

$$\mathbb{E}(\frac{|\R\cap\N|}{|\R|+1/\alpha}) =\mathbb{E}(\frac{|\R\cap\N|}{|\R|+5}) \leq 0.2.$$
This does not have any interpretation that is similar to FDR control or one of the other error rates that we have seen. Since $m$ is small and $\alpha$ is not so small, the method does not even provide something close to FDR control. The knockoffs+ method is stricter, but guarantees that the $FDR$, i.e. $\mathbb{E}(\frac{|\R\cap\N|}{|\R|})$, is at most 0.2.
\end{enumerate}
\end{sols}

\setlength{\bibsep}{3pt plus 0.3ex}  
\def\bibfont{\small}  

\bibliographystyle{biblstyle}
\bibliography{bibliography}

\newpage

\appendix
\section{Measure theoretic probability} 
An extensive treatment of measure theoretic probability  and martingale theory is contained in the freely available  lecture notes of \citt{spreij2023measure}. In \S\ref{apprv} and \S\ref{appmart} we discuss a few concepts from this theory, because theory from  \S\ref{appmart} is used in a few proofs in \S\ref{secFDR} and \S\ref{secknockoffs}.
In \S\ref{apprv} we provide a formal definition of \emph{random variables}, which is required in \S\ref{appmart}.
\S\ref{appmart} contains theory on \emph{martingales} and \emph{stopping times}.

\subsection{Measure theory and random variables} \label{apprv}
In probability theory (and in particular martingale theory), the concept of a measurable random variable can be important. To define this concept, we must first introduce the notion of a $\sigma$-algebra \citp{spreij2023measure}.

Let $\mathcal{S}$ be a non-empty set. As usual, let $2^\mathcal{S}$ denote the set of all subsets of $\mathcal{S}$. A collection $\Sigma\subseteq 2^{\mathcal{S}}$ is called a $\sigma$-algebra on $\mathcal{S}$ if
\begin{enumerate}
\item $\mathcal{S}\in \Sigma,$
\item $\mathcal{E}\in \Sigma \Rightarrow \mathcal{E}^c\in \Sigma,$
\item $\mathcal{E}, \mathcal{F}\in \Sigma  \Rightarrow \mathcal{E}\cup \mathcal{F}\in \Sigma$,
\item if $\mathcal{E}_n\in \Sigma$ ($n=1,2,...$) then $\bigcup_{i=1}^{\infty} \mathcal{E}_n \in \Sigma$.
\end{enumerate}
Note that $\emptyset\in \Sigma$ and if $\mathcal{E}, \mathcal{F}\in \Sigma$, then  $\mathcal{E} \cap \mathcal{F}\in \Sigma$, since $\mathcal{E} \cap \mathcal{F} =  (\mathcal{E}^c \cup \mathcal{F}^c)^c.$

If $\Sigma$ is a $\sigma$-algebra on $\mathcal{S}$, then $(\mathcal{S},\Sigma)$ is called a \emph{measurable space} and the elements of $\Sigma$  are called measurable sets. 

If $\mathcal{D}$ is a collection of subsets of $\mathcal{S}$, then by $\sigma(\mathcal{D})$  we denote the smallest $\sigma$-algebra containing $\mathcal{D}$.  If $\Sigma = \sigma(\mathcal{D})$, we say that $\mathcal{D}$ \emph{generates} $\Sigma$. 

If $\mathcal{S}=\reals$ is the set of real numbers, then the $\sigma$-algebra on $\reals$ that is usually considered is the \emph{Borel} $\sigma$-algebra, which is the $\sigma$-algebra  generated by all \emph{open} subsets of $\reals$. This is the same as the $\sigma$-algebra  generated by all open intervals $(a,b)$. This is also identical to the $\sigma$-algebra  generated by all closed intervals $[a,b]$ for example. This is because an open interval is a union of closed intervals. For example, $(0,5)$ is the union of all intervals of the form  $[0+1/n,5-1/n]$, where $n\in\{1,2,...\}$.

A map $\mu: \Sigma \rightarrow [0,\infty]$ is called a \emph{measure} if it is countable addititive, which means that
\begin{itemize}
\item $\mu(\emptyset)=0,$
\item $\mu(\cup_i\mathcal{E}_i) =  \sum_i \mu(\mathcal{E}_i)$ for every countable collection $\mathcal{E}_1, \mathcal{E}_2,...$ of disjoint elements of $\Sigma$.
\end{itemize}
Further, if $\mu(\mathcal{S})=1$, then $\mu$ is called a probability measure. Such a measure is often denoted by $\mathbb{P}$ and $(\mathcal{S}, \Sigma, \mathbb{P})$ is then called a probability space. It is then also common to use the symbol $\Omega$ for $\mathcal{S}$. We say that a subset of $\mathcal{S}$ is \emph{measurable} if it is contained in $\Sigma$.

We will also need the notion of a \emph{measurable function}. Let $\mathcal{B}$ denote the Borel $\sigma$-algebra. A function
$X: \mathcal{S} \rightarrow \reals$ is called ($\Sigma$-)measurable if for every $\mathcal{A}\in \mathcal{B}$, $X^{-1}(\mathcal{A})\in \Sigma$. $X$ being ($\Sigma$-)measurable means that if we know for each element of $\Sigma$ whether it contains $\omega$, then we know $X(\omega)$.

A \emph{random variable} is a measurable function from a probability space $\Omega$ to  $\reals$ \citp{spreij2023measure}. More precisely, it is a function $X: \Omega \rightarrow \reals$ such that for every $\mathcal{A}\in \mathcal{B}$, $X^{-1}( \mathcal{A})\in \Sigma$. We could also interpret the term ``random variable'' more broadly as a measurable map from a probability space to some other measurable space, e.g. $\reals^n$ with the Borel $\sigma$-algebra on  $\reals^n$, or a space of functions, etc..  In case of $\reals^n$, we often speak of a \emph{random vector}. In general, the space where $X$ takes its values is called the \emph{sample space}. 

The interpretation of the function $X$ is as follows:  ``the probability that the random variable $X$ takes a value in $\mathcal{A}$'' is $\pr(X^{-1}( \mathcal{A}))$. This quantity is well-defined since we know that $X^{-1}( \mathcal{A}) \in \Sigma$, since $X$ is measurable. Note that when we write $\pr(X\in \mathcal{A})$, this is interpreted as $\pr(\omega\in \Omega: X(\omega)\in \mathcal{A})$, which is the probability that $X$ is in $\mathcal{A}$.

The $\sigma$-\emph{algebra generated by} a real-valued variable $X$ is the set 
$$\sigma(X):=\sigma(X^{-1}(\mathcal{A}): \mathcal{A}\in \mathcal{B}),$$
i.e., the $\sigma$-algebra generated by $\{X^{-1}(\mathcal{A}): \mathcal{A}\in \mathcal{B}\}$. $\sigma(X)$ is the smallest $\sigma$-algebra with respect to which $X$ is measurable. Conditioning on $\sigma(X)$ means conditioning on knowing $X$. Indeed, if we know which elements of $\sigma(X)$  contain $\omega$, then we know $X(\omega)$, and vice versa.

The set
$$\{X\in \mathcal{A}\}=\{\omega\in \Omega: X(\omega)\in \mathcal{A}\} = \{X^{-1}( \mathcal{A})\}$$ is called the \emph{event} that $X$ is in $A$. Formally, an event is a  measurable subset of the  measurable  space $(\Omega, \Sigma)$. We say that an event $\A$ holds \emph{almost surely} (``a.s.'') if $\pr(\A)=1,$ i.e., there may be some $\omega$'s outside $\A$, but their measure $\pr(\Omega\setminus\A)$ is $0$. 

In practice we often do not explicitly define random variables or vectors as functions. Instead, we often just define a probability measure on the sample space --- e.g. by defining a joint PDF on the sample space $\reals^n$, which implies a probability measure on  $\reals^n$. 

The probability measure $\pr$ of the probability space $(\Omega, \Sigma, \pr)$ and the variable $X$ imply a probability measure $\mu_X$ on the outcome space (with $\sigma$-algebra $\mathcal{B}$, say), if for every $\mathcal{A}\in \mathcal{B}$ we define $\mu_X(\mathcal{A})=\pr(X^{-1}(\mathcal{A}))$. Note that $\mu_X(\mathcal{A})$ then has the interpretation of being the probability that the random variable $X$ takes a value in $\mathcal{A}$. 

The expected value of a random variable $X$ is defined as the (Lebesgue) integral of $X(\cdot)$ with respect to $\pr$,
$$\mathbb{E}(X) =\int_{\Omega} X(\omega) d\pr(\omega),$$
see \citt{spreij2023measure} for details on Lebesgue integration, which is a generalization of Riemann integration. 
In many basic texts on probability theory, the expectation of a continuous random variable is defined in terms of the PDF of the variable, and  the expectation of a discrete random variable is defined as a sum. Such definitions have limitations. For example, a random variable might be neither fully continuous nor fully discrete. 
The above definition is more general. 



\subsection{Martingales and  optional stopping} \label{appmart}
A stochastic process (in $\reals$) is a sequence of random variables (measurable functions from a measurable space $(\Omega, \mathcal{F})$ to $(\reals, \mathcal{B})$) indexed by a set $T$ that indicates time. That set could e.g. be $T=\{0,1,2,...\}$ (a discrete set of time points), $T=[0,\infty)$ or $T=[0,1]$ (a continuous set of time points). We denote the process by $(X_t)_{t\in T}$. Just like a random variable, the process $(X_t)_{t\in T}$ is a mapping from $\Omega$. A realization of the process corresponds to some $\omega\in \Omega$, so it can be written as $(X_t)_{t\in T}(\omega)$.

In order to define martingales and to give the optional stopping theorem, we need the concept of a \emph{filtration} $\mathbb{F}$, which is a sequence $(\mathcal{F}_t)_{t\in T}$ of sub-$\sigma$-algebras on $\mathcal{F}$, such that for all $s,t\in T$ with $s\leq t$ we have $\mathcal{F}_s\subseteq \mathcal{F}_t\subseteq \mathcal{F}$.

A stochastic process $(X_t)_{t\in T}$ is called \emph{adapted} to $\mathbb{F}$ if for every $t\in T$, $X_t$ is measurable with respect to $\mathcal{F}_t$. Note that 
$X_t$ being measurable with respect to $\mathcal{F}_t$ means that   if we know for every element of $\mathcal{F}_t$ whether it contains $\omega$, then we know $X(\omega)$. 



Conditioning on a $\sigma$-algebra has the interpretation of knowing for every element of the $\sigma$-algebra whether it contains $\omega$.
The interpretation of a filtration $\mathbb{F}$ is often that $\mathcal{F}_t$ represents the information available by time $t$, in the following sense: if we have observed the process $(X_t)_{t\in T}(\omega)$ up to time $s<\max(T)$, then we have observed part of the process but not the whole process, so we only have limited information about which $\omega\in \Omega$ our process corresponds to. However, we do know for each element of $\mathcal{F}_s$ whether it contains $\omega$. 



\begin{definition} \label{defmart}
A stochastic process $(X_t)_{t\in T}$ is called a \emph{martingale} (or $\mathbb{F}$-martingale) if it is adapted to the filtration $\mathbb{F}$ and for $s\leq t$ 
\begin{equation} \label{defmgale}
\mathbb{E}(X_t|\mathcal{F}_s) =  X_s
\end{equation}
almost surely, where we assume Lebesgue integrability where required (see \citt{spreij2023measure} for details).
\end{definition}
The  interpretation is as follows:  conditioning on $\mathcal{F}_s$ means that we know in which subsets of $\mathcal{F}_s$ our realization $\omega$ is contained, which means that we know $(X_t)_{t\leq s}$. The martingale property says that given our information at time $s$, the expected future values of the process are exactly $X_s$. A \emph{supermartingale} is defined in the same way, except that the equality in \eqref{defmgale} becomes ``$\leq$''.

For example, suppose we are playing a gambling game with multiple rounds, and in every round we win a euro or lose a euro, each with probability 0.5. For simplicity, suppose that it is possible to have negative capital, so that we can always keep playing. Then if we are currently in timepoint $s$ and have 40 euros, then conditional on everything that has happened until now, our expected capital at any future timepoint is also 40 euros. 

\begin{remark}
We could consider  a filtration with $\F_t =\F_s$ for all $s,t\in T$, i.e., we already have all information from the start. Then all $X_{t}$ are measurable with respect to $\F_{\min T}$ (assuming the minimum exists), which means that conditional on $\F_{\min T}$,  $X_{t}$ is known for all $t\in T$.
That means that we already know the whole process at timepoint $\min T$. That is only possible if after observing $X_{\min T}$, the process is not random at all. Since a martingale has constant conditional expectation, that must mean that the process is a horizontal line.
Indeed, if we consider Definiton \ref{defmart}, we see that we require $X_t =X_s$ (almost surely), since $\mathbb{E}(X_t|\mathcal{F}_s) = \mathbb{E}(X_t|\mathcal{F}_t) =X_t$. 
\end{remark}

Often we are interested in evaluating or stopping a process at some timepoint that we pick based on everything that has happened before it. For example, we might want to stop the gambling game as soon as our capital exceeds 1000. This relates to the concept of \emph{stopping times}.

\begin{definition}
A map $\tau: \Omega \rightarrow T\cup\{\infty\}$ is called a \emph{stopping time} if for all $t\in T$ it holds that $\{\tau \leq t\}\in \mathcal{F}_t$.  (Recall that $\{\tau\leq t\}$ means  $\{\omega \in\Omega:\tau(\omega)\leq t\}$.)
\end{definition}

Thus, $\tau: \Omega \rightarrow T\cup\{\infty\}$ is a stopping time if the decision to stop at time $t$ depends on the information available at time $t$.

Then we have the following result, which is known a version of the \emph{optional stopping theorem}.

\begin{theorem}[A version of the optional stopping theorem] \label{thmopstop}
Suppose $T$ is bounded and so is $\tau$, and $T$ has a minimum $t_0$. Suppose $(X_t)_{t\in T}$ is an  $\mathbb{F}$-martingale. Then
$$\mathbb{E}(X_{\tau})=\mathbb{E}(X_{t_0}).$$
\end{theorem}

Thus, in the mentioned gambling game (where we can have negative capital), if we start with 100 euros, and we stop once we have won 3 times in a row, then our expected capital at the stopping time is 100 euros. 
(Note that if we stop when we have 200 euros, then Theorem \ref{thmopstop} does not apply since $\tau$ is not bounded then. Note that if we would use that strategy, with probability 1, we would at some point reach 200 euros. However, it turns out that the expected time until we reach the value 200 is infinite. This strategy is not useful in practice, because  we do not get expected gains after finite time.)

For super-martingales the above optional stopping theorem also holds, except that the equality becomes ``$\leq$''. In general, the theorem is called the \emph{optional} stopping theorem because the stopping decision is based on observed information, i.e., you do not need to stop at some fixed time.

\end{document}